\pdfoutput=1  
\documentclass[a4paper,onecolumn, 11pt]{quantumarticle}

\pdfoutput=1
\usepackage[utf8]{inputenc}
\usepackage[english]{babel}
\usepackage[T1]{fontenc}
\usepackage{amsmath}
\usepackage{hyperref}

\usepackage{tikz}
\usepackage{lipsum}
\usepackage{mathtools,amsthm,amssymb,setspace,bbm}
\usepackage{mathrsfs}
\usepackage[utf8]{inputenc}
\usepackage{indentfirst}
\usetikzlibrary{shapes.geometric, arrows}

\definecolor{paleblue}{rgb}{0.69, 0.93, 0.93}
\definecolor{powderblue}{rgb}{0.69, 0.88, 0.9}
\tikzstyle{startstop} = [rectangle, rounded corners, minimum width=3cm, minimum height=1cm, text centered, draw=black, fill=red!30]
\tikzstyle{process} = [rectangle, rounded corners, minimum width=2cm, minimum height=1cm, text centered, draw=black, fill=powderblue]
\tikzstyle{arrow} = [thick,->,>=stealth]
\tikzstyle{decision} = [diamond, minimum width=3cm, minimum height=1cm, text centered, draw=black, fill=green!30]

\usepackage{tabularx}
\usepackage{multirow}
\usepackage{quantikz}
\usetikzlibrary{quantikz2}
\usepackage{tikz-cd}
\usepackage{amsthm}
\usepackage{subcaption}
\usepackage{float}
\usepackage{comment}
\usepackage{graphicx}
\usepackage{dcolumn}
\usepackage{bm}

\newtheorem{theorem}{Theorem}[section]
\newtheorem{corollary}{Corollary}[theorem]
\newtheorem{proposition}{Proposition}[theorem]
\newtheorem{lemma}[theorem]{Lemma}
\theoremstyle{definition}
\newtheorem{definition}[theorem]{Definition}

\usepackage{hyperref}
\hypersetup{colorlinks=true,linkcolor=blue,citecolor=magenta,urlcolor=blue}

\begin{document}

\title{Fault-tolerant Preparation of Distant Logical Bell Pair - with application in the magic square game}

\author{Andy Zeyi Liu}
\email{andy.liu@yale.edu}
\affiliation{Institute for Quantum Computing, University of Waterloo, Ontario, Canada, N2L 3G1}
\affiliation{Department of Combinatorics and Optimization, University of Waterloo}
\affiliation{Perimeter Institute of Theoretical Physics, Ontario, Canada, N2L 2Y5}
\affiliation{Yale Quantum Institute, Yale University, New Haven, Connecticut 06511, USA}
\affiliation{Department of Applied Physics, Yale University, New Haven, Connecticut 06511, USA}

\author{Debbie Leung}
\affiliation{Institute for Quantum Computing, University of Waterloo, Ontario, Canada, N2L 3G1}
\affiliation{Department of Combinatorics and Optimization, University of Waterloo}
\affiliation{Perimeter Institute of Theoretical Physics, Ontario, Canada, N2L 2Y5}

\maketitle

\begin{abstract}



Measures of quantum nonlocal properties are traditionally defined assuming perfect and unlimited local computational ability of each remote party.  In real experiments, each computational primitive will be imperfect. Fault-tolerant techniques, developed to enable simulation of arbitrarily accurate quantum computation using noisy primitives, applies only to problems with classical input and output. and need not preserve optimized measures of nonlocality.  

In this paper, we examine the impact of very low noise in measures of quantum nonlocality in the context of nonlocal games.  We observe that, as a nonlocal game has classical inputs and outputs, fault-tolerant techniques can approximate the game value, yet, even arbitrarily small imperfection can disproportionately affect the amount of entanglement required for approximating the game value. 

Focusing on the \textit{fault-tolerant magic square game}, we seek to optimize the tradeoff between noisy entanglement consumption and the deficit in the game value.
We introduce a novel approach leveraging an interface circuit and entanglement purification protocol (EPP) to translate states between physical and logical qubits and purify noisy logical ebits. This method significantly reduces the number of initial ebits needed compared to conventional strategies. Our analytical and numerical analyses, particularly for the $[[7^k,1,3^k]]$ concatenated Steane code, demonstrate substantial(actually, exponential) ebit savings and higher noise threshold. Analytical lower bounds for local noise threshold of $4.70\times10^{-4}$ and initial ebit infidelity threshold of $18.3\%$ are obtained.

Our framework is adaptable to various quantum error-correcting codes \space (QECCs) and experimental platforms. Our protocol will not only enhance understanding of fault-tolerant nonlocal games, but also inspire further exploration of interfacing between different QECCs, promoting the development of modular quantum architectures and advancing quantum internet.
\end{abstract}

\newpage
\tableofcontents
\section{Introduction}
\textit{Quantum Entanglement}, one of the most profound phenomena in quantum mechanics, allows particles to exhibit correlations that cannot be explained by classical physics, even when separated by large distances. This nonlocal behavior has significant implications for various quantum technologies, such as quantum computation and communication. The entangled particles act as a single system, showcasing nonlocal effects that challenge classical notions of locality and providing a foundational resource for tasks that are otherwise impossible or inefficient with classical systems.

Nonlocal games are one of the scenarios where quantum entanglement demonstrates its power. A nonlocal game involves interaction between a referee and spatially isolated players, where the referee assigns questions to each player based on a publicly known distribution. Players respond to their questions, aiming to satisfy specific criteria and win the game. While players cannot communicate during the game, they can use quantum entanglement prior to the game to increase their chances of winning, thus highlighting the power of quantum resources in achieving correlations that surpass classical limits. Nonlocal games not only offer insights into the fundamental nature of quantum mechanics but also have practical applications in quantum cryptography and device-independent quantum key distribution~\cite{Buhrman_2010}. 

Experimental demonstrations of nonlocal games, such as loophole-free Bell tests, have validated these theoretical predictions under various conditions~\cite{1981PhRvL..47..460A, PhysRevLett.115.250401, PhysRevLett.119.010402, loophole-free}. In recent years, special types of nonlocal games like quantum pseudotelepathy~\cite{Brassard_2005} have been explored, where players can achieve success probability 1 with quantum strategies. They have also been verified to show quantum advantage~\cite{Jiamin}. For the Magic Square Game, in the presence of hardware errors, the average winning probability across all assignments of indices is $93.8\%$, surpassing the classical value but still suboptimal. This prompts the question: Can we achieve a winning probability of 1 when each gate in the circuit has a constant error rate?  Additionally, in both communication and nonlocality, what is the minimal amount of resources necessary to achieve this? One key resource is quantum entanglement, or more specifically \textit{EPR pairs(ebits)}, that enables quantum advantage in nonlocal games. And this quantity will be a central focus of this work. 

One promising solution to the first question is \textit{Fault-tolerant Quantum Computation (FTQC)}, which permits arbitrarily low logical error rates in the presence of physical noise, provided that the error rates of all operations remain below a threshold value~\cite{Knill_1998,aharonov1996fault}. An established method for error suppression is code concatenation~\cite{knill1996concatenated, Rahn_2002} and it was later proved that by concatenating the $[[7,1,3]]$ Steane code~\cite{1996Steane}, \textit{FTQC} is possible with a theoretically proven lower bound on the threshold~\cite{AGP05}. Recent advancements have revitalized interest in this approach, showcasing the feasibility of time-efficient and constant-space-overhead \textit{FTQC}\cite{Yamasaki_2024}. 
Although these code concatenation constructions work well under quantum computation settings, complexities arise when we hope to apply them to quantum communication or nonlocal game scenarios, which are underexplored areas. It was shown that in the noisy setting, capacities of quantum channels can be asymptotically approached with a fault-tolerant construction based on Steane code concatenation\cite{christandl2022faulttolerant, Belzig_2024}. However, an exact threshold value is unknown and the question of resources is not addressed. Another popular family of codes for FTQC is topological codes, or more specifically, surface codes. By encoding logical qubits in 2D arrays of physical qubits and using local measurements, they effectively manage both bit-flip and phase-flip errors. This code is highly favored due to its high error threshold and compatibility with existing quantum hardware technologies~\cite{acharya2022suppressing, Krinner_2022, Bluvstein_2022, ryananderson2024highfidelity}. 

In this work, we restrict to the context of \textit{fault-tolerant magic square game} and provide a partial solution to the following question \\

\textit{Given some fixed noise strength in local devices, if Alice and Bob hope to play the magic square game with a success probability arbitrarily close to 1, what is the minimal number of ebits required?} \\

One essential component to establish the results on fault-tolerant quantum communication is an interface circuit that translates between physical qubits and logical qubits. A similar idea was experimentally tested with 9-qubit Shor's code\cite{Luo_2021}. This concept can potentially save ebits as it moves entangling operations from the logical space to the physical space. In the $[[7^k,1,3^k]]$ case, it circumvents the need for transversal entangling gates. However the interface itself is not fault-tolerant, so we invoke the idea of entanglement purification~\cite{CB} to filter out bad ebits. Combining these ideas, we use a modified interface and propose an alternative scheme for preparing a logical ebit and subsequently playing the magic square game. 
\begin{figure}[H] 
    \centering 
    \includegraphics[width=0.6\linewidth]{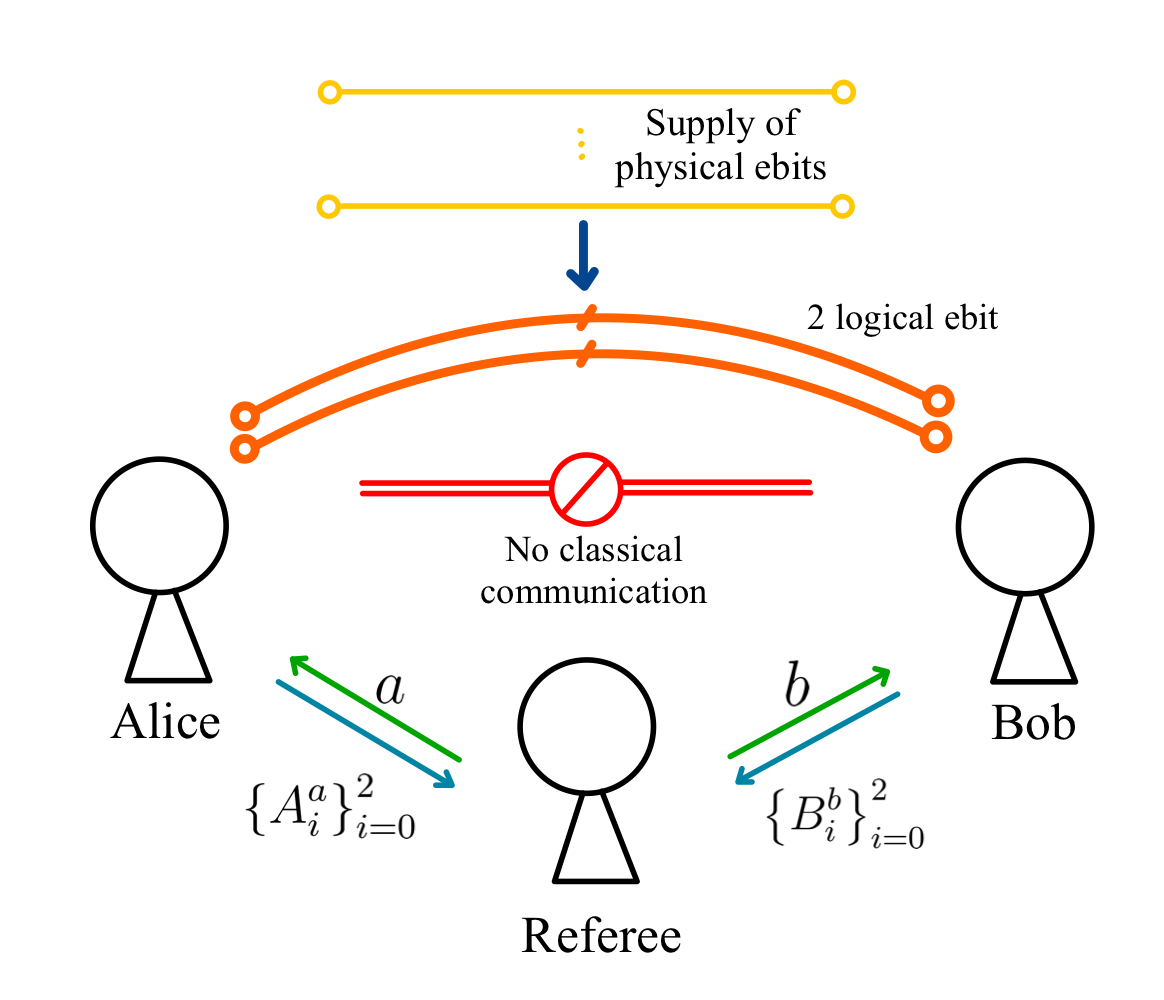} 
    \caption{The procedure of the \textit{Fault-tolerant Magic Square Game}. For the \textit{Magic Square Game}, Alice and Bob are randomly assigned a row and a column index $a,b\in\{0,1,2\}$. They then each reply with three answers $[A_0^a,A_1^a,A_2^a]$ and $[B_0^b,B_1^b,B_2^b]$ with $A_i^a,B_j^b\in\{\pm1\}$ $\forall i,j$. The winning condition is that $\prod_{i=0}^2A_i^a=+1$ and $\prod_{i=0}^2B_i^b=-1$ and $A_b^a=B_a^b$. No communication is allowed during the game. However, if Alice and Bob share two EPR pairs prior to the game and perform corresponding measurements, they are able to win the game with probability 1 (in the error-free scenario). In the FT case, Alice and Bob are supplied with noisy physical ebits. In order to implement the same strategy, they use these to create 2 logical ebits and then perform logical measurements.} 
    \label{intro_illus} 
\end{figure} 
We've compared our scheme with \textit{Direct Encoding}, that is, Alice and Bob prepare $\ket{\overline{+}}$ and $\ket{\overline{0}}$ respectively and they use transversal CNOT to create a logical ebit. We used numerics-assisted methods for  $k\geq1$ and performed exact Monte-Carlo simulation for $k=1$. Let $\Delta$ denote the failure probability of the magic square game. Given that $0<\Delta\ll1$, physical error rate $\epsilon=2.25\times10^{-4}$ and initial infidelity of ebits being $10\%$, our proposal shows substantial savings on the initial ebits consumed. A more explicit comparison can be seen in Figure~\ref{FTMSG_result_intro}.
\begin{figure}[H]
    \centering
    \includegraphics[width=0.6\linewidth]{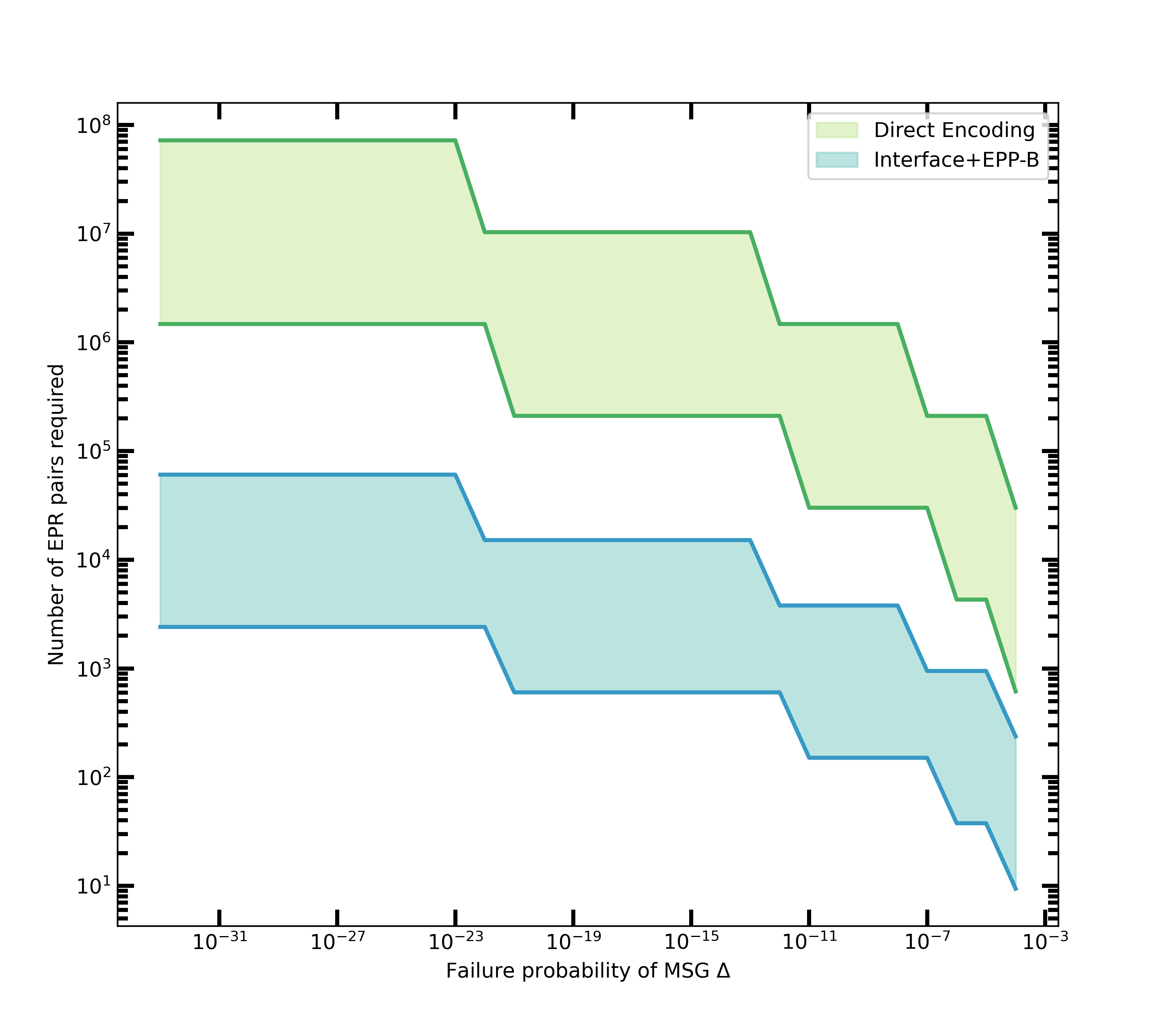}
    \caption{The number of raw ebits required to achieve magic square game failure probability $\Delta$ versus $\Delta$. Due to the complexity of simulation when the concatenation level is high, analytical bounds are obtained instead. The green and blue bands above represent regions bounded by the upper and lower bounds for the two methods respectively.}
    \label{FTMSG_result_intro}
\end{figure}
In fact, there is another advantage of our scheme, the above $\epsilon$ is actually a derived lower bound on the threshold for \textit{Direct Encoding}. Notably, \textit{Interface+EPP} has the potential to tolerate higher local noise levels. Our analysis establishes a lower bound $4.70\times10^{-4}$. Additionally, this method is capable of handling significantly noisier initial ebits, with a lower bound on the infidelity threshold being $18.3\%$.

The rest of this paper is organized as follows. In Section~\ref{chap:QECC} we review the basics of entanglement purification protocol(EPP), fault-tolerance and the magic square game. Familiar readers may skip this section. In Section~\ref{initial_EPP} we begin by discussing the preparation of high-fidelity ebits from noisy ebits. In Section~\ref{exRecthres}, we provide analytical bounds on the logical error rate of exRecs used in our work and thus derive a lower bound on the threshold. Using these results, Section \ref{Dir_Enc} presents bounds on the logical error rate of a logical ebit prepared via \textit{Direct Encoding}. Section \ref{IEPP} provides an overview of the novel \textit{Interface+EPP} method, and describes our modified interface for the Steane code. We then obtain bounds on the logical error rate of the \textit{Interface+EPP} schemes in Section \ref{logicalerrEPP}. In Section \ref{FTMSG}, we combine the results and obtain our main result as above. Finally, we present a full numerical simulation comparing the methods for encoding level-1. Lastly, we summarize our results and outline future directions in Sec.\ref{conclusion}.

\section{Preliminaries}\label{chap:QECC}
\subsection{Notations}
Throughout this work, we use $k$ to denote the level of concatenation. Quantities with superscript $a^{(k)}$ represent the corresponding $a$ encoded to level-$k$. The term \textit{ebit} refers to an EPR pair $|\Psi\rangle=(|00\rangle+|11\rangle)/\sqrt{2}$, while $\overline{\text{ebit}}$ denotes a logical ebit. Hence $\overline{\text{ebit}}^{(k)}$ would denote a logical ebit encoded to level-$k$. We use $\overline{\text{EPP}}$ to denote the local operation component of a logical EPP procedure.

\subsection{Entanglement Purification Protocol (EPP)}
Entanglement purification protocol (EPP) aims at generating high-fidelity entangled quantum states from a larger set of low-fidelity one, through local operations and classical communication (LOCC). It was first introduced by Bennett et al.~\cite{CB}. In this paper we will mainly use two-way communication protocols, where classical communication is bidirectional. It starts with two parties, Alice and Bob initially sharing imperfect ebits. Through local operations, measurements, and classical communication, they refine these pairs into a single pair with higher fidelity. In this work, we shall assume the initial noisy EPR pairs are in the Werner state form, so in the Bell basis they can be written as,
\begin{align*}
    \rho=F|\Phi^{+}\rangle\langle\Phi^+|&+\left(\frac{1-F}{3}\right)\bigg(|\Psi^+\rangle\langle\Psi^+|+|\Psi^-\rangle\langle\Psi^-|+|\Phi^-\rangle\langle\Phi^-|\bigg)
\end{align*}
where $|\Phi^-\rangle=\frac{1}{\sqrt{2}}(|00\rangle-|11\rangle)$, $|\Psi^+\rangle=\frac{1}{\sqrt{2}}(|01\rangle+|10\rangle)$, $|\Psi^-\rangle=\frac{1}{\sqrt{2}}(|01\rangle-|10\rangle)$.
It was shown that any two-qubit state can be converted into this form via the appropriate `twirl' operation. Due to the fact $(U\otimes I)|\Phi\rangle=(I\otimes U^T)|\Phi\rangle$, an error on the ebit is one of $XI,YI,ZI$. The simplest example of EPP is shown in Figure \ref{simple_EPP},
\begin{figure}[H]
    \centering
    \includegraphics[width=0.25\linewidth]{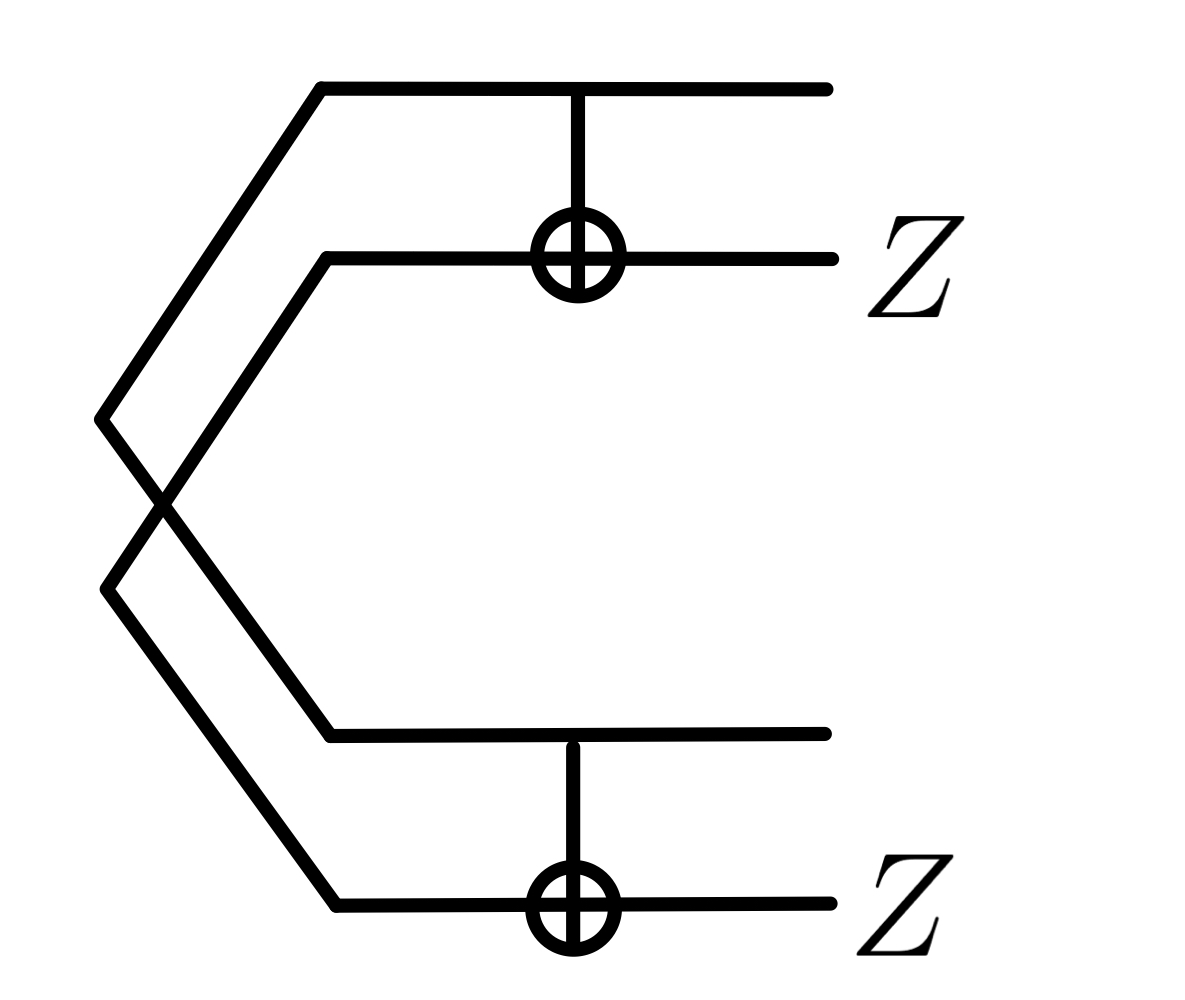}
    \caption{A simple EPP. The upper half and lower half are owned by Alice and Bob respectively. Upon measuring the second qubit in the $Z$-basis, if the outcomes agree it marks a purification.}
    \label{simple_EPP}
\end{figure}
It can be alternatively thought of as an error-detecting circuit. Assuming we start with two perfect ebits, an $XI$ error occurs in the first ebit, then on Alice's side, it will be propagated by the CNOT to the second pair, and the $Z$-basis measurement, the second ebit will now be in the state $|\Psi^+\rangle$. Thus if both parties measure in the $Z$-basis, the outcomes will disagree, leading to rejection. Following the paradigm, alternative EPP schemes are explored, a more detailed discussion can be found in the work of Krastanov et al.~\cite{Krastanov_2019}. In the presence of circuit-level noise, the results ebits will contain errors after EPP. For the specific EPP circuits we use in this work, we will provide detailed descriptions in subsequent sections as they become pertinent to our analysis.

\subsection{Quantum Fault-tolerance}\label{FTQC1}
To achieve a reliable simulation of a quantum circuit, we need to construct logical components that are `good'. This means that the error correction process must not introduce more errors than it can handle. Indeed, the core principle of \textit{fault-tolerance} is to control the propagation of errors within the circuit. We primarily adhere to the framework described in Aliferis et al.~\cite{AGP05}, introducing only the concepts relevant to our context and making appropriate adjustments as needed. 

For the basics, we refer \textit{location} to all the single components in the circuit. Broadly speaking, locations can be categorized into the following types: preparation location, gate location, wait location and measurement location. Classical computation is assumed to be error-free. To derive our main results, it suffices to list the following locations:
\begin{enumerate}
    \item preparation of $|0\rangle$
    \item preparation of $|+\rangle$
    \item measurement of $X$
    \item measurement of $Z$
    \item (local) CNOT gate 
    \item nonlocal resource
    \item identity gate
\end{enumerate}
The location of type $i$ is denoted $Loc_i$. Note that in $Loc_5$ and $Loc_6$ we have the notion of `local'. Let us recall EPP introduced above, the initial ebits shared between Alice and Bob are considered `non-local' while any operation Alice(Bob) does on her(his) side is `local'. Thus $Loc_6$ can be a nonlocal CNOT (CNOT across Alice and Bob) or an ebit, which will be specified depending on the context. 
It's worth distinguishing \textit{error} from \textit{fault}. An \textit{error} occurs when one physical qubit is corrupted while a \textit{fault} occurs when a location goes bad. For example, in the case of Pauli noise, when a CNOT gate is faulty, and $XX$ follows the perfect CNOT. In this case, we say 2 errors occur but there is only one fault. Now we can define the error model.
\begin{definition}[Independent Pauli Noise]
    In an independent Pauli noise model, each location in the circuit is assumed to fail independently. A faulty location $Loc_i$ can be seen as a perfect $Loc_i$ followed by Pauli noise with strength $\epsilon_i$, except the case of measurement, where is seen as an error followed by perfect measurement.
\end{definition}
This noise model falls into the broad category of \textit{circuit-level} noise. In later analysis, we will treat the $Loc_5$ error rate $\epsilon$ as the baseline, and all the other locations are proportional to $\epsilon$ with $\epsilon_i=\sigma_i\epsilon,\space\forall i\in\{1,2,3,4,6,7\}$. By default, in this paper we follow conventions in Knill~\cite{Knill_2005} and have $\sigma_i=\frac{4}{15}$ for $i\in\{1,2,3,4\}$  and $\sigma_7=\frac{4}{5}$. $\epsilon_6$ is the error rate of the nonlocal resource, in the case of an ebit, it is the infidelity. The analytical and numerical methods in this paper will be applicable to alternative configurations of $\sigma_i$'s, accommodating implementation across various platforms and devices. Next, we set to perform simulation for the quantum circuit. We first define what we mean by a \textit{gadget}.
\begin{definition}[Gadget]
   A gadget for a quantum operation is a circuit that executes the corresponding operation on the logical state when no fault occurs. An error-correction (EC) gadget functions by measuring the stabilizer generators to extract the syndrome and subsequently applying Pauli corrections. When operating without faults, it can correct up to $t=\lfloor\frac{d-1}{2}\rfloor$ errors for a distance $d$ QECC.
\end{definition}
Given different types of locations and their corresponding gadgets in the logical space, we may define what we mean by \textit{fault-tolerance}.
\begin{definition}[Fault-tolerance criteria~\cite{AGP05}]\label{FTcriteria}
Suppose we encode qubits in a QECC with distance $d$. Let $r$ denote the number of input errors into a gadget, $s$ denote the number of faults in the gadget ($r=0$ for preparation gadget), and $t=\lfloor\frac{d-1}{2}\rfloor$. Then a gadget is fault-tolerant if
\begin{enumerate}
    \item Preparation gadget\\
        When $s\leq t$, a preparation gadget with $s$ faults produces a logical state with at most $t$ errors.
    \item Measurement gadget \\
        When $r+s\leq t$, the outcome of a measurement gadget agrees with an ideal measurement. In particular, for a non-destructive measurement, we also need the number of errors on the output data block to be at most $t$.
    \item Gate gadget \\
    When $r+s\leq t$, an output state from a gate gadget has at most $t$ errors in each output block.
    \item Error-correction(EC) gadget
        When $r+s\leq t$, the output state deviates from the codespace by at most $t$ errors. In particular, if $s=0$, EC gadget takes any input in the codespace with $r\leq t$ to an output with no errors.
\end{enumerate}
\end{definition}
Suppose we have gadgets that individually meet the specified criteria and we aim to use them to simulate an ideal circuit. We observe that unless the physical error rates are arbitrarily small, the errors will accumulate when the gadgets are combined as the ideal circuit gets larger, thereby violating the criteria. Hence, to ensure fault-tolerance even when gadgets are put together, one solution is to perform error correction following every logical operation. This leads to the following definition.
\begin{definition}[\textit{Rec and exRec}]
An \textit{extended rectangle}(\textit{exRec}) of a FT simulation circuit is defined as
\begin{enumerate}
    \item Preparation-exRec\\
    A preparation-exRec consists of a preparation gadget and the EC gadget after it.
    \item Measurement-exRec\\
    A measurement-exRec consists of a measurement gadget and the preceding EC gadget.
    \item Gate-exRec\\
    A gate-exRec consists of a gate gadget and the EC gadgets immediately before and after it. If omitting the preceding EC, we call it a \textit{gate-rectangle}(\textit{Rec}).
\end{enumerate}
\end{definition}
In Figure \ref{exRec} we show the exRecs for various locations. $|\overline{+}\rangle$-exRec and $\overline{Z}$-mmt-exRecs will be identical to their counterparts in the figure except for changing $|\overline{0}\rangle$ to $|\overline{+}\rangle$ and $\overline{X}$ to $\overline{Z}$. For the CNOT-exRec we illustrate the difference between \textit{Rec} and \textit{exRec}, the part within the dashed line is a \textit{Rec}.
\begin{figure}[H]
    \centering
    \begin{subfigure}[b]{0.45\textwidth}
    \centering
    \begin{tikzpicture}
    \node[scale=1.2]
    {
    \begin{quantikz}
    \lstick{\ket{\overline{0}}}&\qwbundle{}&\gate{\text{EC}}&\qw
    \end{quantikz}
    }; 
    \end{tikzpicture}
    \end{subfigure}
    \hfill
    \begin{subfigure}[b]{0.45\textwidth}
    \centering
    \begin{tikzpicture}
    \node[scale=1.2]
    {
    \begin{quantikz}
    &\qwbundle{}&\gate{\text{EC}}&\meter{\overline{X}}
    \end{quantikz}
    };  
    \end{tikzpicture}
    \end{subfigure}
    \vfill
    \begin{subfigure}[b]{0.45\textwidth}
    \centering
    \begin{tikzpicture}
    \node[scale=1.2]
    {
    \begin{quantikz}
    &\qwbundle{}&\gate{\text{EC}}&\gate{\overline{H}}&\gate{\text{EC}}&\qw
    \end{quantikz}
    }; 
    \end{tikzpicture}
    \end{subfigure}
    \hfill
    \begin{subfigure}[b]{0.45\textwidth}
    \centering
    \begin{tikzpicture}
    \node[scale=1]
    {
    \begin{quantikz}
    &\qwbundle{}&\gate{\text{EC}}&\gate[2]{\overline{\text{CNOT}}}\gategroup[2,steps=2,style={dashed,rounded
corners, inner
xsep=2pt},background,label style={label
position=below,anchor=north,yshift=-0.2cm}]{\text{Rec}}&\gate{\text{EC}} & \qw \\
    &\qwbundle{}&\gate{\text{EC}}&&\gate{\text{EC}}& \qw
    \end{quantikz}
    };  
    \end{tikzpicture}
    \end{subfigure}
    \caption{Examples of different exRec. Upper-left is the $\ket{\overline{0}}$-exRec. Upper-right is the $\overline{X}$-measurement-exRec. Lower-left is a $\overline{H}$-exRec. Other one-qubit gates are constructed analogously. Lower-right is an example of a two-qubit exRec, the $\overline{\text{CNOT}}$-exRec. The part in the dashline, without the preceding ECs, is a \textit{Rec}.} 
    \label{exRec}
\end{figure}
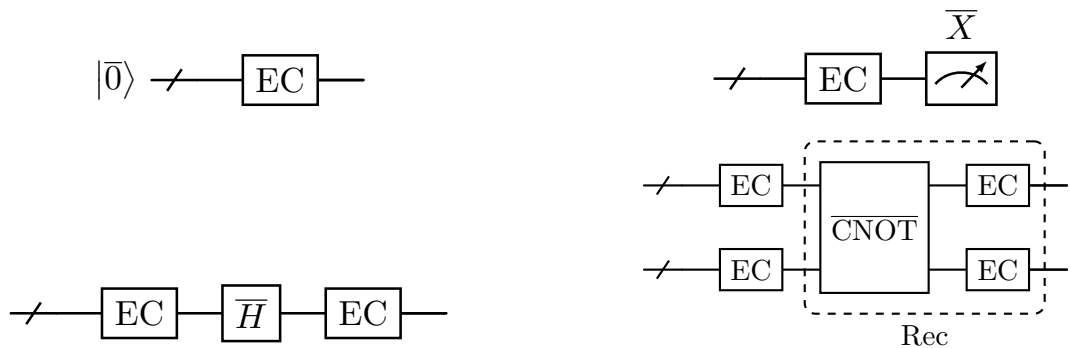
Now, to fault-tolerantly simulate a large circuit, we need to suppress the error rate of each component. To achieve this, we resort to code concatenation. Based on the exRec construction before, we can concatenate QECC and have the following definitions. 
\begin{definition}[\textit{Recursive simulation}]
    Let $C_0$ be the ideal circuit and let $C_l$ be the level-$l$ fault-tolerant simulation of $C_0$. $C_l$ is constructed in the following recursive way: At level-1, each location is simulated by its corresponding Rec. At level-$k$, $k\geq2$, each location is replaced by the $(k-1)$-Rec.
\end{definition}
Next, adopted from Aliferis et al.~\cite{AGP05}, we give definitions on when the exRecs `fail', i.e. have logical errors.
\begin{definition}[\textit{Malignant set}]\label{malig_set}
    A set of locations in an exRec is \textit{benign} if the Rec contained in the exRec has no logical error at the output when any choice of errors occur in these locations. If the set of locations is not benign, they form a malignant set. 
\end{definition}
The reason why we address the correctness of the Rec will be evident in the following definition. 
\begin{definition}[\textit{Badness of exRecs}]
For $k=1$ an exRec is \textit{bad} if it contains it contains faults that form a malignant set; if it is not bad it is good. For $k>1$, a $k$-exRec is bad if it contains independent $(k-1)$-exRecs at a malignant set of locations. Two bad exRecs are independent if they are non-overlapping or if they overlap and the earlier $k$-exRec is still bad when the shared $k$-EC is removed. If it is not bad it is good. For a simulation circuit consisting of exRecs, we call the whole circuit \textit{bad} if the output contains a logical error.
\label{badness}
\end{definition}
The idea in this `independence' definition is that if there are in total 3 faults in two consecutive 1-exRecs in which one fault is in the overlapping EC, then the overlapping pair of bad 1-exRec is really no worse than a single bad 1-exRec. In Aliferis et al.~\cite{AGP05}, it was shown that if all exRecs in the fault-tolerant circuit are good, then this would give a correct simulation of the ideal circuit, i.e. the probability distribution of the final measurement outcomes are the same, thereby validating the fault-tolerance. Based on these definitions we have the following theorem.
\begin{theorem}(Theorem 5~\cite{AGP05})\label{CorrectSim}
    Let $C$ be a circuit that begins with state preparation and ends with measurement. Let $\Tilde{C}$ be the fault-tolerant simulation of $C$ under independent Pauli noise. Suppose that for a particular fault path $\gamma$, the exRecs in $\gamma$ form a benign set. Then, the output distribution of $\Tilde{C}$, is identical to the output distribution of \( C \) with ideal gates.
\end{theorem}
These constructions also lead to the well-known threshold theorem~\cite{aharonov1996fault,AGP05}. It is important to clarify what we mean by \textit{threshold} here, as we will later derive the threshold value specific to the constructions in this work. The following definition is adopted from Svore et al.~\cite{svore2006flowmap}.
\begin{definition}[FT threshold]\label{FTthreshold}
    Consider any ideal quantum circuit $C$. For a series of FT schemes consisting of a family of QECC $[[n(L), k(L), d(L)]]$ parametrized by $L$ and their corresponding FT gadget sets, let $C_L$ denote the $L$th simulation circuit. For a noise model $\mathcal{N}$ of strength $\epsilon$, the $L$-th simulation circuit with noise is denoted $C_{L,\mathcal{N}}$. The failure probability $p_L$ of $C_{L,\mathcal{N}}$ is defined as 
    \begin{equation*}
        p_L=\sup_{\rho}T(C(\rho),C_{L,\mathcal{N}}(\rho))
    \end{equation*}
    where $T(\rho,\sigma)=\frac{1}{2}\mathbf{Tr}|\rho-\sigma|$ is the trace distance. The \textit{fault-tolerance threshold} is defined as $\epsilon_{\text{th}}(C)$ such that when $\epsilon<\epsilon_{\text{th}}(C)$,
    \begin{equation*}
        \lim_{L \rightarrow\infty}p_L=0
    \end{equation*}
\end{definition}

\subsection{Magic Square Game}\label{MSG}
The Mermin-Peres Magic Square game is one of the simplest non-local games in which a referee randomly assigns a row and a column index $a, b\in\{0,1,2\}$ to two parties Alice and Bob. They then each replied with three answers $[A_0^a, A_1^a, A_2^a]$ and $[B_0^b, B_1^b, B_2^b]$. The winning condition is $\Pi_{i=0}^2 A_i^a=+1$ for all $a$ and $\Pi_{i=0}^2 B_i^b=-1$ for all $b$ and most importantly we require $A_b^a=B_a^b$, i.e. the overlapping element of their answers must be the same. The best classical strategy succeeds with probability 8/9, but there is a perfect quantum strategy, assuming all operations are faultless. The strategy is as follows: $A$ and $B$ share the state $\frac{1}{2}(|00\rangle+|11\rangle)_{A_1B_1}(|00\rangle+|11\rangle)_{A_2B_2}$ before the game starts. They then measure in the set of basis given in Table \ref{msgmmt} corresponding to the indices they are assigned. For instance, if the referee assigns Alice 1 and Bob 2, when the game starts Alice will measure her qubits in $ZI,IX,ZX$-basis sequentially and Bob will measure in $XZ,ZX,YY$-basis. Following this procedure, their measurement outcomes will satisfy the criteria. 
\begingroup
\setlength{\tabcolsep}{10pt}
\renewcommand{\arraystretch}{1.35}
\begin{table}[H]
\centering
\renewcommand{\arraystretch}{1.8} 
\centering
\begin{tabular}{cccccc}
                       & 0                         & 1                          & 2                        & \multicolumn{1}{l}{} &  \\ \cline{2-4}
\multicolumn{1}{l|}{0} & \multicolumn{1}{c|}{$I\otimes Z$}  & \multicolumn{1}{c|}{$Z\otimes I$}   & \multicolumn{1}{c|}{$Z\otimes Z$} &            &  \\ \cline{2-4}
\multicolumn{1}{l|}{1} & \multicolumn{1}{c|}{$X\otimes I$}  & \multicolumn{1}{c|}{$I\otimes X$}   & \multicolumn{1}{c|}{$X\otimes X$} &                  &  \\ \cline{2-4}
\multicolumn{1}{l|}{2} & \multicolumn{1}{c|}{$X\otimes Z$} & \multicolumn{1}{c|}{$Z\otimes X$} & \multicolumn{1}{c|}{$Y\otimes Y$} & 
&  \\ \cline{2-4}
\end{tabular}
\caption{Measurement basis for Alice and Bob when they are assigned different indices. Alice holds the row index while Bob holds the column index.}
\end{table}
\label{msgmmt}
\endgroup

\section{Preparation of High-fidelity Physical EPRs}\label{initial_EPP} 
In practice, an ebit shared between two distant parties will have significantly higher infidelity compared to the local physical error rates, e.g. ebits generated with photons, and the attenuation will worsen in proportion to distance. In this paper, we will assume the infidelity of the initial EPR pair to be $3q\approx10\%$, where $q$ is the probability of one of the other Bell states. For higher initial infidelity we can perform some rounds of EPP to reduce to the desired error rate. The physical error rate of local operations is assumed to be below the FT threshold (which is usually much lower than the initial infidelity) to ensure fault-tolerance. Due to this large discrepancy and by the threshold theorem, it's necessary to first bring down the infidelity to a level comparable to local error rates before we prepare $\overline{\text{ebit}}^{(k)}$ so that we can treat them as `the same type of error'. To accomplish this, we will perform physical EPPs (EPPs that do not involve encoded logical qubits) at the start. We will use a scheme explored by Krastanov et al~\cite{Krastanov_2019}, which purifies 1 ebit from 5, as shown below (on one side)
\begin{figure}[H]
    \centering
    \begin{tikzpicture}
    \node[scale=1]
    {
    \begin{quantikz}
        & \ctrl{1} & \qw & \qw & \qw & \targ{} & \qw & \qw \\
        & \control{} & \ctrl{1} & \meter{X} & & \ctrl{-1} & \ctrl{1} & \meter{X} \\
        & \qw & \control{} & \meter{X} & & \qw & \control{} & \meter{X}
    \end{quantikz}
    }; 
    \end{tikzpicture}
    \caption{A purification circuit (one side). Alice and Bob share 5 ebits initially. They then locally perform the circuit and measurements above. Upon comparing the 4 measurement results, if all of them agree, they keep the first ebit; otherwise, discard.}
    \label{expedient}
\end{figure}
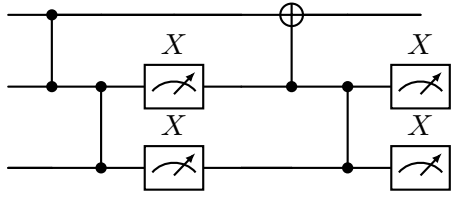
Given the local physical error rate $\epsilon$, we have that the infidelity after one round of purification being 
\begin{align*}
    I(q,\epsilon)&=\frac{1}{2}\epsilon+\frac{63}{16}\epsilon^2+\frac{49}{4}q\epsilon+6q^2+\frac{879}{64}\epsilon^3+\frac{107}{2}q\epsilon^2+\frac{345}{4}q^2\epsilon+48q^3
\end{align*}
If we perform another round of EPP, the infidelity will be $I(I(q,\epsilon)/3,\epsilon)$ etc. In the next section, we will obtain a theoretical lower bound for the threshold value. With this information, we can determine the necessary number of rounds and the success probability of EPP, facilitating subsequent resource comparisons. We note that if we only perform physical EPP, the infidelity is bounded away from zero due to local errors in the locations. This justifies the necessity of performing logical EPP, as detailed later. 
\section{ExRec and Threshold}\label{exRecthres}
In this section, we consider the concatenated $[7^k,1,3^k]$ Steane code. The preparation and measurement gadgets used in this paper are detailed in Appendix~\ref{FT_construction}. At $k=1$, since [7,1,3] is a doubly even self-dual CSS code, it admits a transversal implementation of the logical Clifford group. Hence, transversal CNOT is $\overline{\mathbf{CNOT}}$. Throughout this work, the EC used will be the Steane EC because of its relative convenience in theoretical analysis and relatively high pseudo-threshold. In practice, we may use the flag error correction~\cite{RuiChao} to save physical qubits and similar conclusions will follow. Besides, since the Steane code is of distance 3, it's capable of correcting one error. So at least two faults are needed to cause a logical error. We will thus confine Definition~\ref{malig_set} to \textit{malignant pairs} of locations. In our simulations, a pair of locations in an exRec is identified as malignant when noise is introduced into these specific locations, while others remain fault-free, and this leads to a logical error for the Rec. However, the counting procedure needs more prudent treatment. The Steane code has the nice property that a faultless EC will take any input to the codespace and an EC with one fault will take any input to a state that deviates at most weight-one error from the codespace. So essentially we are testing the correctness of the Rec given the input to the Rec is an operator $E_i$ of weight at most 1. This is sufficient for fault-tolerance since the whole circuit can be seen as a string of Recs followed by one another except the preparation exRec, which was justified to be fault-tolerant. Given previous definitions of malignant set and badness, we can treat the exRecs as independent when generalizing to a higher level of concatenation. Combining the above constructions and ideas, we can prove the following theorem:
\begin{theorem}
    Suppose that independent stochastic noise occurs with probability at most $\epsilon$ at each location in a noisy quantum circuit. Then the logical error rate $\varepsilon^{(k)}$ of the largest $k$-exRec satisfies the following bounds
    \begin{equation*}
        \frac{1}{A_5^{(k)}}\left(A_5^{(k)}\prod_{i=1}^{k-1}\left(A_5^{(i)} \right)^{2^{k-1-i}}\epsilon^{2^{k-1}} \right)^2\leq\varepsilon_5^{(k)}\leq\frac{1}{D_5^{(k)}}\left(D_5^{(k)}\prod_{i=1}^{k-1}\left(D_5^{(i)} \right)^{2^{k-1-i}}\epsilon^{2^{k-1}} \right)^2
    \end{equation*}
    for $\epsilon<\epsilon_0\approx4.70\times10^{-4}$, where $\epsilon_0$ is the threshold value, $\lim_{k\rightarrow\infty}A_5^{(k)}=A_5^*=979.7$ and $\lim_{k\rightarrow\infty}D_5^{(k)}=D_5^*=1827.1$. Both bounds approach 0 as $k\rightarrow\infty$. 
\end{theorem}
The formal proof of this theorem is provided in Appendix~\ref{thres_thm_proof} and the proof of the auxiliary lemma is provided in Appendix~\ref{third-order}.
\begin{proof}[Proof sketch]
\hangindent = 2em
\hangafter = 1
To set the stage, we will first introduce the concept of \textit{malignant pair matrix (MPM)}, denoted $\alpha$. It is a $7\times7$ real symmetric matrix where the rows and columns correspond to the 7 types of locations. The entry $\alpha(i,j)$ denotes the number of malignant pairs caused by one location of type $i$ and another of type $j$. The MPMs associated with different gadgets/exRecs are provided in Appendix~\ref{MPM}. Additionally, we denote the vector representing the number of different locations by $\mathbf{n}$.\\
\hspace*{1em}To establish bounds for the logical error rate at level $k$, we begin by deriving bounds for the logical error rates of 1-exRecs of all locations. These bounds serve as the foundation for determining error rates at higher levels through recursive simulation. This initial step is crucial, particularly when aiming for rigorous lower bounds and tighter upper bounds because the locations are interdependent. For example, CNOT $k$-exRec uses $(k-1)$-exRecs of preparation/measurement/single-gate/CNOT. Thus, the relative proportions of error rates at $k=1$ do not hold at $k=2$. \\
\hspace*{1em}Next, we use combinatorial methods to obtain the bounds. Second-order terms (in $\epsilon$) can be computed from MPMs. For third-order terms and beyond, we will apply a lemma, which prevents double-counting the cases already accounted for by the second-order term. Since this lemma will be applied in various derivations of the main text, we explicitly state it here:
\begin{lemma}
\hangindent = 2em
\hangafter = 0
    Given malignant pair matrix $\alpha\in\mathbb{R}^{n\times n}$, $\mathcal{O}(\epsilon^3)$ that potentially causes a logical error can be upper bounded by $F\epsilon^3$ where
    \begin{equation*}
        F=F(\mathbf{n},
\mathbf{\sigma},\alpha
)=\sum_{s=1}^n f_{ss}+\sum_{\{s,t\}\in\binom{n}{2}}f_{st}
    \end{equation*}
    for $f_{ss}=\binom{n_s}{3}\sigma_s^3-\frac{1}{3}(n_s-2)\sigma_s\alpha_{ss}$ and $f_{st}=\binom{n_s}{2}\cdot n_t\cdot\sigma_s^2\sigma_t+\binom{n_t}{2}\cdot n_s\cdot\sigma_s\sigma_t^2-\frac{1}{3}(n_s-1)\sigma_s\alpha_{st}-\frac{1}{3}(n_t-1)\sigma_t\alpha_{st}$.
    \label{third_order_thm}
\end{lemma}
Another subtlety to be taken care of is that EC and preparation gadgets contain ancilla verification. Thus we need to apply Bayes' rule to account for these. In the end, we hope to obtain bounds on the logical error rates of the form
\begin{equation*}
    A_i^{(k)}\left(\epsilon_i^{(k-1)} \right)^2\leq\epsilon_i^{(k)}\leq D_i^{(k)}\left(\epsilon_i^{(k-1)} \right)^2
\end{equation*}
where $\epsilon_i^{(k)}$ denotes the logical error rate of $Loc_i$ $k$-exRec, $A_i^{(k)},D_i^{(k)}$ are constants. As we generalize to level-$k$, these constants form a discrete-variable dynamical system in $k$. As we have the initial point $k=1$, we can use the fixed-point iteration method to find the non-trivial fixed point and compute the Jacobian to verify its stability. With the numerics, we can find $A_i^{(k)}, D_i^{(k)}$ $\forall i,k$. The recursive relation will give the desired result. For the threshold value, we identify CNOT-exRec as the largest exRec. Thus we may simply take $1/\max_kD_5^{(k)}$ as the threshold value, although the actual threshold will be higher than this. 
\end{proof}
In the rest of the paper, for each $k$, we denote the lower and upper bound by $\mu_5^{(k)}$ and $\nu_5^{(k)}$ respectively. Similarly we can obtain $\mu_i^{(k)}$ and $\nu_i^{(k)}$ for other locations given the corresponding $A_i^{(k)}$ and $D_i^{(k)}$. Up to now, we have obtained bounds on the logical error rates of the exRecs and a threshold on fault-tolerant computation locally.

\section{Direct Encoding}\label{Dir_Enc}
In this section, we discuss about the commonly conceived way of preparing $\overline{\text{ebit}}^{(k)}$. We will outline the implementation under the fault-tolerant construction and subsequently provide an analysis of the logical error rate for $\overline{\text{ebit}}^{(k)}$ prepared.
\begin{figure}[H]
    \centering
    \begin{subfigure}[b]{0.45\textwidth}
        \centering
        \includegraphics[width=\textwidth]{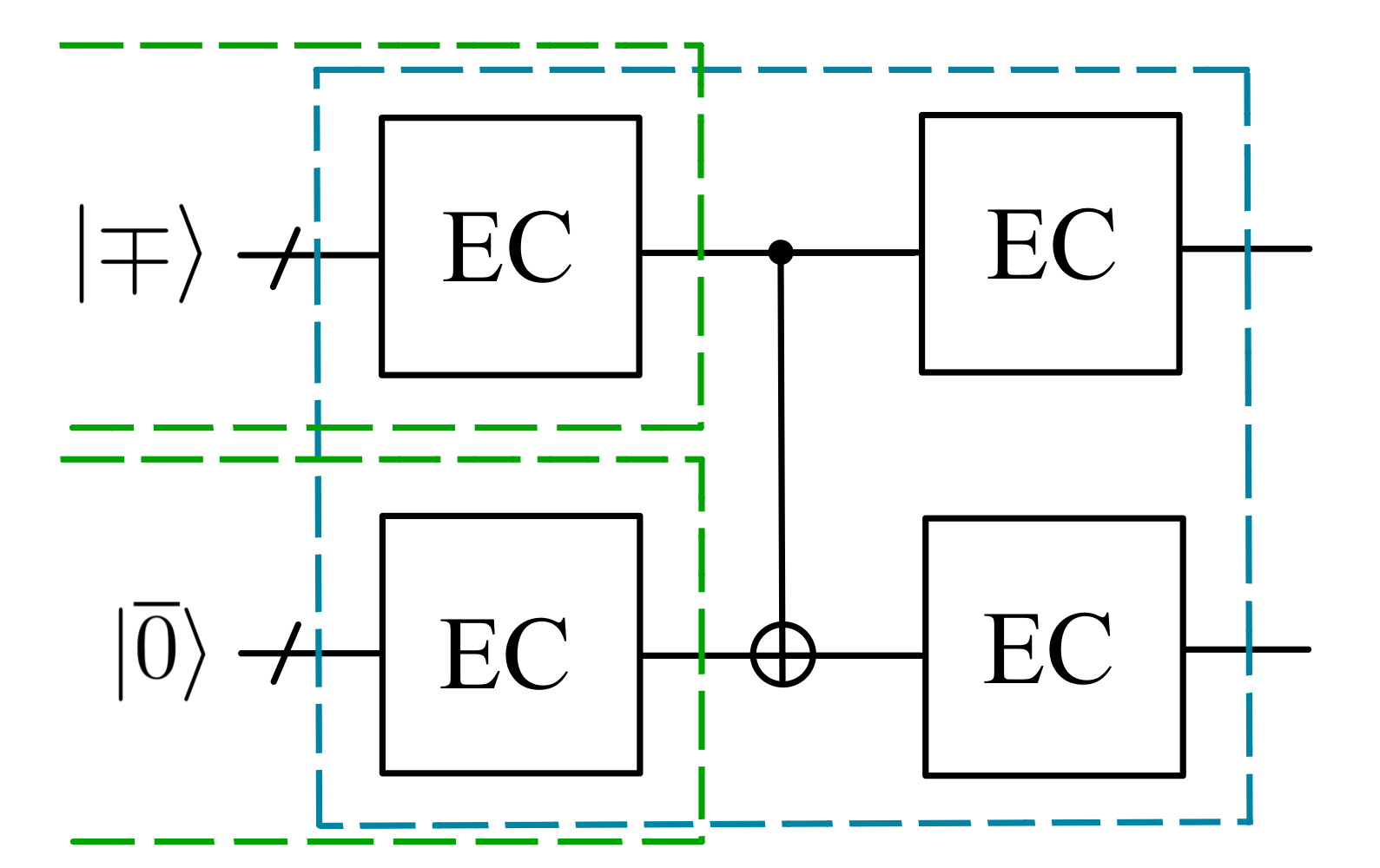}
        \subcaption{}
        \label{DEscheme1}
    \end{subfigure}
    \hfill
    \begin{subfigure}[b]{0.45\textwidth}
        \centering
    \begin{tikzpicture}
    \node[scale=1.2]
    {
    \begin{quantikz}
      \qw  &\ctrl{1} & \qw & \qw & \gate{Z} & \qw \\
     \makeebit[angle=-40,label style=blue]{}  &\targ{} &\meter{Z} & \cwbend{2} \\
     \qw & \ctrl{1}  &\meter{X} & \setwiretype{c}& \cwbend{-2} \\
      & \targ{} & & \gate{X} & &
    \end{quantikz}
    }; 
    \end{tikzpicture}
    \subcaption{}
    \label{gate_tele}
    \end{subfigure}
    \caption{(a)\textit{Direct Encoding}: Alice and Bob locally prepare $|\overline{+}\rangle$ and $|\overline{0}\rangle$ respectively. They then prepare $\overline{\text{ebit}}$ via a non-local logical CNOT. The green boxes represent preparation-exRec while the blue box is a CNOT-exRec. (b)The circuit for performing a CNOT with one ebit and local operations. The first and last qubits are the independent qubits we hope to perform CNOT on with the first being control and the last being target. In the middle two sides share an ebit.}
    \label{DE_schemes}
\end{figure}
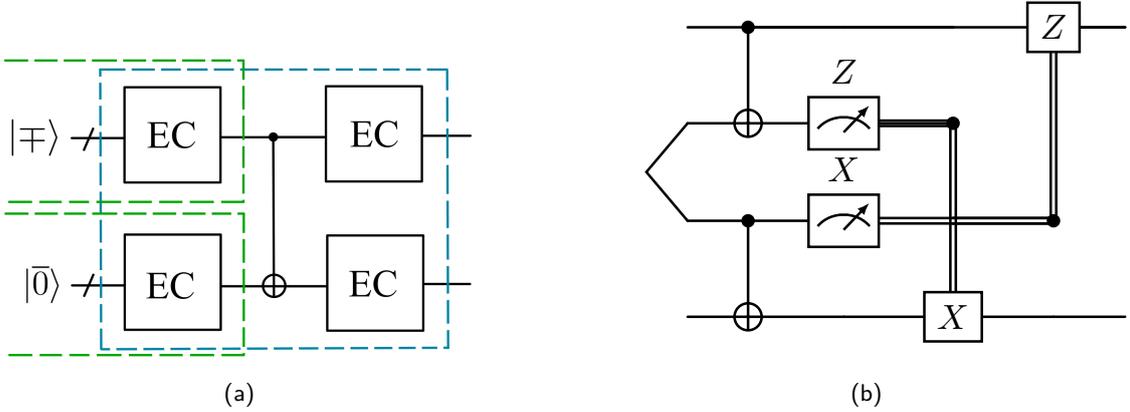
In this case, Alice locally prepares a $|\overline{+}\rangle$ and Bob prepares $|\overline{0}\rangle$. They then jointly perform a logical-CNOT, giving $\overline{\text{ebit}}$. The procedure is illustrated in Figure~\ref{DEscheme1}. However, there is an extra layer of complication here. Based on our assumption that the only entangling resources shared by two parties are noisy ebits, i.e. Alice and Bob cannot directly apply CNOT between their qubits. To tackle this issue, we observe that by utilizing a single ebit and gate teleportation, we can implement a CNOT gate between two independent qubits. Such an efficient circuit is shown in Figure~\ref{gate_tele} (Zhou et al.~\cite{Zhou_2000}, also experimentally demonstrated by Chou et al.~\cite{Chou_2018}). In fact, Alice and Bob can also start with $7^k$ physical ebits and locally measure the stabilizers to prepare $\overline{\text{ebit}}^{(k)}$. However, to ensure fault-tolerance, multiple rounds of stabilizer measurements are required. Due to the complexity of the circuit, our analysis will focus solely on the previous scheme. 
\subsection{Logical error rate}
The error rate of the nonlocal CNOT circuit in Figure~\ref{gate_tele} will be $\epsilon_6$, which depends on $q$, the infidelity of ebit and $\epsilon$. To guarantee that the above construction works, we require $\epsilon_6\leq\epsilon_0$. Concerning the fidelity of the EPR pair, if we start with an initial value of $3q=10\%$, we observe $I(q,0)=8.44\times10^{-3}$ and $I(I(q,0)/3,0)=4.86\times10^{-5}\leq\epsilon_0$. Since $I(q,\epsilon)> I(q,0)\geq\epsilon_0$, at least two iterations of initial EPP are required to sufficiently reduce the infidelity to a level comparable to the threshold. Subsequent iterations of EPP are not expected to significantly decrease the infidelity as it's lower bounded by $I(0,\epsilon)$. As we need $\epsilon_6(I(I(q,\epsilon)/3,\epsilon),\epsilon)\leq\epsilon_0$, from simulation we obtain an upper bound on $\epsilon$ being $\epsilon\leq\epsilon_0'=2.25\times10^{-4}$. Hence, to ensure the efficacy of \textit{Direct Encoding}, it is imperative to reduce the threshold, and thus the physical error rate. For higher-level simulation, when $\epsilon=\epsilon_0',\epsilon_6=\epsilon_0$, we will run the system of equations with $\epsilon_6$ included, such that $\sigma_6=2.09$ and $\epsilon_0^{(0)}=\epsilon_0'$. The threshold equation is modified as $\epsilon_0^{(k+1)}=\min_{i\leq (k+1) } \left\{1/D_5^{(i)}, 1/D_6^{(i)}\right\}$.
\\

Now, to analyze the bounds for $\mathbb{P}(\overline{\text{ebit}}^{(k)}\text{ bad})$, we will employ the bounds on exRecs outlined in the preceding section. Instead of directly utilizing the bounds on $\varepsilon_1^{(k)}, \varepsilon_2^{(k)},$ and $\varepsilon_5^{(k)}$, we choose to analyze with level-$(k-1)$ gadgets as a detour. As we will see later, this is necessary to make a fair comparison with the \textit{Interface+EPP} methods and also to obtain tighter bounds. The total number of locations in the circuit is $\gamma_{\text{EPR}} = 319$. It is noteworthy that, in addition to the standard stabilizers of the Steane code,  the encoded EPRs are stabilized by $\{\overline{XX}, \overline{ZZ}\}$. For this reason we will obtain the MPM for $\overline{\text{ebit}}^{(k)}$ and bound the logical error rate with level-$(k-1)$ simulation. Definition \ref{badness} justifies the feasibility of the following calculations. In summary, we have, for level-$k$ encoding,
\begin{align*}
    &\sum_{j\leq i=1}^6\alpha_{\text{EPR}}(i,j)\mu_i^{(k-1)}\mu_j^{(k-1)}\left(1-\nu_5^{(k-1)}\right)^{\gamma_{\text{EPR}}-9}\dots\\
    &\dots\left(1-\nu_6^{(k-1)}\right)^{7}\leq \mathbb{P}(\overline{\text{ebit}}^{(k)}\text{ bad})\\
    &\leq\sum_{j\leq i=1}^6\alpha_{\text{EPR}}(i,j)\nu_i^{(k-1)}\nu_j^{(k-1)}+\dots\\
    &\dots\overline{F}(\vec{n}_{\text{EPR}},\vec{\sigma}_L^{(k)},\vec{\sigma}_U^{(k)},\alpha_{\text{EPR}})\left(\nu_5^{(k-1)}\right)^3
\end{align*}
where $\alpha_{\text{EPR}}$ is as in Appendix~\ref{MPM}. For $k=1$, if $\epsilon,\epsilon_6$ are as above, then
\begin{align*}
    &2197.6\epsilon^2(1-\epsilon)^{310}(1-\epsilon_6)^7\leq\mathbb{P}(\overline{\text{ebit}}^{(1)}\text{ bad})\\
    \leq &2197.6\epsilon^2+1.355\times10^6\epsilon^3
\end{align*}
For $k\geq2$, To put the bounds in a simpler form, we note that $\{\vec{\sigma}_L^{(k)}\}, \{\vec{\sigma}_U^{(k)}\}$ are bounded sequences, so for the lower and upper bounds we may take the infimum and supremum respectively. The detailed data is left in Appendix \ref{sys_eqn}. For the upper bound, if we denote $\sigma_{U,i}^{\sup}=\sup_k\sigma_{U,i}^{(k)}$ and $\sigma_{L,i}^{\inf}=\inf_k\sigma_{L,i}^{(k)}$,
\begin{align*}
    &\mathbb{P}(\overline{\text{ebit}}^{(k)}\text{ bad})\\
    \leq&\sum_{i,j}\alpha_{\text{EPR}}(i,j)\sigma_{U,i}^{\sup}\sigma_{U,j}^{\sup}\left(\nu_5^{(k-1)} \right)^2\\
    &+\overline{F}(\mathbf{n}_{\text{EPR}},\underline{\sigma}_L^{\inf},\underline{\sigma}_U^{\sup},\alpha_{\text{EPR}})\left(\nu_5^{(k-1)} \right)^3\\
    \leq& 1757.3\left(\nu_5^{(k-1)} \right)^2+1.143\times10^{6}\left(\nu_5^{(k-1)} \right)^3\\
    \leq& 1880.5\left(\nu_5^{(k-1)} \right)^2\\
    \leq& 0.883\epsilon_0\left(\frac{\epsilon}{\epsilon_0} \right)^{2^{k}}
\end{align*}
where the last two inequalities follow from the fact that $\epsilon\leq\epsilon_0'<\epsilon_0/2$ and given this, when $k\geq2$, $\nu_5^{(k)}\leq\epsilon_0(\epsilon/\epsilon_0)^{2^k}$ is a good upper bound. Similarly for the lower bound,
\begin{align*}
        &\mathbb{P}(\overline{\text{ebit}}^{(k)}\text{ bad})\\
    \geq&\sum_{i,j}\alpha_{\text{EPR}}(i,j)\sigma_{L,i}^{\inf}\sigma_{L,j}^{\inf}\left(\mu_5^{(k-1)} \right)^2-317\nu_5^{(k-1)}-7\nu_6^{(k-1)})\\
    \geq&1087.2\left(\mu_5^{(k-1)} \right)^2\\
    \geq&1.02\mu_0\left(\frac{\epsilon}{\mu_0} \right)^{2^k}
\end{align*}
where $\mu_0=1/1061.0\approx9.43\times10^{-4}$.

\section{Interface + EPP}\label{IEPP}
\subsection{Overview}\label{IEPPSchemes}
We will first give an overview of the \textit{Interface+EPP} method. We start with Alice and Bob sharing several noisy ebits. They then use an interface to encode the information in their physical qubits to the logical level. However, any such interface cannot be made fault-tolerant since they are encoding (locally) unknown states. Therefore they perform EPP to filter out the `bad' $\overline{\text{ebit}}$. The logical EPP scheme used in our work is illustrated below.
\begin{figure}[H]
    \centering
    \includegraphics[width=0.6\linewidth]{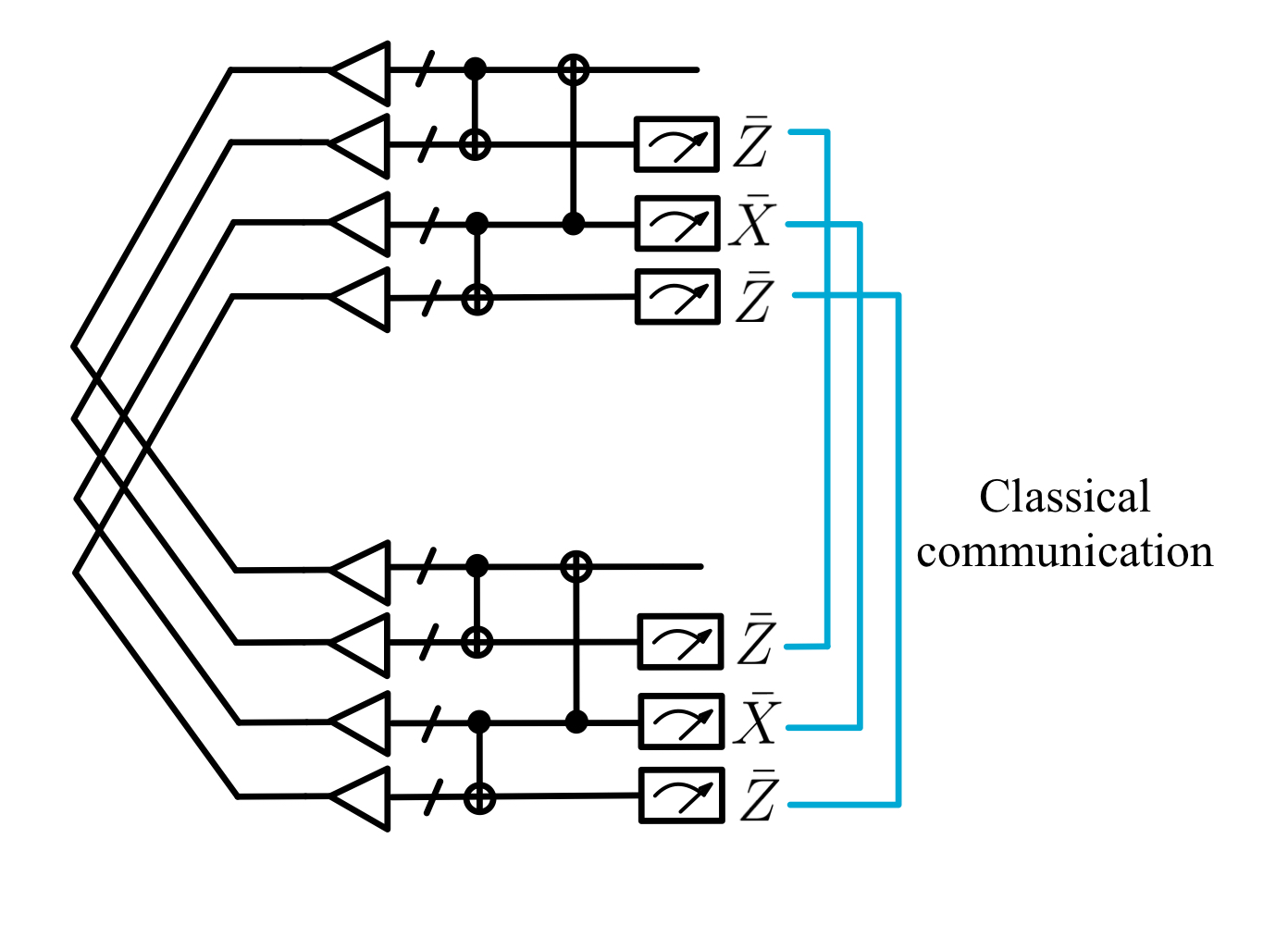}
    \caption{A broad overview of the Interface+EPP method. On the left-hand side, the curved lines represent ebits. The triangles are the interfaces. This is followed by $\overline{\text{EPP}}$ and classical communication. }
    \label{Interface+EPP}
\end{figure}
Alice and Bob commence with 4 physical ebits. They then employ interfaces to encode them into $\overline{\text{ebit}}^{(k)}$. Subsequently, they individually perform logical CNOT operations on the four logical qubits they possess. Finally, they conduct destructive measurements on the last three logical qubits, applying measurements in the $\overline{Z}$, $\overline{X}$, and $\overline{Z}$ bases, respectively. Upon obtaining classical outputs, they post-process (decode) the results and compare them through error-free classical communication.

In the remainder of the section, we will first introduce the specific construction of the interface in this paper. We then establish that the failure probability of the EPP using the interface remains bounded, independent of the parameter $k$. Following this, we analyze the logical error rate of the $\overline{\text{ebit}}^{(k)}$ given acceptance of $\overline{\text{EPP}}$. Regarding the physical error rate, the local threshold value depends solely on the local CNOT-exRec, which follows the previously described construction, with the same threshold $\epsilon_0$ applicable. We note that the nonlocal resources, in this case, are bare ebits(potentially after initial rounds of EPP), thus $\epsilon_6=3q$, the infidelity of the initial ebits. Unlike in Direct Encoding, there's no need to convert them to CNOT gates here. Besides, from the circuit construction, it's evident that $\epsilon_6$ doesn't affect the local threshold. 

However, as demonstrated later, the final logical error rate will have a dependence on $q$, thus we will also establish a lower bound on the threshold value for $q$. The number of initial EPP as in Figure \ref{expedient} needed will also depend on this threshold. To achieve exponential suppression, we explore two schemes.
\begin{figure}[H]
    \centering
    \begin{subfigure}[b]{0.7\textwidth}
        \centering
        \includegraphics[width=\linewidth]{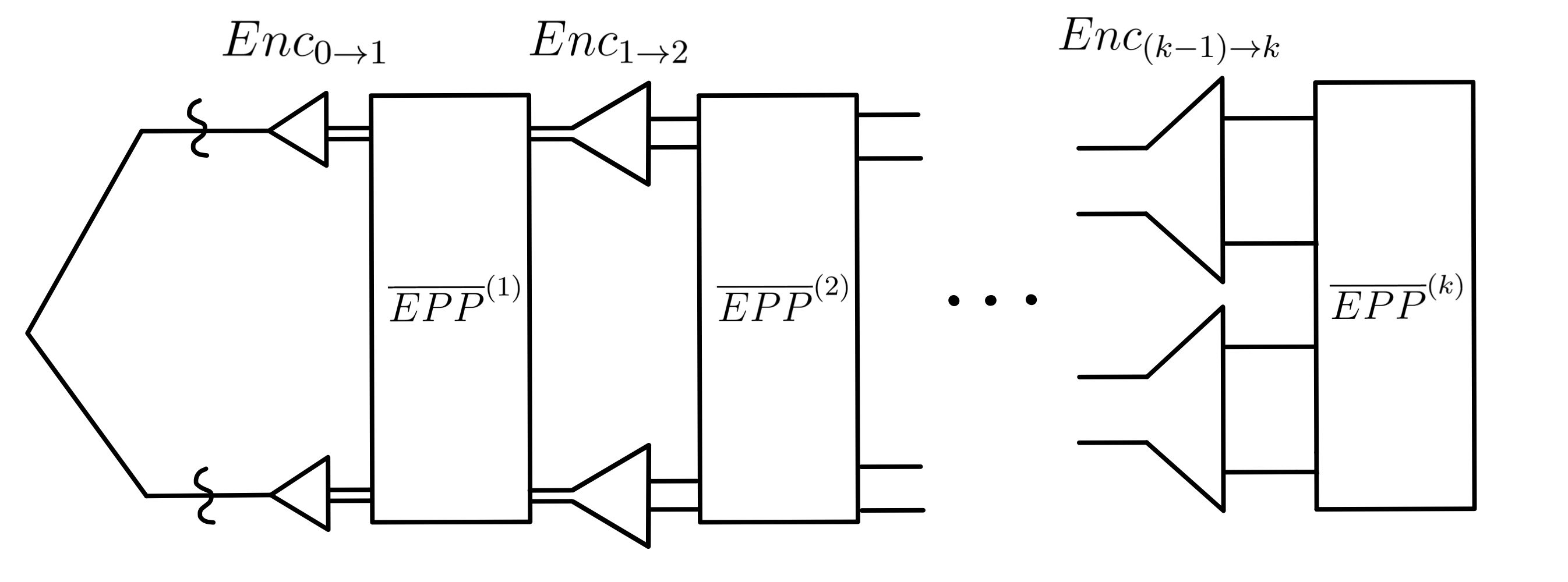}
        \subcaption{}
        \label{IFEPP_1}
    \end{subfigure}
    \vfill
    \begin{subfigure}[b]{0.6\textwidth}
        \centering
        \includegraphics[width=\linewidth]{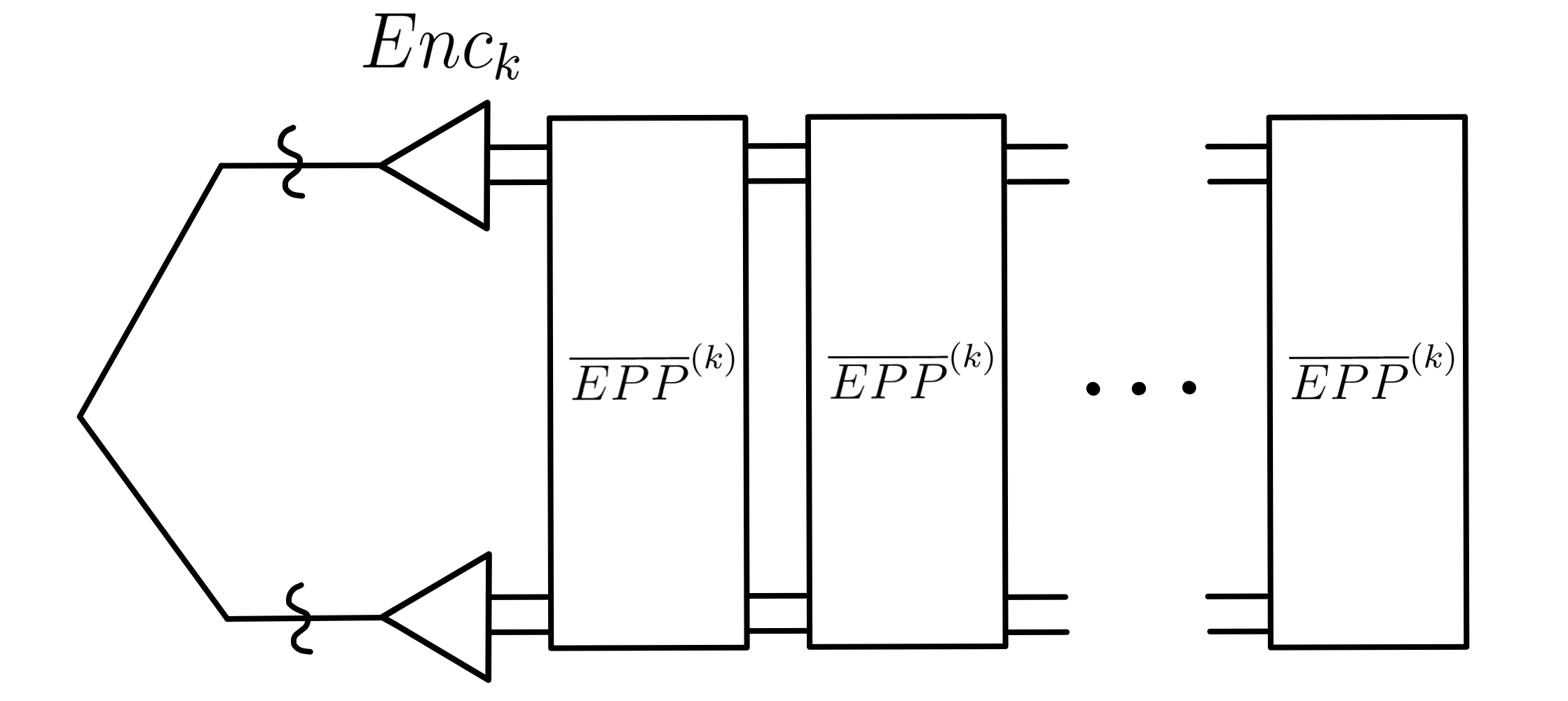}
        \subcaption{}
        \label{IFEPP_2}
    \end{subfigure}
    \caption{(a) Scheme $A$. Encoding is done sequentially, level by level, each followed by a logical EPP with 3 other similarly prepared $\overline{\text{ebit}}$. (b) Scheme $B$. State is encoded to the $k$-th level in a single interface, followed by rounds of $\overline{\text{EPP}}^{(k)}$ to reduce the logical error rate.}
    \label{IFEPP}
\end{figure}
In Scheme $A$, for the ebits Alice and Bob share, they encode locally to level-$k$ through the interface first. They then conduct $\overline{\text{EPP}}^{(k)}$ repeatedly with similarly prepared $\overline{\text{ebit}}^{(k)}$.  In Scheme $B$, Alice and Bob share 4 physical EPR pairs, encode to $k=1$, and perform $\overline{\text{EPP}}^{(1)}$ as above. They then encode to $k=2$ and do $\overline{\text{EPP}}^{(2)}$ with other 3 similarly prepared $\overline{\text{ebit}}^{(2)}$ etc. Through these two procedures, the logical error rate will decrease as $\mathcal{O}\left((\epsilon_6)^{2^l} \right)$ where $l$ is the number of EPPs performed. We will provide a more accurate analysis on logical error rates in the following sections to enable comparison with \textit{Direct Encoding}.

\subsection{Interface}\label{theInterface}
We now formally introduce the \textit{Interface}. To motivate the construction, we recall that in Section~\ref{exRecthres} we mentioned fault-tolerant encoding of a known state, in that case, $|0\rangle$. Here as we hope to minimize the number of initial ebits, it would be feasible for Alice and Bob to share a single ebit and locally encode their qubit into the logical space. In this case, locally they are encoding an unknown state, and the encoding gadget is called an \textit{interface}. Below, we present one construction, which is a modified circuit based on Christandl et al.~\cite{christandl2022faulttolerant}.
\begin{figure}[H]
    \centering
    \begin{tikzpicture}
    \node[scale=0.6]
    {
    \begin{quantikz}
     \lstick{\ket{\psi}} & \qw & \qw & \qw & \qw &\qw & \qw &\qw &\qw &\qw &\qw &\qw  & \ctrl{2} & \qw & \meter{} & \cwbend{2}\\
    \lstick{\ket{0}} &\qw & \ctrl{1}\vqw{1} & \qw &\qw &\qw &\qw &\qw &\qw &\qw &\ctrl{1} & \targ{} & \qw & \meter{Z}\\
     \lstick{\ket{+}} & \qw &\targ{} & \ctrl{1} & \ctrl{2} &\ctrl{3} &\ctrl{4} &\ctrl{5} & \ctrl{6} & \ctrl{7} &\targ{} &\ctrl{-1} &\targ{} & \qw & \meter{} & \cwbend{2}\\
    \lstick[7]{\ket{\overline{0}}} & \qw & \qw & \targ{} & \qw &\qw&\qw &\qw &\qw &\qw &\qw &\qw &\qw &\qw &\qw &\gate[7][1cm]{P}& \qw & \gate[7][3cm]{EC} &\qw\\
      & \qw & \qw & \qw & \targ{} & \qw &\qw &\qw & \qw &\qw & \qw & \qw& \qw& \qw & \qw & \qw &\qw &\qw&\qw\\
      & \qw &\qw &\qw &\qw &\targ{} &\qw &\qw &\qw &\qw & \qw & \qw& \qw& \qw & \qw & \qw &\qw &\qw&\qw\\
      & \qw &\qw & \qw &\qw &\qw &\targ{} &\qw &\qw &\qw &\qw  &\qw &\qw &\qw &\qw &\qw &\qw &\qw&\qw\\
      & \qw &\qw & \qw &\qw &\qw &\qw &\targ{} &\qw &\qw &\qw &\qw &\qw &\qw &\qw &\qw &\qw &\qw&\qw\\
      & \qw &\qw & \qw &\qw &\qw &\qw &\qw &\targ{} &\qw &\qw&\qw &\qw &\qw &\qw &\qw &\qw&\qw&\qw\\
     & \qw &\qw & \qw &\qw &\qw &\qw &\qw &\qw &\targ{} &\qw &\qw &\qw &\qw &\qw &\qw &\qw &\qw &\qw
    \end{quantikz}
    };
    \end{tikzpicture}
    \caption{Interface $\text{Enc}_{0\rightarrow1}$}
    \label{Inter_cir}
\end{figure}
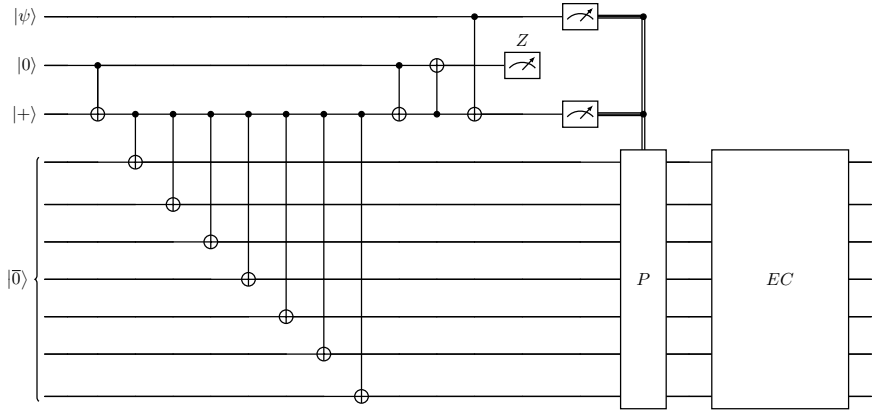
The circuit is essentially teleportation of a physical state into the logical space. We denote the interface from level-0 to level-1 (1 qubit $\rightarrow$ 7 qubits in $[7,1,3]$) by $\text{Enc}_{0\rightarrow1}$. If we hope to encode the state in 1 qubit into $[7^k,1,3^k]$, we use the interface $\text{Enc}_k=\text{Enc}_{(k-1)\rightarrow k}\circ\dots\circ\text{Enc}_{1\rightarrow2}\circ\text{Enc}_{0\rightarrow1}$ where each $\text{Enc}_{(i-1)\rightarrow i}$ has exactly the same structure as $\text{Enc}_{0\rightarrow1}$ but with each location replaced with the corresponding $(i-1)$-exRec. Note that we modified the original circuit by using an ancillary qubit to verify the validity of the prepared $|\Omega\rangle=\frac{1}{\sqrt{2}}(|0\rangle\otimes |\Bar{0}\rangle+|1\rangle\otimes|\Bar{1}\rangle)$. Through recursive simulation, we denote $|\Omega\rangle^{(k)}=\frac{1}{\sqrt{2}}(|\overline{0}^{(k-1)}\rangle\otimes|\overline{0}^{(k)}\rangle+|\overline{1}^{(k-1)}\rangle\otimes|\overline{1}^{(k)}\rangle)$. In the noisy circuit scenario, the recovery operation $P$ in Figure~\ref{Inter_cir} is followed by an identity gate on each qubit as the noise. We will investigate some of the properties of this interface. We will state a result similar to Lemma III.6 in Christandl et al.~\cite{christandl2022faulttolerant} but with a tighter upper bound.
\begin{lemma}[\text{Error probability of Enc$_l$}]\label{interface_prob}
\begin{align*}
    \mathbb{P}(\text{Enc}_l \text{ bad})\leq\exp\left(\beta_0\cdot C_1\left(\epsilon+\rho_1(\epsilon) \right)\right)\left(\beta_1\epsilon+\beta_2\epsilon^2+\beta_3\rho_1(\epsilon)+\beta_4\rho_1(\epsilon)^2 \right)
\end{align*}
where $C_1=18.1$, $\underline{\beta}=\begin{pmatrix}1.0043&2.8&628.5&2.43&521.3\end{pmatrix}$ and $\rho_1(\epsilon)=\sum_{k=1}^{\infty}\nu_5^{(k)}$.
\label{EpEnc}
\end{lemma}

We will use $f_{in}(\epsilon)$ to denote this upper bound in later context. From the bound, we observe that it's independent of $k$. In particular, when $\epsilon=\epsilon_0$, $\mathbb{P}(\text{Enc}_l \text{ bad})\approx4.40\times10^{-3}$, the detailed behaviour will be illustrated in Appendix \ref{error_type}. We note that, for $\epsilon\leq\epsilon_0$, we can actually upper bound $\rho_1(\epsilon)$ with a geometric series, which evaluates to a constant independent of $\epsilon$, which in turn gives a linear upper bound. However, in the pursuit of tight bounds, we will refrain from further elaboration here. Additionally, it's worth noting that based on the aforementioned proof, the upper limit of the failure probability for $Enc_{l\rightarrow (l+1)}$ is bounded by $f_{in}^{(l+1)}(\epsilon)=521.4\nu_5^{(l)}(1-C\nu_5^{(l)})$, which will prove to be useful in subsequent derivations.

\subsection{EPP rejection probability}
We will first give bounds on the failure probability of EPP, part of the results in the proofs are helpful for later derivation. We will state the upper and lower bounds for the above building block below and the proofs will be left in Appendix \ref{EPP_failure}. Thus we have the following
\begin{theorem}[Upper bound for $\mathbb{P}(\text{EPR rejected})$]
   Let Enc$_k$ be the $k$-th level interface circuit from Figure~\ref{Inter_cir}, then $\forall k\geq1$,
   \begin{equation}
       \mathbb{P}(\overline{\text{ebit}}^{(k)}\text{ rejected})\leq 4\epsilon_6+12\epsilon+8f_{in}(\epsilon)
   \end{equation}
   \label{ub_on_InEPP_A}
\end{theorem}
and the lower bound,
\begin{corollary}[Lower bound for $\mathbb{P}(\text{EPR rejected})$]
   Let Enc$_k$ be the $k$-th level interface circuit from Figure~\ref{Inter_cir}, then for $k\geq2$,
   \begin{equation*}
       \mathbb{P}(\overline{\text{ebit}}^{(k)}\text{ rejected})\geq 22.4\epsilon+4\epsilon_6-\sum_{i=0}^2u_i\epsilon^i\epsilon_6^{2-i}
   \end{equation*}
   where $\underline{\mathbf{u}}=\begin{pmatrix}16 & 4404 & 6067\end{pmatrix}$
   up to $\mathcal{O}(\epsilon^3)$.
   \label{lb_on_InEPP_A}
\end{corollary}
We will denote the lower and upper bound by $f_1(\epsilon, \sigma_6)$ and $f_2(\epsilon, \sigma_6)$ respectively for convenience.

\subsection{Logical error rate of Interface+EPP}\label{logicalerrEPP}
Having obtained bounds for the failure probability of encoded EPP, we analyze the logical error rate given that the EPR pair is accepted via EPP. The error probability can be analyzed similarly as in the previous section, except now the individual components of the circuit are not all fault-tolerant, and the interface preparation will cause problems. In addition, the fact that EPP rejects some EPR pairs makes the analysis more complex. In particular, we can no longer upper bound $\mathcal{O}(\epsilon^3)$ by simply calculating all combinations, which will result in significant overcounting. Thus, for a more systematic analysis, we will convert the circuit into a network flow problem, which will be discussed in detail for two schemes below.
Note that we will obtain $\mathbb{P}(\text{EPR bad}\wedge\text{EPR accepted})$ first and $\mathbb{P}(\text{EPR bad}\wedge\text{EPR accepted})$ will follows by
\begin{align*}
    &\mathbb{P}(\text{EPR bad}\wedge\text{EPP accepted})\\
    =&\frac{\mathbb{P}(\text{EPR bad}\wedge\text{EPR accepted})}{1-\mathbb{P}(\text{EPR rejected})}
\end{align*}
Before we delve into the analysis, we will first translate the circuit into a network flow graph to help us compute $\mathbb{P}(\text{EPR bad}\wedge\text{EPR accepted})$. The flow represents the propagation of errors. For convenience, we will analyze the $\overline{X}$ and $\overline{Z}$ errors separately. The graphs will help us decide the combinations of failures of exRecs/interfaces that would lead to logical $\overline{X}$-error despite passing the EPP test. 
The graph can be set up as follows: Let $N_{A,X}=(V,E)$ be the network above with $S,T\in V$ being the source and sink of $N_{A,X}$ respectively and $f$ is a function on the edges of $N_{A,X}$, its value on $(x,y)\in E$ is denoted by $f(x,y)$. $c(x,y)$ denotes the capacity of the edge. The propagation of $\overline{X}$-errors is then shown by the directed-graph
\begin{figure}[H]
    \centering
    \includegraphics[width=1.1\linewidth]{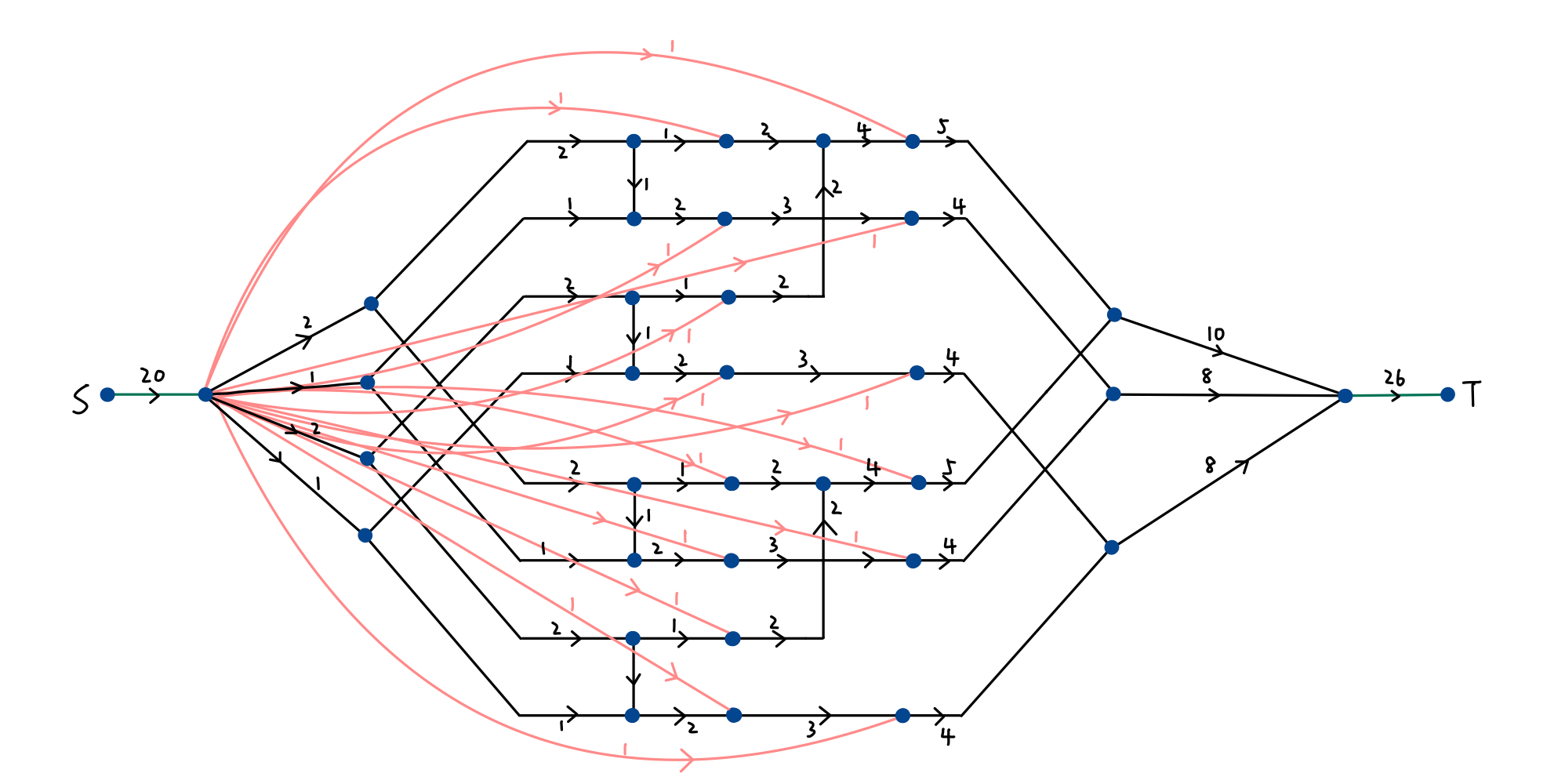}
    \caption{Network for $\overline{X}$-error. The source $S$ and sink $T$ are introduced for completeness. The numbers above the edges indicate the capacities. The pink edges all have capacity 1. }
    \label{Network_A}
\end{figure}
There are a couple of details worth clarifying
\begin{itemize}
    \item The central block is the standard \textit{Interface+EPP} scheme. The third logical EPR pair has no directly-connected $\{D_i\}$ vertex because in the scheme we perform $X$ measurement here and $\overline{X}$-error will not cause an issue.
    \item The parity of flow in the circuit indicates the presence of $\overline{X}$-error. If the flow is even, then the combined $\overline{X}$-errors cancel out, equivalent to no error; otherwise there is $\overline{X}$-error.
    \item Some edges with capacity 1 in the middle of the graph are directly connected with the source e.g. $A_2$ as the CNOT-exRec in the circuit will potentially introduce $\overline{X}$-error. For instance, if $f(S,A_2)=1$ and $f((S,A_6))=0$, this indicates the first CNOT-exRec introduces $\overline{XI}$-error. Besides, $SA_7, SA_{12}, SB_7, SB_{12}$ are introduced to indicate potential errors on mmt-exRecs.
\end{itemize}
Then the combinations of failures that contribute to $\mathbb{P}(\overline{\text{ebit}}\text{ bad}|\overline{\text{EPP}}\text{ accepts})$ can be determined via the following feasibility problem
\begin{align*}
    \min&\space\text{ } 0\\
    \text{s.t. }\space f(x,y)&\leq c(x,y)\text{ }\space\forall(x,y)\in E\\
    \space \sum_{x:(x,y)\in E,f(x,y)>0}f(x,y)&=\sum_{x:(y,x)\in E,f(y,x)>0}f(y,x)\\
    &\forall y\in V\setminus\{S,T\}\\
    \space f(x,y)&\cdot\left(f(x,y)-c(x,y) \right)=0\\
    &\forall (S,y)\in E\\
    f(D_1,T)&=2y_1+1\\
    f(D_2,T)&=2y_2\\
    f(D_3,T)&=2y_3\\
    y_1,y_2,y_3&\in\mathbb{N}\\
    f(x,y)\in\mathbb{N}\space\text{ }&\forall(x,y)\in E
\end{align*}
The first two are standard flow constraints. The third one makes sure either the initial edge has $\overline{X}$-error or error-free. $f(D_2,T)$ and $f(D_3,T)$ need to be even as for acceptance, the measurement results of two parties need to agree. $f(D_1,T)$ odd as for the desired scenario, we need the first $\overline{\text{ebit}}$ to have error. If $f(D_1,T)$ is even, this means the first logical EPR pair either has no error, or has $\overline{XX}$-error. However it is stabilized by $\overline{XX}$, this will not cause any problem. We give one example of the solution in Figure~\ref{Example_network}.
\begin{figure}[H]
    \centering
    \includegraphics[width=0.6\linewidth]
    {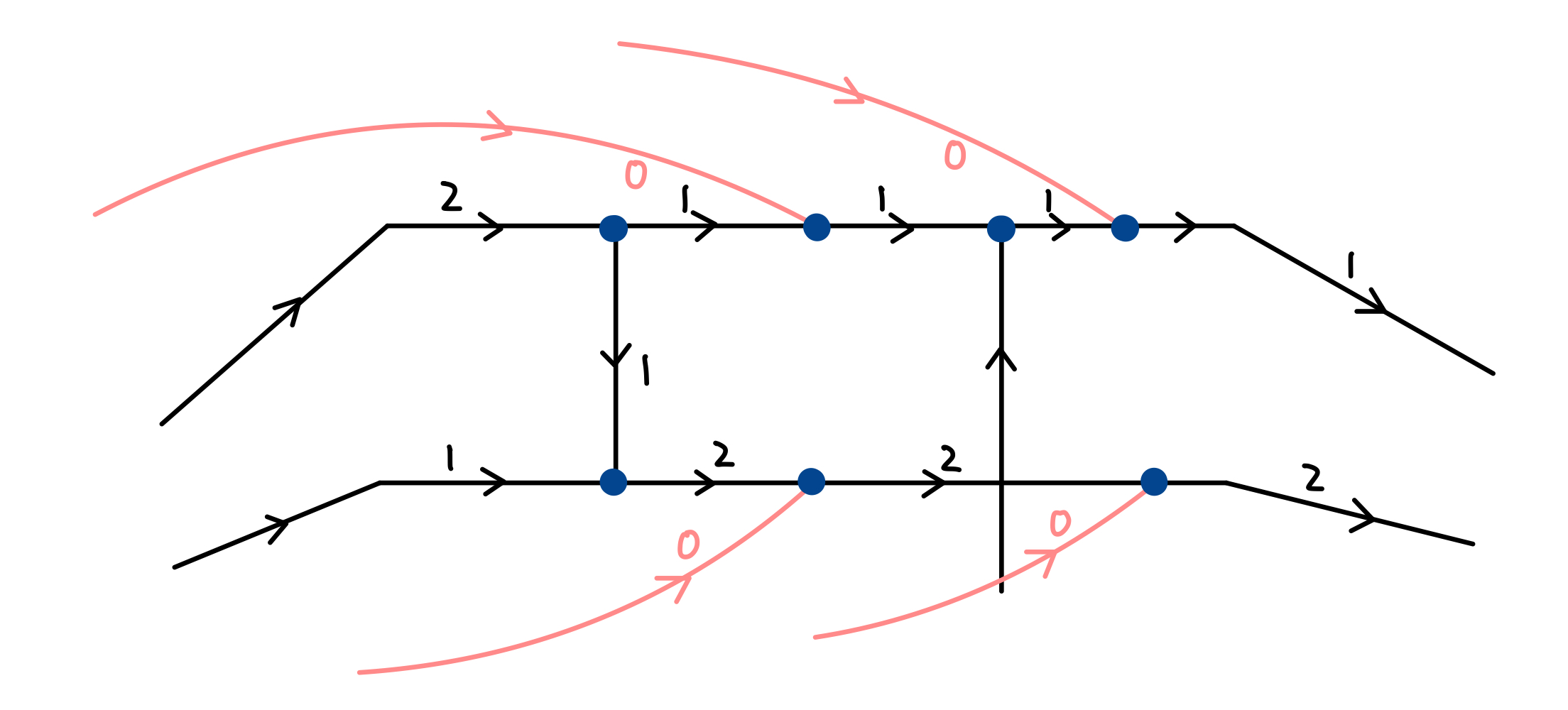}
    \caption{An example solution of the feasibility problem. This indicates an $\overline{X}$-error on the first $\overline{\text{ebit}}$ given $\overline{\text{EPP}}$ accepts. }
    \label{Example_network}
\end{figure}
This shows that on Alice' side, her logical qubits of the first two EPR pairs both contain $\overline{X}$-errors and they end up canceling out on the second data block, thus leading to acceptance while having errors on the first data block. Due to the third constraint, it would be easier to search for solutions exhaustively. \\

Analogously based on the propagation rule of $\overline{Z}$-error, we can generate a network graph for the propagation of $\overline{Z}$-error, the central block on Alice' side is shown in Figure~\ref{Z_prop}. The capacities can be chosen suitably. The details will not be repeated. 
\begin{figure}[H]
    \centering
    \includegraphics[width=0.5\linewidth]{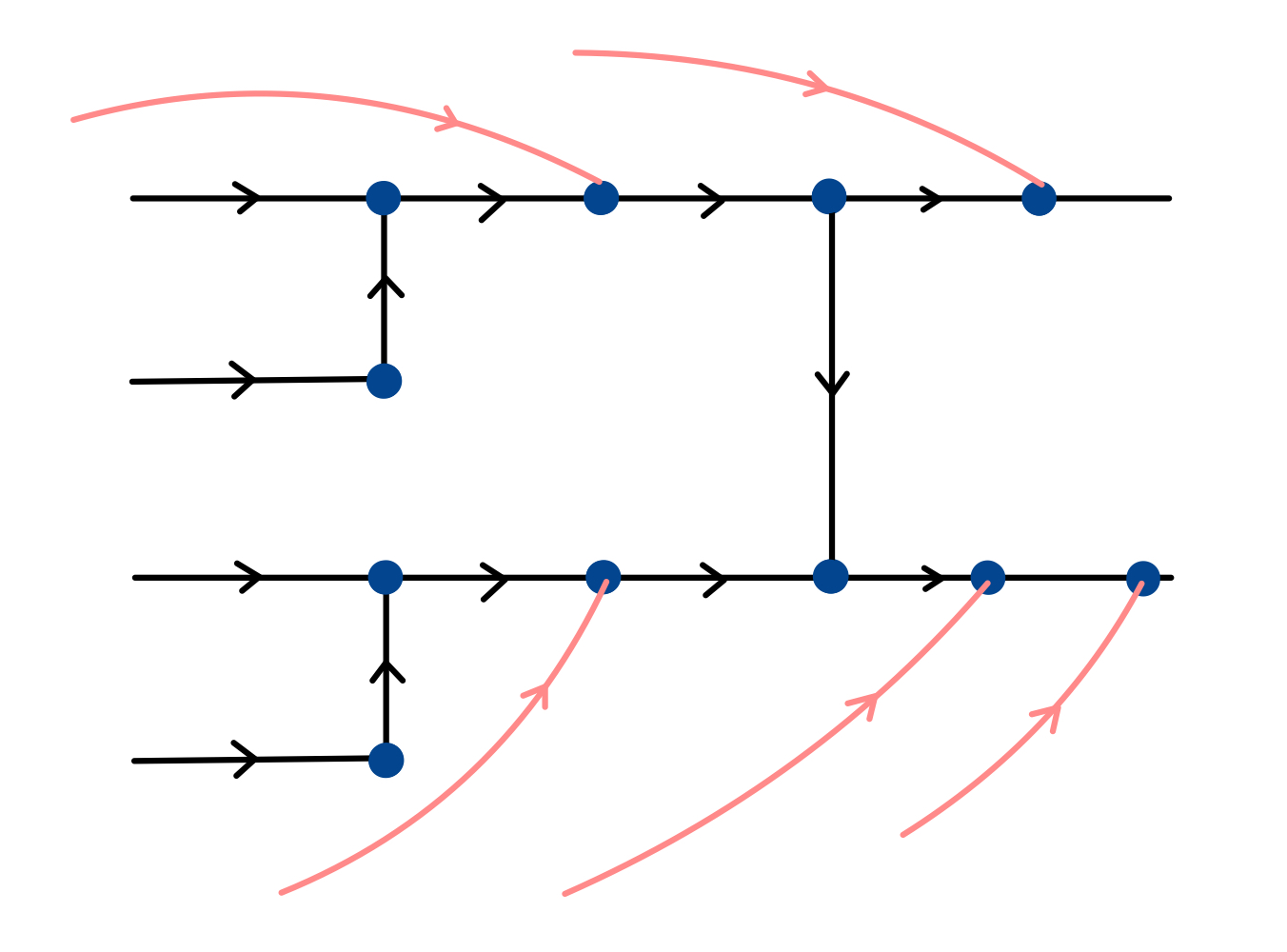}
    \caption{}
    \label{Z_prop}
\end{figure}
Any solution that is common to both types of errors indicates a potential $\overline{Y}$-error. Now we have a way to identify cases that lead to failure, we will match them with error terms of different orders. In the following we will denote the 4 $\overline{\text{ebit}}$s part as $\mathcal{P}_1$ and the $\overline{\text{EPP}}$ part as $\mathcal{P}_2$. Again we will obtain bounds for $k=1$ and then generalize to higher-level encoding. When there is one fault in the circuit, it will not cause a logical error in the EPP circuit according to fault-tolerance conditions. It can only cause logical error when it's present in the EPR preparation, thus one of $SC_1,SC_2,SC_3,SC_4$ is saturated while all other edges from the source are empty. There is no solution in both $\overline{X}$ and $\overline{Z}$ problems. Hence there are no first-order terms (which in turn also verifies our claim that the scheme is fault-tolerant). \\

For second-order terms, we will consider the distribution of the faults. For convenience of discussion, we will first label the $\overline{\text{ebit}}^{(k)}$'s and exRecs as follows,
\begin{figure}[H] 
    \centering
    \includegraphics[width=0.8\linewidth]{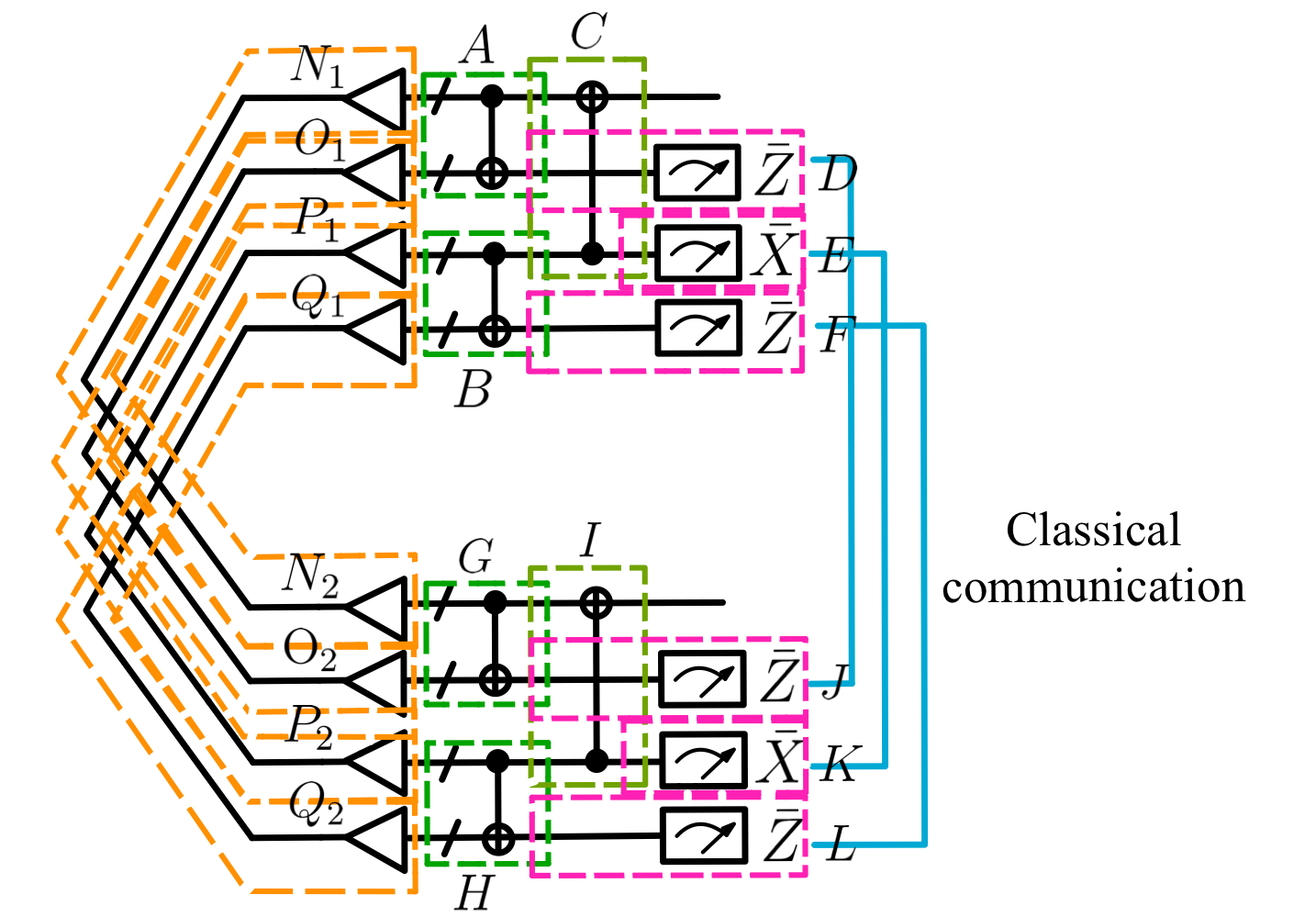}
    \caption{An Interface+EPP procedure broken down into its individual components. The orange parts denote $\overline{\text{ebit}}$s prepared via physical ebits and interfaces. Other parts correspond to exRecs within $\overline{\text{EPP}}$. For example, $A$ denotes the first CNOT-exRec on Alice's side.  }
    \label{IP_exRecs}
\end{figure}
From the different cases and the fact that 2 faults can cause at most 2 logical errors, we can post-process the feasible solutions and we'll present them in a more readable form in terms of the components that are bad. By combining cases that will result in both $\overline{X}$ and $\overline{Z}$ errors, we obtain the following cases
\begin{itemize}
    \item (2I) If they are both in $\mathcal{P}_2$, then they would contribute to the joint probability only when they are in the same exRec and form a malignant pair. From the filtered feasible solutions, we have the cases,
    \begin{itemize}
        \item $\overline{XI},\overline{XZ}$ on $A/B/G/H$
        \item $\overline{IX},\overline{IZ},\overline{IY},\overline{XY},\overline{XZ},\overline{XX}$ on $C/I$
    \end{itemize}
    The errors on the second one are reversed because the CNOT-exRec in the circuit is reversed.
    \item (2II) If they are both in $\mathcal{P}_1$, they must each be in one $\overline{\text{ebit}}$ and cause canceling errors. They must both have the same error type (meaning both are $\overline{X},\overline{Y}$ or $\overline{Z}$). Therefore,
    \begin{itemize}
        \item Both $\overline{X}$: $N_1O_1,N_1O_2,O_1N_2,P_1Q_1$, $P_1Q_2, P_1Q_1,N_2Q_2,P_2Q_2$
        \item Both $\overline{Z}$: $N_1P_1,N_1Q_1,N_1P_2,N_1Q_2$, $O_1P_1,O_1Q_1,O_1P_2,O_1Q_2$
    \end{itemize}
    \item (2III)If one is in $\mathcal{P}_1$ and one in $\mathcal{P}_2$, this is equivalent to the order-one case because according to fault-tolerance properties, the fault in $\mathcal{P}_2$ will result in correct simulation. 
\end{itemize}
Identifying cases that include errors in $\mathcal{P}_1$ is crucial. The reason will be made clear when we generalize to level-$k$. So for $k=1$ we have
\begin{itemize}
    \item (2I) To enumerate the number of such cases, we recall that the number of malignant pairs in a CNOT-exRec is 1656.3 and we then use the probability distribution of errors from Appendix \ref{dist_CNOT} and we count, in total 2587.3 malignant pairs.
    \item (2II) Since both errors are in $\mathcal{P}_1$, we refer to Table \ref{error_type} and count 29.1 pairs that do not contain $Loc_6$, 11.6 pairs that contain one $Loc_6$ and 1.32 pairs that are both $Loc_6$'s. This differentiation is essential, which we will see when we generalize to level-$k$.
\end{itemize}
Now we consider third-order terms. There are again a number of cases and from combining feasible solutions and suitable upper bounds, we summarize the following cases:
\begin{itemize}
    \item (3I) When the three faults are all in $\mathcal{P}_2$, this will only contribute to the desired probability if, they are all in $C/I$; they are in two consecutive exRecs and there is one error in the overlapping EC. So we upper bound the combinations and count 589499.6 such cases for $k=1$.
    \item (3II) One fault in $\mathcal{P}_1$ and two faults in $\mathcal{P}_2$ forming a malignant pair. We count an upper bound of 281160.3 cases not containing $\epsilon_6$ and 40165.8 containing $\epsilon_6$.
    \item (3III) The faults are in $\mathcal{P}_1$, two faults in one $\overline{\text{ebit}}$ and one in another. We count 8591.3 triples excluding $\epsilon_6$, 2645.4 triples including one $\epsilon_6$.
    \item (3IV) The faults are all in $\mathcal{P}_1$, with each in one 'bad' interface respectively. From the feasible solutions, we see there are no combinations that are of the same error type, so cases resulting in $\overline{X}$ and $\overline{Z}$ errors with one overlapping part. This case will be included in the second-order terms.
\end{itemize}
Having analyzed the different cases, we can obtain the following approximate bounds for $\mathbb{P}(\overline{\text{ebit}}^{(1)}\text{ bad})$,
\begin{align*}
    &(2587.3\epsilon^2(1-\epsilon)^{\gamma_1-2}(1-\epsilon_6)^4+11.6\epsilon_6\epsilon(1-\epsilon)^{\gamma_1-1}\dots \\
    &\dots(1-\epsilon_6)^3+1.32\epsilon_6^2(1-\epsilon)^{\gamma_1}(1-\epsilon_6)^2)\left(1-f_1(\epsilon) \right)^{-1}\\
    \leq&\mathbb{P}(\overline{\text{ebit}}^{(1)}\text{ bad})\\
    \leq&(1.32\epsilon_6^2+11.6\epsilon\epsilon_6+2587.9\epsilon^2+879251.2\epsilon^3\\
    &+42811.2\epsilon_6\epsilon^2)(1-f_2(\epsilon))^{-1},
\end{align*}
where $\gamma_1=1788$ is the total number of locations excluding $Loc_6$ at $k=1$. To ensure a fair comparison with \textit{Direct Encoding} later, we will confine $\epsilon\leq\epsilon_0'$ and consequently employ statistics of $\sigma$'s based on this constraint. We will denote the upper bound by $\kappa_2(1,\sigma_6)\epsilon^2$, where $\kappa_2(1,\epsilon_6)$ is obtained from the fact that $\epsilon\leq\epsilon_0'$ here. For the lower bound we will denote as $\sum_{i=0}^2\kappa_{1,i}(1,\sigma_6)\epsilon^2$ where each term is the coefficient for $(\sigma_6)^i\epsilon^2$, the reason for this special treatment will be clear when we generalize to level-$k$. Theorem \ref{ub_on_InEPP_A}(by taking $\epsilon=\epsilon_0'$ we have a supremum for $f_2$).
We will then generalize this to level-$k$. In fact, \textit{Interface+EPP} is guaranteed to function effectively as long as $\epsilon\leq\epsilon_0$, demonstrating one main advantage of this scheme. 
\subsection{\texorpdfstring{Scheme $A$}{Scheme A}}
For scheme $A$, we encode physical qubits to level-$k$ and then perform EPP repeatedly. So we would like to compute the logical error rate of $\overline{\text{ebit}}^{(k)}$ first. We will again consider errors of different order. For the second-order terms, the contribution from (2I) terms are upper bounded by $\sum_{i,j}\alpha(i,j)\sigma_{U,i}^{\sup}\sigma_{U,j}^{\sup}\left(\nu_5^{(k-1)} \right)^2\leq 2084.2\left(\nu_5^{(k-1)} \right)^2\leq2084.2/2.1^{2^k-2}\epsilon^2$. Thus we see that the contribution from $\mathcal{P}_2$ diminishes exponentially as $k$ increases. In fact, when $k\geq5$, this term is negligible. Next, (2II) terms persist. For third-order terms, again (3I) and (3II) are negligible as they give $\mathcal{O}(\epsilon^5)$ contribution; (3III) terms remain. When $k\geq2$, there is an additional case, that is, when there are two faults in one $Enc_{1\rightarrow2}$ and one fault in some other interface that combine and have a canceling effect. From the proof of Lemma \ref{EpEnc}, we can upper bound such cases by $\frac{1}{2}(\sigma_2^{(1)}+2\sigma_3^{(1)}+2)\varepsilon_5^{(1)}\cdot(2.8\cdot64\epsilon+34\epsilon_6)$, from the simulated values of $\sigma_i^{(1)}$ we have $445775.5\epsilon^3+84577.9\epsilon^2\epsilon_6$. 
If we denote the $k$-th level encoded $\overline{\text{ebit}}$ after $l$ rounds of $\overline{\text{EPP}}$ by $\overline{\text{ebit}}_l^{(k)}$, therefore we would have the upper bound $\forall k\geq2$
    \begin{align*}
        \mathbb{P}(\overline{\text{ebit}}_1^{(k)}&\text{ bad})\leq(2084.2/2.1^{2^k-2}\epsilon^2+29.1\epsilon^2+11.6\epsilon\epsilon_6\dots\\
        \dots+1.32\epsilon_6^2&+454366.8\epsilon^3+87223.3\epsilon^2\epsilon_6)(1-f_2(\epsilon))^{-1}\\
            &\leq \kappa_2^A(k,\sigma_6)\epsilon^2.
    \end{align*}
    where $\kappa_2^A(k,\sigma_6)$ is obtained through the same method as above. An example value is $\kappa_2^A(2,1)=644.1$. For the lower bound, we will obtain an approximation. From steps in the analysis of Cor. \ref{lb_on_InEPP_A} we obtain, $\forall k\geq2$,
    \begin{align*}
        &\mathbb{P}(\overline{\text{ebit}}_1^{(k)}\text{ bad})\\
        \geq&29.1\epsilon^2(1-\epsilon_6)^4(1-\epsilon)^{-2}+11.6\epsilon\epsilon_6(1-\epsilon_6)^3(1-\epsilon)^{-1}+\\
        &1.32\epsilon_6^2(1-\epsilon_6)^2)\left(1-8\sum_{t=1}^7n_i\epsilon_i\right)(1-f_1(\epsilon))^{-1}\\
        \geq&\sum_{i=0}^2\kappa_{1,i}^A(k,\sigma_6)\epsilon^2.
    \end{align*}
    Thus we have 
    \begin{equation*}
    \sum_{i=0}^2\kappa_{1,i}^A(k,\sigma_6)\epsilon^2\leq\mathbb{P}(\overline{\text{ebit}}_1^{(k)}\text{ bad})\leq \kappa_2^A(k,\sigma_6)\epsilon^2.
    \end{equation*}
    Now to achieve the same exponential suppression of logical error rate, we would need to perform further rounds of EPP. We may denote the upper bound of the logical error rate after the $l$-th EPP as $g^{(l)}_A(k,\sigma_6)$, where $g^{(1)}_A(k,\sigma_6)=\kappa_2^A(k,\sigma_6)\epsilon^2$. We may first obtain an upper bound on EPP rejection probability. Following a similar argument as in \ref{ub_on_InEPP_A}, 
    \begin{equation*}
        \mathbb{P}(\overline{\text{ebit}}_l^{(k)}\text{ rejected})\leq 4g_A^{(l-1)}(k,\sigma_6)+12\nu_5^{(k)}.
    \end{equation*}
    Thus for the upper bound, we establish the following recursive relationship for $2\leq l\leq k$,
    \begin{align*}
        g^{(l)}_A(k,\sigma_6)&=\Bigg( 6\left( g^{(l-1)}_A(k,\sigma_6)\right)^2\dots\\
        \dots+&2084.2\left(\nu_5^{(k-1)} \right)^2+24g^{(l-1)}_A(k,\sigma_6)\nu_5^{(k-1)}\Bigg)\dots\\
        \dots &\left(1-4g_A^{(l-1)}(k,\sigma_6)-12\nu_5^{(k)} \right)^{-1},
        \tag{$\dagger$}
        \label{dag}
    \end{align*}
    which follows from the union bound and the constants are obtained by counting the feasible solutions. In this case, some feasible solutions are combined. For example, both errors into 0 and 13 in Figure \ref{Network_A} are included by $g^{(l)}_A$. When $k\geq2$, the contribution from the rejection probability term is negligible. If we examine the magnitude of the above terms more closely and denote the second term by $b$, we can establish the following simple lemma,
    \begin{lemma}\label{l_prime}
        $\forall k\geq 2,k\in\mathbb{N}, \forall\epsilon>0$ and $\forall\delta>24b^2,\space \exists l'\in\mathbb{N}$ s.t. when $g_A^{(1)}\leq\frac{1-\sqrt{b}}{6}$, $\forall l\geq l'$, $|g_A^{(l)}-b|<\delta$. 
    \end{lemma}
    \begin{proof}
        The above recursive relation can be written as $g_A^{(l)}=6\left(g_A^{(l-1)} \right)^2+0.53\sqrt b g_A^{(l-1)}+b$. By solving $g_A^{(l)}<g_A^{(l-1)}$ and suitably lower-bounding the larger solution we obtain the upper bound on $g_A^{(1)}$. This implies that under this condition the series is decreasing, it's also bounded below by $b$, thus it converges. Since our recursion function is continuous, by Banach fixed-point theorem the fixed point is the converging limit. Hence, by solving the quadratic equation and upper-bounding the solution we obtain a bound on allowed $\delta$. $l'$ is chosen to be the minimal that satisfies the conditions.
    \end{proof}
    In the following context, we shall use $l'$ to denote the constant that satisfies Lemma~\ref{l_prime}. This Lemma indicates that when we perform sufficiently many EPP, $g^{(l)}_A$ will be dominated by the $\left(\nu_5^{(k-1)} \right)^2$ term and no further rounds of EPP will noticeably reduce the logical error rate. This aligns with the observation in the non-encoded case, that the effectiveness of EPP is limited by the error rate of local operations. In the later section, $l'_A$ will be determined numerically given $k,\epsilon$ and $\sigma_6$. This lemma greatly simplifies the upper bound as some terms dominate others. So we have
    \begin{equation*}
        g^{(l)}_A(k,\sigma_6)=\begin{cases}
        \frac{1}{6}\left(\sqrt{6\kappa_2^A}\epsilon\right)^{2^l},\space 2\leq l<l'_A\\
        2084.2\left(\nu_5^{(k-1)} \right)^2,\space l> l'_A
        \end{cases}.
    \end{equation*}
    Given the same threshold value $\epsilon_0'$, solving $\sqrt{6\kappa_2^A}=1/\epsilon_0'$ numerically would give a sufficient bound for $\epsilon_6$(That is, when $\sigma_6$ is below this bound, the logical error rate converges as $k$ increases). For $k\geq3$, $\sigma_6\leq813$, so $\epsilon_6\leq0.183$. And this surely satisfies the bound in the lemma. Thus theoretically if we start with noisy EPR pairs with infidelity of less than $18.3\%$ this protocol will work well. But in practice, if we hope to minimize the use of nonlocal resources, it might be better off to do some initial purifications. This consideration will be accounted for in more detail in the later section where we analyze resources. \\

Now for the lower bound, we will first consider the derivation of a lower bound for the combined circuit when we have two sequential block components under our construction, with known bounds of the logical error rate of each component. An illustration is shown in Figure~\ref{overlapEC}
\begin{figure}[H]
    \centering
    \includegraphics[width=0.5\linewidth]{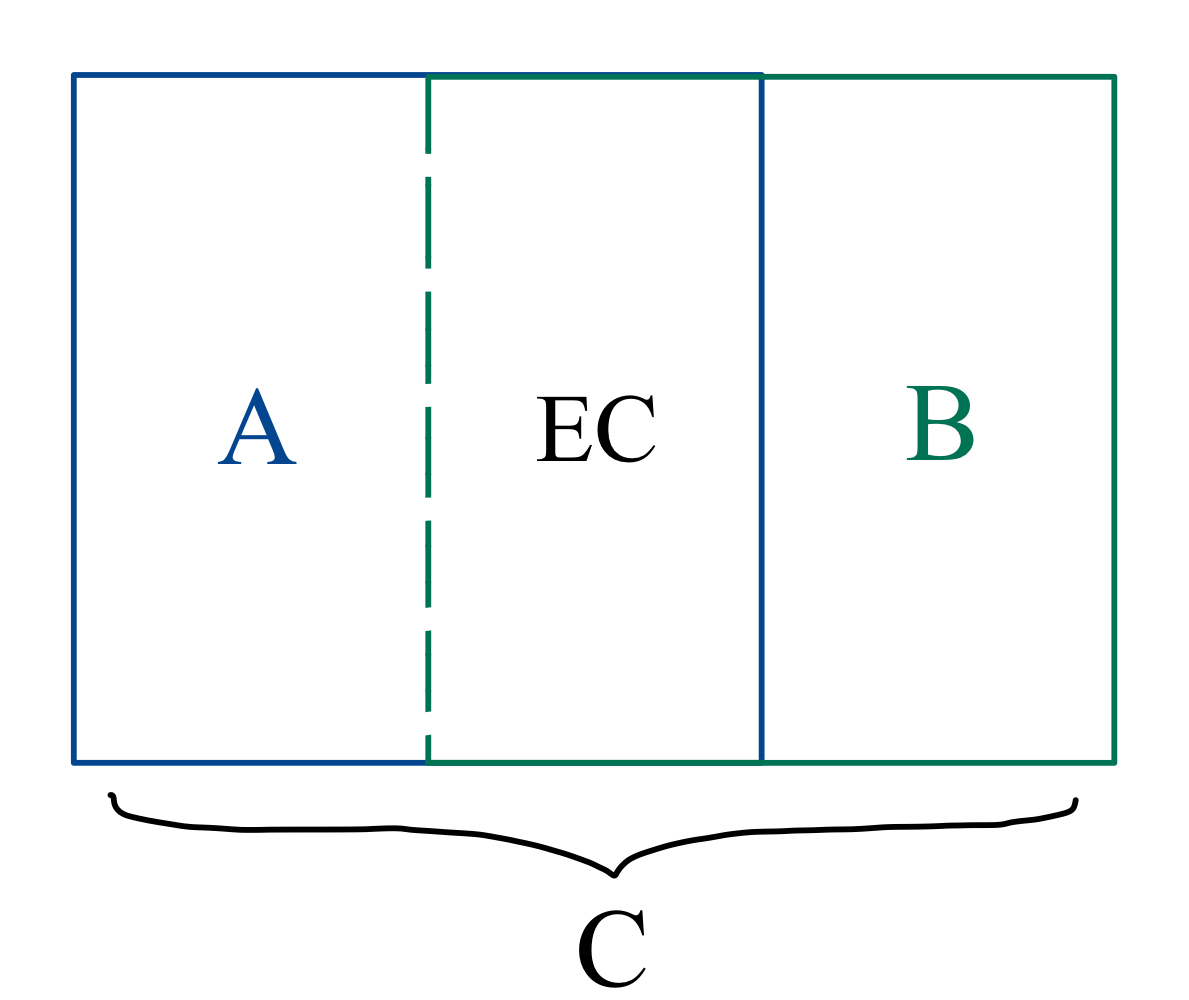}
    \caption{Two components $A$ and $B$ with an overlapping set of ECs}
    \label{overlapEC}
\end{figure}
According to the definition of badness, components $A$ and $B$ can be treated as independent and thus
\begin{align*}
    \mathbb{P}(C\text{ bad})&\geq1-(1-\mathbb{P}(A\text{ bad}))(1-\mathbb{P}(B\text{ bad}))\\
    &=\mathbb{P}(A\text{ bad})+\mathbb{P}(B\text{ bad}))-\mathbb{P}(A\text{ bad})\cdot\mathbb{P}(B\text{ bad}).
\end{align*}
Note that as we encode to higher levels the logical error rate is rather small, so the last term is negligible. Since we can compute lower bounds for $\mathbb{P}(A\text{ bad})$ and $\mathbb{P}(B\text{ bad})$, the lower bound for $\mathbb{P}(C\text{ bad})$ can be established. Further rounds of $\overline{\text{EPP}}$ follow the same rationale. Let us denote the probability of terms of order $(\sigma_6)^i$ of $\kappa_1^A(k,\sigma_6)$ being $Z$-logical error by $p_{Z,i}^{(k)}(\sigma_6)$, which is determined when we considered $Z$ feasible solutions. It's necessary to separate out $\overline{Z}$-error from $\overline{X}/\overline{Y}$-error. As previously indicated, 6 cases of logical errors arise from the failures of two $\overline{\text{ebit}}$ pairs, in which 4 are $\overline{Z}$-error and 2 are $\overline{X}/\overline{Y}$-error. We will denote the lower bound on $\overline{Z}$-error rate and $\overline{X}/\overline{Y}$-error rate after the $l$-th $\overline{\text{EPP}}$ as $h_{A,Z}^{(l)}(k,\sigma_6)$ and $h_{A,X/Y}^{(l)}(k,\sigma_6)$ respectively, thus we arrive at the following recursive relations
    \begin{align*}
        h_{A,Z}^{(l)}(k,\sigma_6)=& 4\left(h_{A,Z}^{(l-1)} \right)^2+517.4\left(\mu_5^{(k-1)} \right)^2 \\
        h_{A,X/Y}^{(l)}(k,\sigma_6)=&2\left(h_{A,X/Y}^{(l-1)} \right)^2+1252.4\left(\mu_5^{(k-1)} \right)^2,
        \tag{$\ddagger$}
        \label{ddagger}
    \end{align*}
    where $h_{A,Z}^{(1)}=4\sum p_{Z,i}^{(2)}\kappa_{1,i}^A(2,\sigma_6)\epsilon^2$ and $h_{A,X/Y}^{(1)}=2\left(\kappa_1^A\epsilon^2- h_{A,Z}^{(1)}/4\right)$. Again in this case we will have similar lemma, thus $\exists l''_A$ s.t.
    \begin{equation*}
        h_A^{(l)}(k,\sigma_6)=\begin{cases}
            \frac{1}{4}\left(\sqrt{h_{A,Z}^{(1)}} \right)^{2^l}+\frac{1}{2}\left(\sqrt{h_{A,X/Y}^{(1)}} \right)^{2^l},\space\text{ if } l<l''_A\\
            1770\left(\mu_5^{(k-1)} \right)^2,\space\text{ if } l\geq l''_A
        \end{cases}
    \end{equation*}
    
\subsection{\texorpdfstring{Scheme $B$}{Scheme B}}
For scheme $B$, we perform interface encoding and EPP iteratively, i.e. encode to $k=1$ with $Enc_{0\rightarrow1}$ followed by $\overline{\text{EPP}}^{(1)}$ and then encode with $Enc_{1\rightarrow2}$ followed by $\overline{\text{EPP}}^{(2)}$ etc. For the logical error rate in this scheme, we follow the same methodology as above. The only difference here is that in every round of $\overline{\text{EPP}}$, contributions from $\mathcal{P}_2$ have the same magnitude and thus need to be included. Hence for $k\geq 2$ and $2\leq l\leq k$ we have
    \begin{align*}
        g_B^{(l)}(k,\sigma_6)&=\bigg(6\left( g_B^{(l-1)}\right)^2+24g_B^{(l-1)}\nu_5^{(l-1)}+2084.2\left(\nu_5^{(l-1)} \right)^2\bigg)\dots\\
        &\dots\left(1-4g_B^{(l-1)}(k,\sigma_6)-12\nu_5^{(l)} \right)^{-1}\\
        h_{B,Z}^{(l)}(k,\sigma_6)&= 4\left(h_{B,Z}^{(l-1)} \right)^2+517.4\left(\mu_5^{(l-1)} \right)^2 \\
        h_{B,X/Y}^{(l)}(k,\sigma_6)&=2\left(h_{B,X/Y}^{(l-1)} \right)^2+1252.4\left(\mu_5^{(l-1)} \right)^2
    \end{align*}
    where $g_B^{(1)}=\kappa_2(1,\sigma_6)\epsilon^2$, $h_{A,Z}^{(1)}=4\sum p_{Z,i}^{(1)}\kappa_{1,i}(1,\sigma_6)\epsilon^2$ and $h_{A,X/Y}^{(1)}=2\left(\kappa_1(1,\sigma_6)\epsilon^2-h_{A,Z}^{(1)}/4 \right)$. In this case, the lemma no longer works because for each level $l\leq k$, the magnitudes of each term are comparable, and since $\kappa_2(1,\sigma_6)\geq 2815.6>1/\epsilon_0$. To simplify the upper bound, we may convert them to $C^{(l)}\left(\nu_5^{(l-1)} \right)^2$ at each level, where $C^{(l)}$ satisfies the difference equation
    \begin{equation*}
        C^{(l+1)}=\left( 6\left( C^{(l)}\right)^2+\frac{24}{\epsilon_0}C^{(l)}\right)(\epsilon_0)^2+2084.2
    \end{equation*}
    and $C^{(1)}=\kappa_2(1,\sigma6)$. It's easy to check that 2114.0 is a stable equilibrium and as long as $\sigma_6\leq528$, the sequence will converge to this fixed point. This puts a sufficient bound of $11.9\%$ on $\epsilon_6$. When $l>k$, we are back to equation (\ref{dag}), except in this case, we start with an initial point of $C^{(k)}\left(\nu_5^{(k-1)}\right)^2$ and previous results hold. Depending on the initial value, if $C^{(k-1)}$ is already $\delta$-close to 2114.0 for small $\delta$, then one additional round of $\overline{\text{EPP}}$ will diminish the logical error rate to $2084.2\left(\nu_5^{(k-1)} \right)^2$, but further EPP iterations will not enhance the outcome. We will denote the minimal number of $\overline{\text{EPP}}$ for this to occur by $l_B'$. Similarly, for the lower bound we can establish the difference equations and the lower bound for the logical error rate converges to $1774\left(\mu_5^{(l-1)} \right)^2$ provided the converges occurs for $l\leq k$; otherwise after $l=k$, the suppression of logical error rate follows equation (\ref{ddagger}). Similarly, we denote this boundary $l$ by $l''_B$. Hitherto we have completed analysis for these two schemes.

 \section{Fault-tolerant Magic Square Game}\label{FTMSG}
Having prepared the logical EPR pairs fault-tolerantly, we are now ready to play the magic square game. In this section, we will perform an analogous analysis of the failure probability of the magic square game for all three approaches considered in this paper. In addition, we address the main concern on the number of nonlocal and local resources used. We will first state the main result formally. More explicit numerical comparisons based on these bounds will be presented at the end of this section.
\begin{theorem}
    Given physical error rate $\epsilon\leq\epsilon_0'\approx2.25\times10^{-4}$, if we hope to play the Magic Square Game with success probability $1-\Delta$, $0<\Delta\ll1$. Assuming the physical EPR pairs have infidelity 0.1, let $\chi$ denote the number of EPR pairs used to achieve this.  Then there exists a method that uses an average
    \begin{equation*}
        \chi\leq 2\left\lceil\left(\frac{\log(\Delta/9.64c)}{\log\epsilon/c} \right)^2/0.54\right\rceil
    \end{equation*}
    number of EPR pairs where $c=1/2129.4$. In particular, when $\epsilon=\epsilon_0'$, we further have a lower bound on $\chi$ for this method
    \begin{equation*}
        \chi\geq 2\left\lfloor\left(\frac{\log(\Delta/23.0\mu_0)}{\log\epsilon/\mu_0} \right)^2/0.85\right\rfloor,
    \end{equation*}
    where $\mu_0=1/1061.0$.
    \label{main_thm}
\end{theorem}
Next we establish essential elements for the proof of the theorem. 

\subsection{Failure probability}
Here we will derive bounds on the failure probability of the magic square game when the encoded EPR pairs are prepared using \textit{Direct Encoding Scheme $A$}, \textit{Interface+EPP Scheme $A$} and $B$ discussed previously. There is another component in the circuit, that is, the measurement of the corresponding operators which depends on the row/column the party is given. We want to carry out this non-destructive measurement fault-tolerantly. To achieve this, we will adopt Shor's approach of non-destructive measurement, utilizing fault-tolerantly prepared cat-states, as discussed in Appendix~\ref{Shors'}. However, as they will measure operators such as $\overline{X}\otimes \overline{Z}$, to allow subsequent measurements, a `big' cat state, consisting of 14 qubits(when $k=1$) is needed (Of course, other methods of Shor measurement that save ancillary qubits or use fewer rounds of syndrome measurements can be applied. Since this measurement procedure is the same for all three methods discussed, we will choose the original method and it will serve our purpose). Note that we want Alice and Bob to win the game with probability $\Delta$. We just need to ensure that the row/column assignments\footnote{Recall Table~\ref{msgmmt}} that potentially induce the greatest failure probability is less than $\Delta$, that is when they are both assigned indices 2. A circuit for the FT magic square game is shown below,
\begin{figure}[H]
    \centering
    \includegraphics[width=0.75\linewidth]{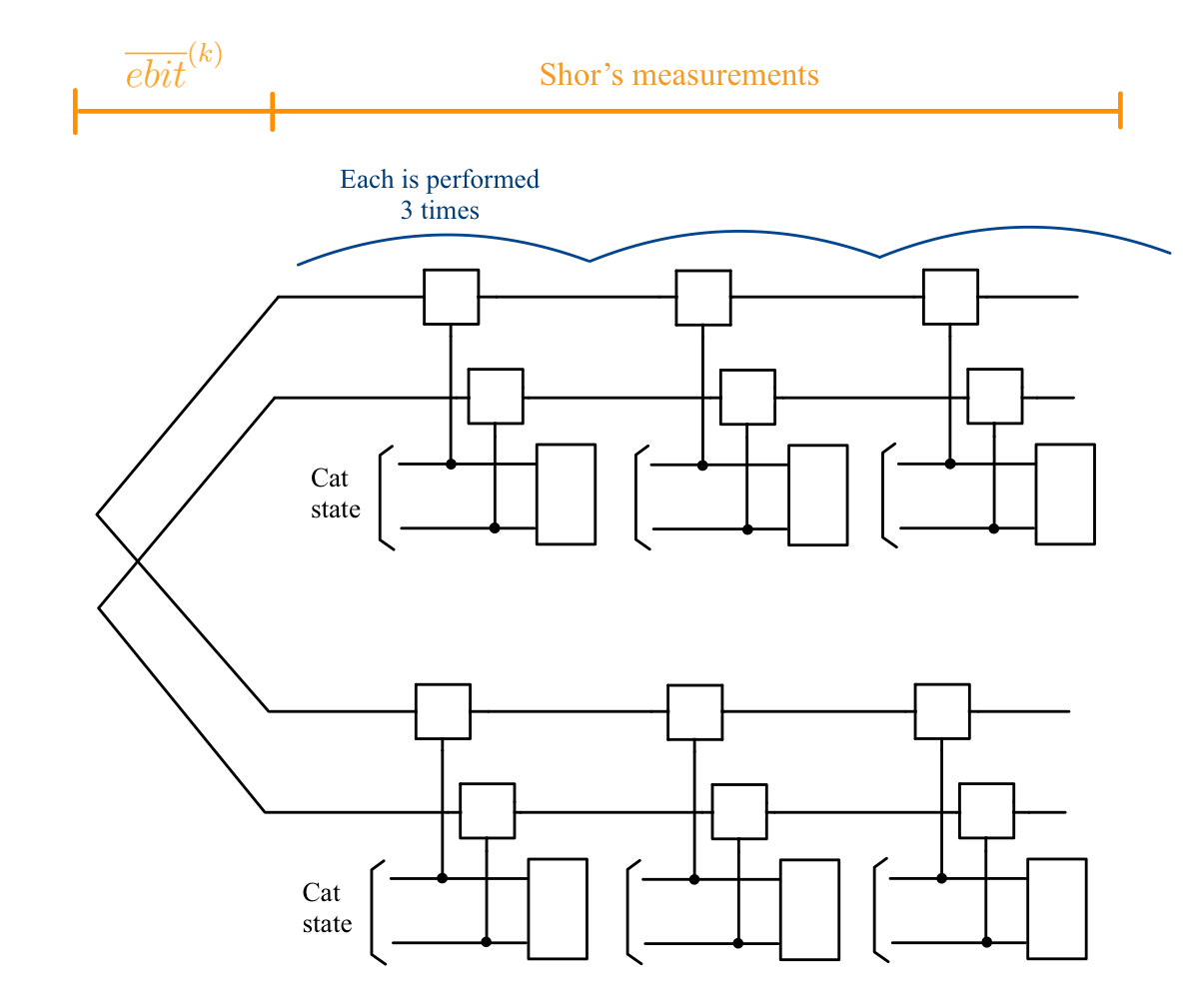}
    \caption{Circuit for FT magic square game}
    \label{FTMSGcircuit}
\end{figure}
In Appendix \ref{Shor-mmt} we compute bounds on one full round of Shor's measurement. Thus the upper bound for the failure probability of the magic square game follows from the union bound
\begin{equation*}
    \mathbb{P}(\text{Magic square game fails})\leq2\cdot\mathbb{P}(\overline{\text{ebit}}^{(k)}\text{ bad})+6\nu_{\text{Shor}}^{(k)}.
\end{equation*}
Similarly, an approximate lower bound can be obtained, with negligible higher-order terms 
\begin{equation*}
    \mathbb{P}(\text{Magic square game fails})\geq2\cdot\mathbb{P}(\overline{\text{ebit}}^{(k)}\text{ bad})+6\mu_{\text{Shor}}^{(k)}.
\end{equation*}
Thus for the three methods, after simplification, we have the bounds for $k\geq2$ and $\epsilon\leq\epsilon_0'$.
    \begingroup
    \setlength{\tabcolsep}{10pt}
    \renewcommand{\arraystretch}{1.35}
\begin{table}[H]
\centering
\renewcommand{\arraystretch}{1.8} 
\begin{tabular}{c|c|c}
 & Upper bound & Lower bound \\
\hline
\textit{Direct Encoding} & $9.42\epsilon_0\left(\epsilon/\epsilon_0 \right)^{2^{k}}$ & $21.6\mu_0\left(\epsilon/\mu_0\right)^{2^k}$ \\
\hline
\textit{Interface+EPP-A} & $9.61\epsilon_0\left(\epsilon/\epsilon_0 \right)^{2^{k}}$ & $23.0\mu_0\left(\epsilon/\mu_0\right)^{2^k}$ \\
\hline
\textit{Interface+EPP-B} & $9.64\epsilon_0\left(\epsilon/\epsilon_0 \right)^{2^{k}}$ & $23.0\mu_0\left(\epsilon/\mu_0\right)^{2^k}$ \\
\end{tabular}
\caption{Bounds for the failure probability of the magic square game under different $\overline{\text{ebit}}$ preparation methods at concatenation level $k$.}
\end{table}
\endgroup
Note that in the last two cases, we also have another parameter $l$, the number of logical EPP performed. Thus the above represents the lowest upper bounds and highest lower bounds achievable with $k$-th level encoding with our analytical methods. Therefore if we hope to play the magic square game with a success probability of at least $1-\Delta$, let us denote $k_0$ to be the minimal number of concatenation levels that could achieve this, we can bound $k_0$ by
\begin{equation*}
    \left\lfloor\log_2\left(\frac{\log\Delta/c_i\mu_0}{\log\epsilon/\mu_0} \right)\right\rfloor\leq k_0\leq\left\lceil\log_2\left(\frac{\log\Delta/d_i\epsilon_0}{\log\epsilon/\epsilon_0} \right)\right\rceil,
\end{equation*}
where $c_i$ and $d_i$ correspond to the leading constants of the lower and upper bounds for each method. 

\subsection{Nonlocal resources}
Here we first compare the non-local resources(noisy ebits) required to prepare one $\overline{\text{ebit}}^{(k)}$ for each method. Since the magic square game requires Alice and Bob to share 2 ebits, we would multiply by 2.
Note that before encoding to logical qubits, we would first perform rounds of physical EPP as in Section~\ref{initial_EPP}. As argued earlier, for \textit{Direct Encoding} we would need to perform at least two rounds to ensure that the initial error is suppressed enough. Let $m$ denote the number of rounds of initial EPPs. Thus for \textit{Direct Encoding}, $m\geq2$. Given the local gates also have an error rate $\epsilon\leq\epsilon_0'$, it's rather hard to work with an analytical form, thus we may resort to numerical tools and set the local physical error rate to be $\epsilon_0'=2.25\times10^{-4}$ in our subsequent analysis. Assuming the infidelity of the initial EPR pair is $10\%$, the success probability of the first round of EPP is $59.6\%$ and the second is $95.6\%$. If we denote the number of EPR pairs required to prepare a $\overline{\text{ebit}}^{(k)}$ by $\chi^{(k)}$, for \textit{Direct Encoding} we have $\forall k\geq2$,
\begin{equation*}
    \chi_{1}^{(k)}=\frac{5}{0.596}\cdot\frac{5}{0.956}\cdot7^k\approx43.9\cdot 7^k.
\end{equation*}
This is actually the exact number of noisy ebits we need in this case since transversal CNOT in the Steane code is the logical CNOT. For the \textit{Interface+EPP} methods, a more meticulous treatment is needed due to the extra parameter $l',l''$ which are determined by both $\epsilon$ and the initial infidelity. We will first obtain lower bounds on the total acceptance probability for the \textit{Interface+EPP} methods. For Scheme $A$, we have the lower bound,
\begin{align*}
    &\mathbb{P}(\text{All }\overline{\text{EPP}}\text{ accepted})\\
    \geq&\left(1-f_2(\epsilon,\sigma_6)\right)\prod_{l=2}^{l'}\left(1-4g_A^{(l-1)}(k,\sigma_6)-12\nu_5^{(k)}\right)\\
    \geq&\left(1-f_2(\epsilon,\sigma_6)\right)\left( 1-\frac{2}{3}\frac{1}{1/(6\kappa_2(2,\sigma_6)\epsilon^2)-1}-12l'\nu_5^{(k)}\right).\\
\end{align*}
For Scheme $B$ it's similar, except when $l\leq k$ we may replace the last term above by $\nu_5^{(l)}$. For the upper bound we will use $1-f_1(\epsilon,\sigma_6)$ in both cases. According to analysis in \ref{logicalerrEPP}, we obtain the following upper and lower bounds on $\chi$. 
\begin{center}
    \textit{Upper bound on $\chi$ for Interface+EPP}
\end{center}
To compute the best bounds on the resource overheads, we need to first determine the relationships between the parameters $k,m$ and $l',l''$. Note that as we have $m=2$ for \textit{Direct Encoding}, in the following we will discuss results for $m\leq2$. In the following table, we show the values of $l'$ for Interface+EPP scheme $A$ and $B$ as in Lemma~\ref{l_prime} for given $k$ and $m$.
\begingroup
\setlength{\tabcolsep}{10pt}
\renewcommand{\arraystretch}{1.35}
\begin{table}[H]
\centering
\renewcommand{\arraystretch}{1.5} 
\begin{tabular}{c|c|c|c}
\multirow{4}{*}{$k=2$}     &       & $A$        & $B$        \\ \cline{2-4} 
                           & $m=0$ & $l'_A=4$   & $l'_B=4$   \\ \cline{2-4} 
                           & 1     & 2          & 3          \\ \cline{2-4} 
                           & 2     & 2          & 2          \\ \hline
\multirow{4}{*}{$k=3$}     &       & $A$        & $B$        \\ \cline{2-4} 
                           & $m=0$ & $l'_A=5$   & $l'_B=4$   \\ \cline{2-4} 
                           & 1     & 3          & 3          \\ \cline{2-4} 
                           & 2     & 2          & 3          \\ \hline
\multirow{4}{*}{$k=4,5$}   &       & $A$        & $B$        \\ \cline{2-4} 
                           & $m=0$ & $l'_A=k+1$ & $l'_B=k+1$ \\ \cline{2-4} 
                           & 1     & $k-1$      & $k$        \\ \cline{2-4} 
                           & 2     & $k-1$      & $k$        \\ \hline
\multirow{4}{*}{$k\geq 6$} &       & $A$        & $B$        \\ \cline{2-4} 
                           & $m=0$ & $l'_A=k+1$ & $l'_B=k$   \\ \cline{2-4} 
                           & 1     & $k-1$      & $k$        \\ \cline{2-4} 
                           & 2     & $k-2$      & $k$        
\end{tabular}
\caption{Results for $l'_A,l'_B$ as in Lemma~\ref{l_prime} for given $k$ and $m$}
\end{table}
\endgroup
Combining the above results we can see that when $k=2$, the minimal upper bound on $\chi_2^{(2)}$ can be achieved with Scheme $A$ and $m=0$. Using this, an upper bound on $\chi$ is 484.2, representing over $77\%$ improvement over \textit{Direct Encoding}.  For $k\geq6$, the optimal scheme would be Scheme $B$ with $m=0$ and the corresponding upper bound on $\chi$ is 
\begin{equation*}
    \chi\leq \lceil 4^k/0.54 \rceil
\end{equation*}
\begin{center}
    \textit{Lower bound on $\chi$ for Interface+EPP}
\end{center}
For the lower bound on $\chi$ we can obtain a similar table when $\epsilon=\epsilon_0'$
\begingroup
\setlength{\tabcolsep}{10pt}
\renewcommand{\arraystretch}{1.35}
\begin{table}[H]
\centering
\renewcommand{\arraystretch}{1.5} 
\begin{tabular}{c|c|c|c}
\multirow{4}{*}{$k=2$}     &       & $A$        & $B$        \\ \cline{2-4} 
                           & $m=0$ & $l''_A=3$   & $l_B''=3$   \\ \cline{2-4} 
                           & 1     & 2          & 2          \\ \cline{2-4} 
                           & 2     & 1          & 2          \\ \hline
\multirow{4}{*}{$k=3$}     &       & $A$        & $B$        \\ \cline{2-4} 
                           & $m=0$ & $l''_A=3$   & $l''_B=4$   \\ \cline{2-4} 
                           & 1     & 2          & 3          \\ \cline{2-4} 
                           & 2     & 2          & 3          \\ \hline
\multirow{4}{*}{$k\geq4$}   &       & $A$        & $B$        \\ \cline{2-4} 
                           & $m=0$ & $l''_A=k$ & $l''_B=k$ \\ \cline{2-4} 
                           & 1     & $k-1$      & $k$        \\ \cline{2-4} 
                           & 2     & $k-2$      & $k$        \\       
\end{tabular}
\caption{Results for $l''_A,l''_B$ as in Lemma~\ref{l_prime} for given $k$ and $m$}
\end{table}
\endgroup
For the lower bounds corresponding to methods discussed in \textit{Upper bounds}, when $k=2,m=0$ with Scheme $A$, on average a lower bound on $\chi_2^{(2)}$ is 74.9. For $k\geq6$ (actually works for $k\geq4$), $m=0$ and Scheme $B$ gives
\begin{equation*}
    \chi_2^{(k)}\geq\lfloor 4^k/0.85\rfloor
\end{equation*}
These bounds are sufficient to establish our theorem. By contrast, if we hope to win the magic square game with probability $1-\Delta$ with \textit{Direct Encoding}, we would have the bounds
\begin{align*}
       87.7\cdot7^{k_1}\leq\chi\leq 87.7\cdot7^{k_2}
\end{align*}
where $k_1 = \left\lceil\log_2\left(\frac{\log\Delta/d_i\epsilon_0}{\log\epsilon/\epsilon_0} \right)\right\rceil$ and $k_2=\left\lfloor\log_2\left(\frac{\log\Delta/c_i\mu_0}{\log\epsilon/\mu_0} \right)\right\rfloor$.
A more explicit numerical comparison will be described in detail later. 
\\

\subsection{Space-time overheads}
One potential drawback of \textit{Interface+EPP} seems to be the use of more local resources. However, we claim that while our method achieves exponential savings on nonlocal resources, the space-time overhead grows only linearly in $k$. Below we will provide estimates and compare the total overhead used by the two parties. Here we first state the assumptions:
\begin{enumerate}
    \item Qubits can be recycled and reset relatively easily and quickly.
    \item Qubits can undergo parallel operations, with each qubit engaged in a single operation simultaneously.
    \item Qubits have coherence time long enough to allow consecutive EPPs.
\end{enumerate} 
For \textit{Direct Encoding}, to prepare one logical EPR pair, we have 2 code blocks and an additional 2 for error correction. Observing the structure for \textit{Interface+EPP}, the largest overhead comes from the logical EPP part, which consists of 8 data blocks and 8 ECs. Accounting for the overall EPP acceptance probability, \textit{Interface+EPP} uses at most 7.5 times more qubits for all $k$. \\

Now we estimate the time overhead for both methods. Let $N$ denote the time-overhead, we will first consider $N$ for preparing $|\overline{0}\rangle^{(k)}$. Since transversal gates can be implemented in parallel and accounting for ancilla rejection probability, we have $N_0^{(k)}= a^{k-1}+7\cdot\frac{a^k-1}{a-1}$ where $a=1.0024$. For $k\leq50$, $N_0^{(k)}\approx7k+1$. By examining the Steane EC, we have $N_{EC}^{(k)}=2N_0^{(k)}+4$. Thus for \textit{Direct Encoding}, $N_{DE}^{(k)}= 2N_{EC}^{(k)}+N_0^{(k)}+1\approx35k+14$. 
For \textit{Interface+EPP-A}, note that preparation of $|\Omega\rangle^{(k)}$ can be done in parallel. Thus the interface uses approximately $\left( N_0^{(k)}+11\right)+\sum_{l=1}^kN_{EC}^{(l)}+2k$ timesteps. The EPP procedure (regardless of rejection) uses $3N_{EC}^{(k)}+2$ timesteps. Therefore given $m=0$ and $k\leq50$ accounting for rejection, the average time overhead after $l'$ rounds of EPP is $1.89l'_A(7k^2+50k+26)=1.93(k+1)(7k^2+50k+26)$. For Scheme $B$, since we have interface-EPP alternately to encode to level-$k$, the time overhead is $1.89\sum_{l=1}^k(N_0^{(l)}+13+3N_{EC}^{(l)}+2)=0.945(49k^2+115k)$. Thus Scheme $B$ outperforms Scheme $A$ in terms of time-overhead and it is only a linear increase in $k$ compared to \textit{Direct Encoding}. Together with the analysis on space-overhead, our claim is confirmed.

\subsection{Numerical results}\label{Numerical_result}
Firstly, we will provide a full stabilizer circuit Monte Carlo simulation comparing the two methods for $k=1$. Numerically we determined the pseudo-threshold for CNOT-exRec, the largest gadget in the circuit, to be $4.90\times10^{-4}$. As argued before, to guarantee fault-tolerance, we would need the gate-teleported CNOT to have an error rate below this value. Therefore for a fair comparison, we restrict to $\epsilon\leq2.42\times10^{-4}$ (although Interface+EPP tolerates higher error rates). Since the final basis measurements in FT magic square game are the same regardless of $\overline{\text{ebit}}$ preparation, we will only present the logical error rate against $\epsilon$ for preparing $\overline{\text{ebit}}^{(1)}$. We describe the two specific methods we are comparing
\begin{enumerate}
    \item For Direct Encoding, we perform two rounds of 5-to-1 physical EPP followed by the gate-teleported encoding.
    \item For Interface+EPP, at $k=1$, the two schemes are the same. For this method, while it can operate with $m=0$ or $m=1$ and save more initial ebits, we will present an optimal method~\footnote{Optimal in the sense that, the logical error rate of the resulting $\overline{\text{ebit}}^{(1)}$ is comparable to Direct Encoding while the ebit overhead and local space-time overhead are suitably minimized.}, that is, performing two rounds of 4-to-1 EPP~\footnote{This 4-to-1 EPP is the same as in Figure~\ref{Interface+EPP} at the physical level.} followed by one round of Interface+$\overline{\text{EPP}}^{(1)}$. 
\end{enumerate}
We note that in Direct Encoding we use the 5-to-1 EPP instead of 4-to-1 because 5-to-1 is more powerful and can reduce the infidelity to below the threshold value in 2 rounds of initial EPP, whereas the 4-to-1 EPP takes 3 rounds. Thus using 5-to-1 EPP will save ebits. 

In Figure \ref{Comparison}, we show the results. Figure~\ref{LE_comparison} shows the logical error rates of $\overline{\text{ebit}}^{(1)}$ resulting from the two preparation methods. Figure~\ref{Failure_prob} demonstrates the overall success probability for Interface+EPP. We run $10^6$ rounds for each value of $\epsilon$ and 5 iterations. Multiprocessing was used to maximize performance. The mean and standard deviation are also exhibited by the error bars. 
\begin{figure}[H]
    \centering
    \begin{subfigure}[b]{0.55\textwidth}
        \centering
        \includegraphics[width=\textwidth]{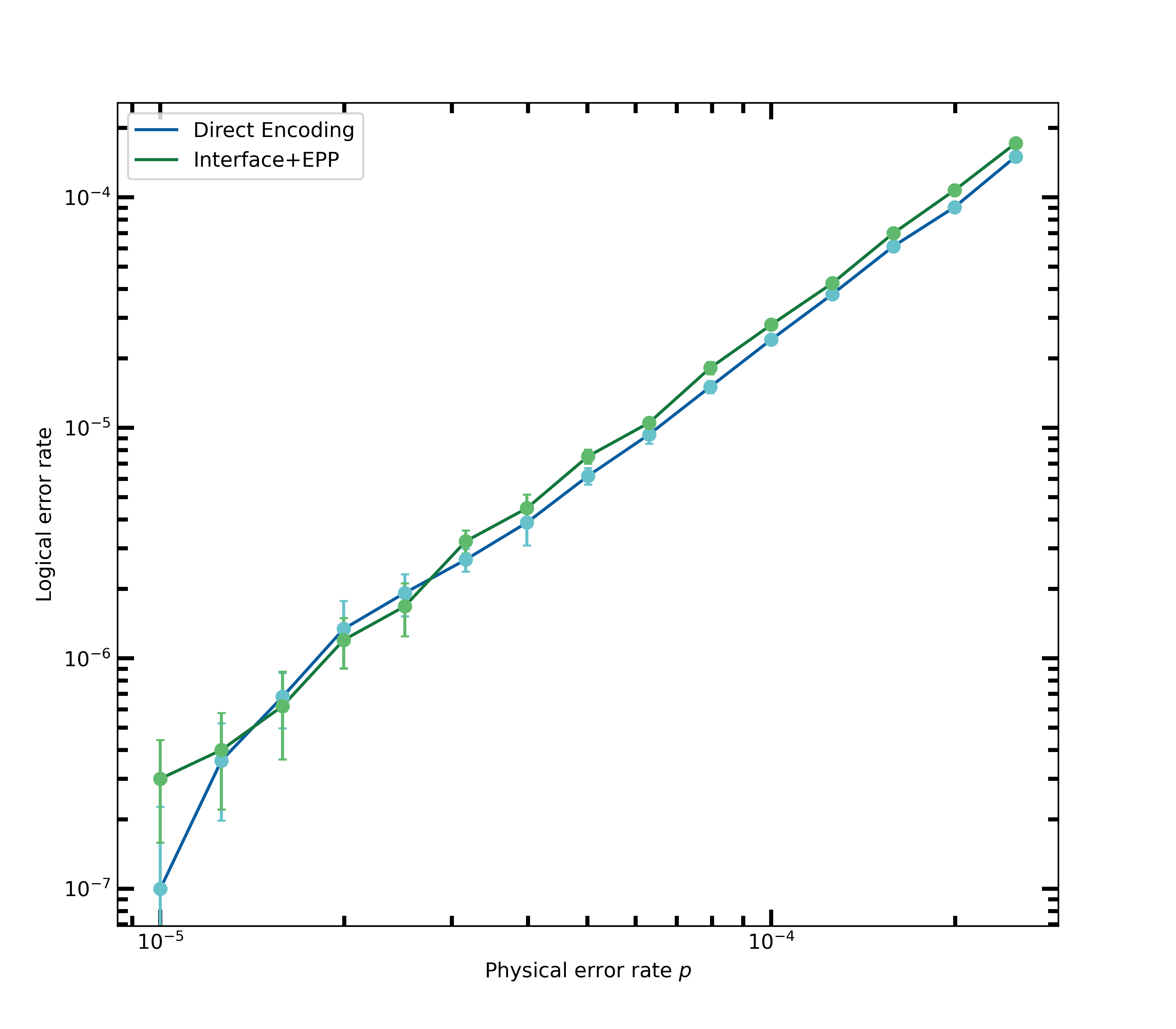}
        \caption{}
        \label{LE_comparison}
    \end{subfigure}
    \vfill
    \begin{subfigure}[b]{0.55\textwidth}
        \centering
        \includegraphics[width=\textwidth]{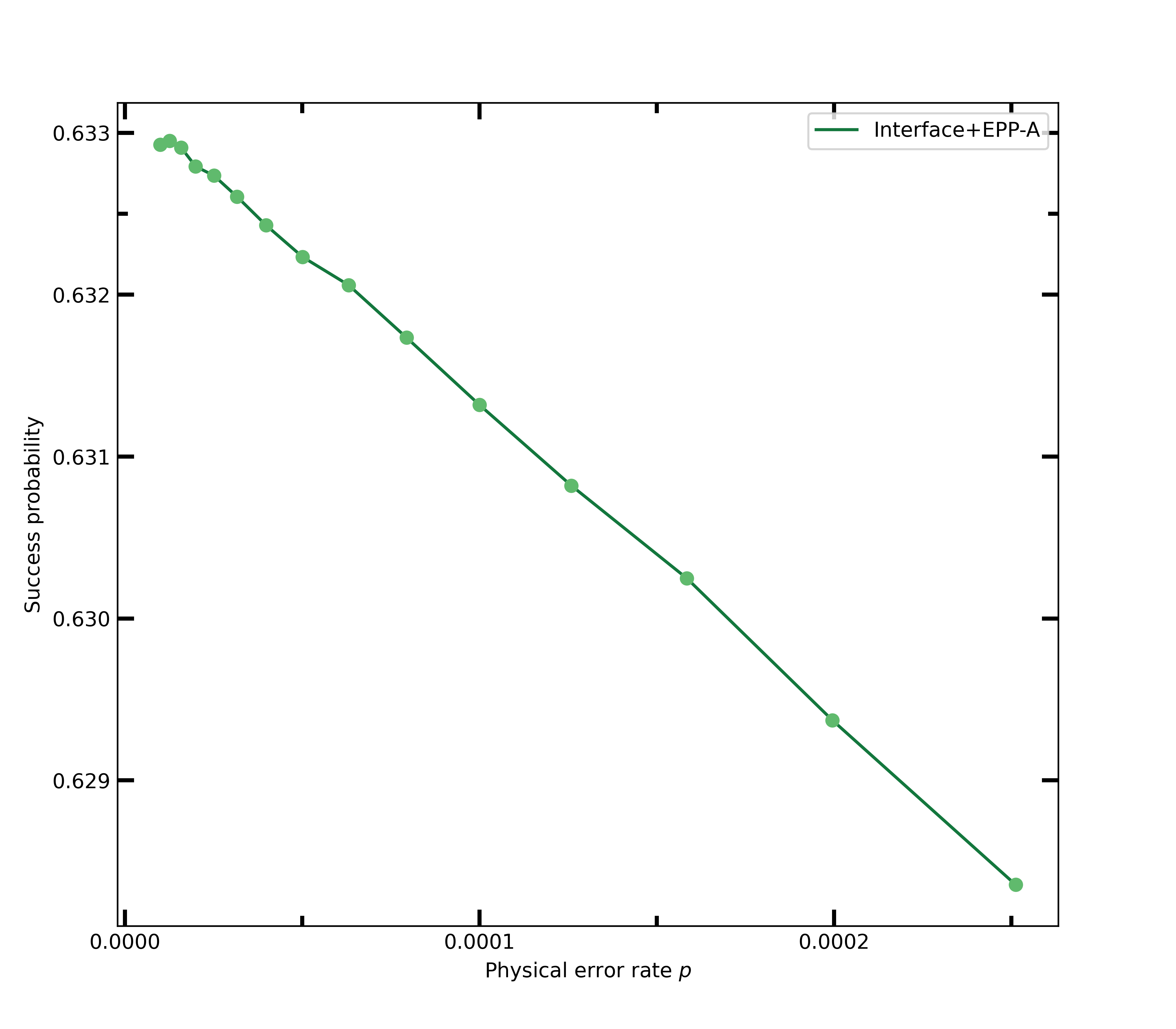}
        \caption{}
        \label{Failure_prob}
    \end{subfigure}
    \caption{(a) Logical error rates of $\overline{\text{ebit}}^{(1)}$ prepared with \textit{Direct Encoding} and \textit{Interface+EPP} against physical error rate $p$. Both methods yield similar logical error rates, thereby justifying a subsequent comparison of resource overheads. (b) Overall success probability against $\epsilon$ for \textit{Interface+EPP}.}
    \label{Comparison}
\end{figure}
From the first figure, we observe that both methods exhibit similar effectiveness in reducing the logical error rate. As shown in Figure~\ref{LE_comparison}, even accounting for statistical error, \textit{Direct Encoding} performs only slightly better than \textit{Interface+EPP}. Since the logical error rates are comparable, it is meaningful to compare the ebit overheads. In Figure~\ref{Failure_prob}, we show the overall success probability of the \textit{Interface+EPP} method against $\epsilon$. Now we consider the average number of ebits needed at $\epsilon=2.42\times10^{-4}$. For \textit{Direct Encoding}, we use $\frac{5^2\times 7}{0.57}\approx 301.7$ ebits. For \textit{Interface+EPP}, we use $\frac{4^3}{0.63}\approx101.6$ ebits, representing a $66.3\%$ saving relative to \textit{Direct Encoding}. This provides numerical evidence for our claims and demonstrates the advantage of our schemes even for near-term implementation. 

Finally, we present results for the Magic Square Game. In Figure~\ref{FTMSG_result}, by using the bounds obtained previously, we compare the two methods in terms of the average number of raw ebits required if we hope to win the magic square game with probability $1-\Delta$. In the following, for Interface+EPP, we plot for Scheme $B$, which we proved to have better performance. 
\begin{figure}[H]
    \centering
    \includegraphics[width=0.7\linewidth]{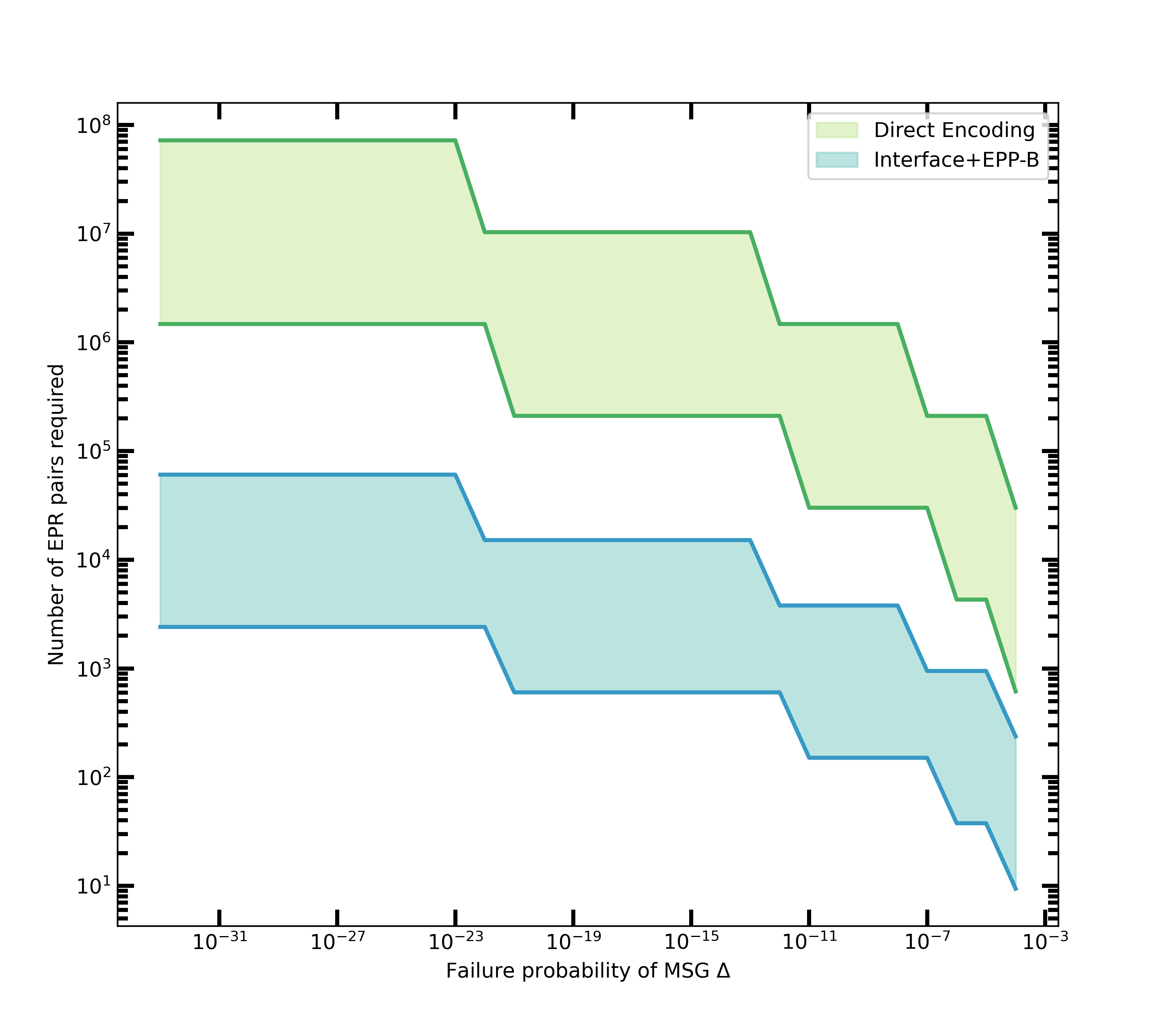}
    \caption{The number of raw ebits required against failure probability of the magic square game $\Delta$. Log-scale is used for both axes.}
    \label{FTMSG_result}
\end{figure}
From the above figure, we can conclude that Interface+EPP-$B$ is `definitely' more superior than Direct Encoding for $\Delta\leq10^{-3}$. As log-scale is used, the advantage of \textit{Interface+EPP-}$B$ increases significantly as $\Delta$ decreases further. Moreover, by comparing the lower bound of \textit{Direct-Encoding} and the upper bound of \textit{Interface+EPP-B}, we can see that for $10^{-33}\leq\Delta\leq10^{-3}$, the least improvement is at least $61.4\%$. As $\Delta$ decreases, the improvement reaches at least $95.9\%$ for $\Delta \leq 10^{-23}$. The disparity potentially becomes even more evident as $\Delta$ decreases further. 
\section{Conclusion}\label{conclusion}
In this paper, we investigate the preparation of long-range $\overline{\text{ebit}}$s in the context of the Magic Square Game. Specifically, we aim to minimize the number of consumed ebits while ensuring that the failure probability does not exceed $\Delta$.

To this end, we introduce a novel approach that leverages an *interface*, which translates an unknown state into the logical space, along with the purifying power of entanglement purification protocols (EPP). For the $[[7^k,1,3^k]]$ concatenated Steane code, our analytical bounds indicate that the \textit{Interface+EPP} scheme offers a significant advantage over the conventional method—encoding qubits on both sides and performing transversal gates—in terms of reducing the number of required initial ebits. This advantage becomes even more pronounced as $\Delta$ decreases. Moreover, for the case of $k=1$, we performed full-circuit simulations to support our claims. On the nonlocal game side, our analysis provides an upper bound on the number of ebits required to play the Magic Square Game fault-tolerantly using the Steane code. Future research directions include tightening these upper bounds and establishing a lower bound, though the method for deriving the latter remains unclear.

More broadly, our proposal presents a flexible and comprehensive framework that can be adapted for similar tasks. Many aspects of our construction are customizable to suit different needs. For instance, in Figure~\ref{Interface+EPP}, instead of using destructive measurements for $\overline{\text{EPP}}$, Alice and Bob could further reduce local qubit consumption and initial ebit requirements by performing non-destructive $\overline{X}/\overline{Z}$ measurements. The final three $\overline{\text{ebit}}$s could then be repurposed for subsequent $\overline{\text{EPP}}$. However, practical implementation is constrained by qubit decoherence times ($T_1/T_2$), making some hardware platforms more suitable than others~\cite{Majidy_Wilson_Laflamme_2024}. Additionally, the choice of EPP can be varied. In this work, we primarily employ a 4-to-1 $\overline{\text{EPP}}$ scheme capable of correcting both $\overline{X}$- and $\overline{Z}$-errors. More powerful EPP protocols could be used at the cost of additional ebits, while in certain special cases, fewer ebits may suffice. For example, in systems with biased noise~\cite{Tuckett_2019,xu2022tailored}, only two $\overline{\text{ebit}}$s might be required to purify either $X$- or $Z$-errors.

Our framework is also extendable to other quantum error-correcting codes (QECCs), provided a suitable interface is constructed. For instance, local overhead can be reduced by employing concatenated quantum Hamming codes alongside a carefully designed interface~\cite{Yamasaki_2024}. However, one challenge with this approach is that, when applying transversal CNOT gates locally, it is not possible to target individual logical qubits—rather, the operation must be performed on entire blocks of logical qubits simultaneously. Another potential extension involves quantum low-density parity-check (qLDPC) codes. Since lattice surgery has been adapted for qLDPC codes to enable fault-tolerant gate implementation~\cite{Cohen_2022}, a similar interface could be designed for specific qLDPC codes. Adapting our framework to other QECCs would allow us to leverage techniques developed for those codes. For instance, single-shot error correction could eliminate the need for repeated measurements~\cite{Bomb_n_2015,gu2023singleshot}, while adaptive syndrome measurement techniques could reduce the number of measurement rounds when using higher-distance codes~\cite{Tansuwannont_2023,delfosse2020short}.

From a practical standpoint, our protocol for the Steane code requires all-to-all connectivity, making experimental verification feasible on certain platforms such as neutral atoms~\cite{Bluvstein_2023} and trapped ions~\cite{cieran}. Furthermore, our framework enables interfacing between different codes without significantly compromising the logical error rate. We hope our work will inspire further research in this direction, which is crucial not only for fault-tolerant nonlocal games but also for the advancement of modular quantum architectures and the quantum internet.

\section{Acknowledgement}
We acknowledge computing resources from IQC. We thank Xiao Yang for the general maths discussion and Amit Anand for the discussion on the dynamical systems. This work was done when Zeyi Liu was a Masters student at University of Waterloo.

\section{Remarks}
\textit{During the course of writing-up, a few works using the idea of logical EPP to prepare logical Bell pairs have surfaced~\cite{pattison2024fastquantuminterconnectsconstantrate,gu2025constantoverheadentanglementdistillation,ataides2025constantoverheadfaulttolerantbellpairdistillation}. However, we approach the problem from a different perspective and we are hoping our rigorous theoretical analysis of error bounds will provide values to a broader audience. From an error-correction standpoint, it will be insightful to compare our Interface+EPP protocol with the Direct Encoding method proposed in the aforementioned works in the context of other error-correcting codes. Notably, our interface leverages a pre-prepared ancilla and requires significantly fewer measurement rounds in the actual distillation protocol. This characteristic may be advantageous for physical platforms such as neutral atoms, where ancilla patches can be prepared in separate zones.}

\bibliographystyle{unsrt}
\bibliography{ref}

\onecolumn\newpage
\appendix
\section{Fault-tolerant Constructions}\label{FT_construction}
\subsection{Steane's fault-tolerant error correction}
The QECC we use in this work is the basic Steane [7,1,3] code. It has stabilizers
\begin{align*}
    IIIXXXX & \;\; IIIZZZZ \\
    IXXIIXX & \;\; IZZIIZZ \\
    XIXIXIX & \;\; ZIZIZIZ
\end{align*}
and logical operators $\overline{X}=X^{\otimes7}$, $\overline{Z}=Z^{\otimes7}$
The property that it's a perfect code makes it intriguing and this implies a simple decoding procedure and the fact that it's a CSS code implies that transversal CNOT can be implemented with CNOT$^{\otimes n}$. Besides, by code concatenation of level $l$ we can construct codes of distance $3^l$ that can handle errors with weight up to $\frac{3^l-1}{2}$. In this subsection we set forth the specific constructions for the Steane code used in the work and some pseudo-threshold results.
\begin{figure}[H]
    \centering
    \begin{tikzpicture}
    \node[scale=0.5]
    {
    \begin{quantikz}
    \lstick{$\ket{+}$}  & \qw & \ctrl{4} & \qw & \ctrl{2} &\qw &\qw &\qw \\
    \lstick{$\ket{+}$}  & \ctrl{1} & \qw & \ctrl{4} & \qw &\qw &\qw &\qw \\
    \lstick{$\ket{0}$}  & \targ{} &\qw & \qw & \targ{} &\qw &\ctrl{4} &\ctrl{5} &\qw\\
    \lstick{$\ket{+}$}  & \ctrl{2} & \qw & \qw & \ctrl{3} &\ctrl{1}&\qw &\qw &\qw\\
    \lstick{$\ket{0}$}  & \qw & \targ{} & \qw & \qw & \targ{} &\qw &\qw & \ctrl{3}&\qw\\
    \lstick{$\ket{0}$}  & \targ{} &\qw &\targ{} &\qw &\qw &\qw &\qw &\qw &\ctrl{2} & \qw \\
    \lstick{$\ket{0}$}  & \qw &\qw & \qw &\targ{} &\qw &\targ{} &\qw &\qw &\qw &\qw \\
    \lstick{$\ket{0}$}  & \qw &\qw &\qw &\qw &\qw &\qw  &\targ{} &\targ{} &\targ{} &\meter{$X$} \\
    \end{quantikz}
    };
    \end{tikzpicture}
    \caption{Fault-tolerant level-1 $|\Bar{0}\rangle$ encoding circuit for the [[7,1,3]] code. }
    \label{prep_anc1}
\end{figure}
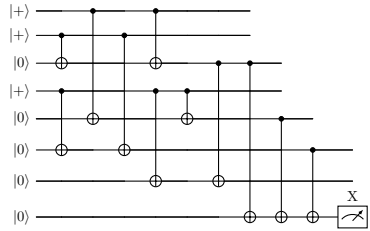
Figure \ref{prep_anc1} presents the circuit for preparation of $|\Bar{0}\rangle$. Note that there is an additional ancilla qubit which is used for catching bit flip errors occurred during preparation. $Z$ error verification is not needed because all $Z$ errors can be reduced to weight 0 or 1, thus obeying the 'goodness' criteria in the fault-tolerant section. In this way the encoding circuit is fault-tolerant. $|\overline{+}\rangle$ can be analogously prepared, by reversing the CNOT gates and interchanging $|0\rangle$ and $|+\rangle$. 

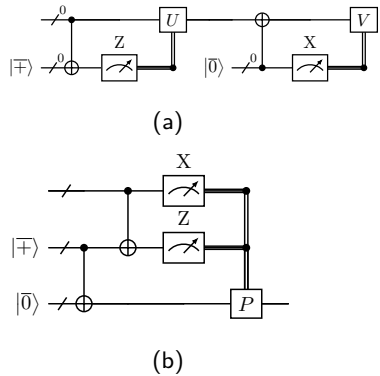
\begin{figure}[H]
    \centering
    \begin{subfigure}[b]{0.3\textwidth}
    \begin{tikzpicture}
    \node[scale=0.6]
    {
    \begin{quantikz}
        \qw & \ctrl{1}\qwbundle{0} &\qw  & \gate{U} &\qw&\qw &\targ{0}\vqw{1} & \qw  & \gate{V}\\
        \lstick{$\ket{\overline{+}}$} & \targ{0}\qwbundle{0} & \meter{$Z$} & \cwbend{-1} &  \wireoverride{n} &\wireoverride{n}\lstick{$\ket{\overline{0}}$} &\ctrl{0}\qwbundle{0} &\meter{$X$} & \cwbend{-1}
    \end{quantikz}
    }; 
    \end{tikzpicture}
    \subcaption{}
    \end{subfigure}
    
    \begin{subfigure}[b]{0.3\textwidth}
    \begin{tikzpicture}
    \node[scale=0.7]
    {
    \begin{quantikz}
        \qw &\qw\qwbundle{} & \ctrl{1} & \meter{$X$} & \cwbend{2} \\
        \lstick{$\ket{\overline{+}}$} & \ctrl{1}\qwbundle{} & \targ{} & \meter{$Z$} &\cwbend{1} \\
        \lstick{$\ket{\overline{0}}$} & \targ{}\qwbundle{} & \qw &\qw &\gate{P} & \qw
    \end{quantikz}
    };  
    \end{tikzpicture}
    \subcaption{}
    \end{subfigure}
    \caption{(a) Steane's FTEC gadget. $U$ and $V$ are corrections corresponding to measured syndromes. (b) Knill's FTEC gadget using transportation. The applied gate $P$ depends again on the syndrome measurements.}
    \label{EC_gadget}
\end{figure}
We shall also present numerical estimates for the pseudo-threshold of different gadgets used in our work ($k=1$):
\begingroup
\begin{table}[H]
\centering
\begin{tabular}{c|c|c|c|c}
 Encoding & Steane EC & Knill EC & Logical CNOT & CNOT-exRec \\
\hline
0.058 & $4.5\times10^{-3}$ & $6.1\times10^{-3}$ & $8.5\times10^{-2}$ & $4.9\times10^{-4}$
\end{tabular}
\end{table}
\endgroup
Note that the thresholds obtained here are slightly higher than those reported in previous literature due to different state preparation procedures. We also tested out other FTEC schemes for the Steane code, such as Shor's\cite{Shor} and the Flag\cite{RuiChao} methods, but they turn out to have lower pseudo-thresholds so we won't elaborate here. Also we use Steane EC for easier analytical threshold analysis. For flag EC, due to the fact that multiple rounds of measurements are required, conditioned on the flag values, it's hard to count the number of malignant pairs. But the flag EC method does significantly reduce resource overhead and is potentially useful if the number of available qubits is limited. Another observation is that if we wish to implement a circuit with \textit{exRecs} as introduced in \cite{AGP05}, the threshold will be dominated by the EC gadgets whereas if we are concerned with a short circuit, e.g. preparing an EPR pair, without the ECs in between, the threshold is dominated by the encoding.

\subsection{Shor's style syndrome measurement}\label{Shors'}
When we perform non-destructive measurements in this work, we resort to Shor's style syndrome measurement. The following is an illustration adopted from \cite{AGP05},\cite{gottesman2009introduction}. To make the preparation of the cat state fault-tolerant, we have an ancilla qubit. From the measurement results, we can infer the parity (i.e. eigenvalue of $\overline{X}$). So when the outcome $x$ has even weight, we say the eigenvalue is $+1$, otherwise it's $-1$. This is valid by noting the fact that
\begin{align*}
H^{\otimes n}(|0\rangle^{\otimes n}+|1\rangle^{\otimes n})&\propto\sum_{\text{wt$(x)$=even}}|x\rangle\\
H^{\otimes n}(|0\rangle^{\otimes n}-|1\rangle^{\otimes n})&\propto\sum_{\text{wt$(x)$=odd}}|x\rangle
\end{align*}
\begin{figure}[H]
    \centering
    \begin{tikzpicture}
    \node[scale=0.5]
    {
    \begin{quantikz}
      & \qw & \qw & \qw & \qw &\qw &\targ{0}\vqw{7} &\qw &\qw &\qw &\qw &\qw &\qw &\qw \\
      & \qw & \qw & \qw & \qw &\qw &\qw &\targ{0}\vqw{7} &\qw &\qw &\qw &\qw &\qw &\qw \\
      & \qw &\qw & \qw & \qw &\qw &\qw &\qw &\targ{0}\vqw{7}&\qw &\qw &\qw &\qw &\qw \\
      & \qw & \qw & \qw & \qw &\qw&\qw &\qw &\qw &\targ{0}\vqw{7} &\qw &\qw &\qw &\qw \\
      & \qw & \qw & \qw & \qw & \qw &\qw &\qw & \qw &\qw &\targ{0}\vqw{7} &\qw &\qw &\qw \\
      & \qw &\qw &\qw &\qw &\qw &\qw &\qw &\qw &\qw &\qw &\targ{0}\vqw{7} &\qw &\qw \\
      & \qw &\qw & \qw &\qw &\qw &\qw &\qw &\qw &\qw &\qw &\qw &\targ{0}\vqw{7}&\qw \\
      \lstick{$|0\rangle$}&\qw &\qw & \qw &\targ{0}\vqw{1} &\ctrl{7}&\ctrl{0} &\qw &\qw &\qw &\qw &\qw &\qw &\meter{$X$}\\
      \lstick{$|0\rangle$}&\qw & \qw &\targ{0}\vqw{1} & \ctrl{0} &\qw&\qw&\ctrl{0} &\qw &\qw &\qw &\qw &\qw &\meter{$X$} \\
      \lstick{$|0\rangle$}&\qw & \targ{0}\vqw{1} &\ctrl{0} &\qw &\qw&\qw&\qw &\ctrl{0} &\qw &\qw &\qw &\qw &\meter{$X$}\\
      \lstick{$|+\rangle$}&\ctrl{1} & \ctrl{0} &\qw &\qw &\qw &\qw &\qw &\qw &\ctrl{0} &\qw &\qw &\qw &\meter{$X$}\\
      \lstick{$|0\rangle$}&\targ{0}&\ctrl{1} &\qw &\qw &\qw &\qw &\qw &\qw &\qw &\ctrl{0} &\qw &\qw &\meter{$X$} \\
      \lstick{$|0\rangle$}&\qw &\targ{0} &\ctrl{1} &\qw &\qw &\qw &\qw &\qw &\qw &\qw &\ctrl{0} &\qw &\meter{$X$}\\
      \lstick{$|0\rangle$}&\qw &\qw  & \targ{0} &\ctrl{1} &\qw &\qw &\qw &\qw &\qw &\qw &\qw &\ctrl{0} &\meter{$X$} \\
      \lstick{$|0\rangle$}&\qw &\qw &\qw & \targ{0} &\targ{0} & \meter{$Z$}\\
    \end{quantikz}
    };
    \end{tikzpicture}
    \label{Stateprep}
    \caption{Fault-tolerant Shor's style measurement of $\overline{X}$}
\end{figure}
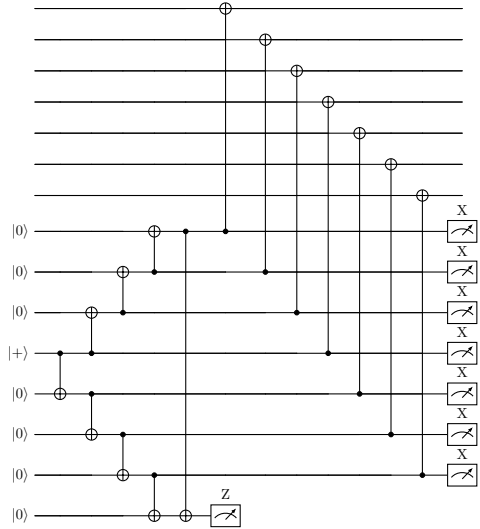
On the other hand, if we hope to measure parity of $\overline{Z}$ we replace the transversal CNOT by transversal CZ. However, sometimes a single qubit error from the ancilla can still cause parity change. So syndrome measurement is repeated multiple times. In the case of Steane code, it suffices to repeat measurements at most 3 times. If the same parity is measured twice consecutively, then we accept the parity and there is no need to perform the third round of measurement, otherwise 3 rounds are measured and we
take the majority as the result (with EC gadget on the data block before and after each measurement). Through this repeated measurement, no one fault can cause logical error in the data block or lead to rejection, thus satisfying fault-tolerant criteria.
\section{Proof of Theorem}\label{thres_thm_proof}
We shall proceed by analyzing the bounds for the logical error rate of the CNOT-exRec and other exRecs at $k=1$ and then generalize to higher-level encoding according to their interdependence. It will be clear later why this is necessary instead of generalizing each exRec independently. For the CNOT-exRec, there are $\gamma_5=279$ locations. For the upper bound, if we start with second-order terms, we may enumerate the number of malignant pairs $\alpha$. For a tighter upper bound and also a rigorous lower bound, not only must we consider the probability of malignant pairs, but we also need to take into account the type of faults($X/Z$) causing the logical error. For example, if a $Loc_1$ and a $Loc_5$ form a malignant pair, then $IZ/ZI/ZZ$ error following $Loc_5$ will not cause a logical error. For the lower bound, we observe that if two locations of a malignant pair are faulty and every other location is fault-free, then this must result in a logical error. \\

 The joint probability $\varepsilon_{5.\text{joint}}^{(k)}$ of failure for the CNOT $k$-exRec given acceptance of all ancillas (from the circuit in Appendix~\ref{FT_construction}, every $|\overline{0}\rangle$ or $|\overline{+}\rangle$ state we use in EC gadget has a verifying ancilla to ensure fault-tolerance) can be generally expressed as 
\begin{align*}
    &\sum_{j\leq i=1}^5\alpha_{5}(i,j)\epsilon_i^{(k-1)}\epsilon_j^{(k-1)}(1-\max\{\epsilon_i^{(k-1)},\epsilon_j^{(k-1)}\})^{\gamma_5-2}\\
    &\leq \varepsilon_{5,\text{joint}}^{(k)}\leq\sum_{j\leq i=1}^5\alpha_{5}(i,j)\epsilon_i^{(k-1)}\epsilon_j^{(k-1)}+\beta(\epsilon_{\max}^{(k-1)})^3
    \label{DE_bound}
\end{align*}
where $\alpha_5$ is the malignant pair matrix(MPM) with $\alpha_{5}(i,j)$ denote the number of pairs of locations of type $i$ and $j$ where faults can cause logical errors, given that all ancillas in the exRec are accepted. $\epsilon_6$, the error rate of CNOT gates across two data blocks is addressed separately. The rationale for this approach will become evident in the subsequent section. Note that each pair of locations is run $10^3$ number of times and the average is taken.
\begin{equation*}
 \alpha_5=
    \begin{pmatrix}
43 &  &  & & &  \\ 
 0&  43 &   &   &   &  \\
 0. &  176.0 &  168.0 &   &  &      \\
 176.0  &  0 &   0 &  168.0 &  &    \\
 299.2 & 298.7 & 570.0 & 572.8 & 793.4 &  \\
  49.2 &  50.0 & 100.0 & 100.5 & 251.0 & 12.8\\
   0 & 0 & 0 & 0 & 0 & 0 & 0

    \end{pmatrix}
\end{equation*}
There are some details worth commenting, 
\begin{enumerate}
    \item Notice that in $\alpha_5$, $\alpha_5(1,1)=\alpha_5(2,2)$ and $\alpha_5(3,3)=\alpha_5(4,4)$. This is no coincidence. These two pairs of locations in this exRec are in some sense "symmetrical" to each other. 
    \item We observe that $\alpha_5(2,1)=0$, this is because if the pair consists of a $Loc_1$ and a $Loc_2$. $Loc_1$ can only induce $X$-error and $Loc_2$ can only induce $Z$-error. The EC we use is capable of correcting one $X$ and one $Z$ error. Therefore they won't add up to a logical error. The same applies to other 0 entries.
    \item For the CNOT-exRec we treat $Loc_5$ and $Loc_6$ separately. If later we use CNOT-exRecs in recursive simulation or local operations, we can simply set $\epsilon=\epsilon_6$.
\end{enumerate}
If we further consider $\mathcal{O}(\epsilon^3)$ terms, if there are three faults in the circuit, this could potentially result in a logical error. To obtain a tighter upper bound, we break down the various types of locations. There are $n_5=[32, 32, 32, 32, 144, 7,0]$ locations of each type and we hope to exclude the cases where a malignant pair is among the three locations since they are accounted for in the prior term. We have Lemma~\ref{third_order_thm} for an upper bound on the constant of $\mathcal{O}(\epsilon^3)$.
The proof of the lemma is left in Appendix \ref{third-order}. This result will be used throughout the paper where similar scenarios occur. We will further mention the specification of one component, the EC gadget. For the EC gadget, for $k=1$, there are 186.4 malignant pairs and for third-order terms, applying the theorem gives an upper bound of $F=8847.5$. 
\\
Now by Bayes' rule, $\varepsilon_5^{(k)}=\frac{\varepsilon_{5,\text{joint}}^{(k)}}{\mathbb{P}(\text{ancillas accepted})}$. Since we have 8 ancilla qubits, if we let $\mathbb{P}_{|\overline{0}\rangle,\text{accept}}^{(k)}$ denote the probability of a level-$k$ $|\overline{0}\rangle$ or $|\overline{+}\rangle$ being accepted, we have
\begin{equation*}
    \varepsilon_{5}^{(k)}=\left(\mathbb{P}_{|\overline{0}\rangle,\text{accept}}^{(k)}\right)^{-8}\varepsilon_{5,\text{joint}}^{(k)}
\end{equation*}
To establish an upper bound on $\varepsilon_5^{(k)}$, we need a lower bound on $\mathbb{P}_{|\overline{0}\rangle,\text{accept}}^{(k)}$. Through simulation, we enumerate on average 10.8 fault locations that will cause rejection of the data block. So 
\begin{equation*}
1- C\varepsilon_5^{(k-1)}\leq\mathbb{P}_{|\overline{0}\rangle,\text{accept}}^{(k)}\leq1
\end{equation*}
where $C=10.8$. Putting all the above together, at $k=1$, we will have the bounds for $\varepsilon_{5}^{(1)}$,
\begin{align*}
 &1312.6\epsilon^2(1-\epsilon)^{\gamma_5-9}(1-\epsilon_6)^7\\
 &+330.9\sigma_6\epsilon^2(1-\epsilon)^{\gamma_5-8}(1-\epsilon_6)^6+12.8\sigma_6^2\epsilon^2(1-\epsilon)^{\gamma_5-7}(1-\epsilon_6)^5\\
 \leq&\varepsilon_{5}^{(1)}\\
 \leq& (1-10.8\epsilon)^{-8}\cdot \dots\\
 &\dots\bigg(1312.6\epsilon^2+(12.8\sigma_6+330.9)\sigma_6\epsilon^2 +734691.4\epsilon^3\\
 &+(57999.0\sigma_6+3079.0\sigma_6^2+13.6\sigma_6^3)\epsilon^3 \bigg)
\end{align*}
Invoking Prop \ref{prop_bound} we will obtain $A_5^{(1)}$ and $B_5^{(1)}$ such that
\begin{equation*}
    A_5^{(1)}\epsilon^2\leq \varepsilon_5^{(1)}\leq (1-C\epsilon)^{-8}B_5^{(1)}\epsilon^2
\end{equation*}
which holds for $\epsilon\leq1/B_5^{(1)}$. In particular, for $\sigma_6=1$, e.g. in local CNOT-exRec, we have $A_5^{(1)}=1431.4,B_5^{(1)}=2038.9$. In this case, we can further simplify the upper bound since $\varepsilon_5^{(1)}\leq(1-C/B_5^{(1)})^{-8}B_5^{(1)}\epsilon^2\leq D_5^{(1)}\epsilon^2$ where $D_5^{(1)}=2127.4$. We denote the lower and upper bounds by $\mu_5^{(1)}$ and $\nu_5^{(1)}$ respectively. In Appendix \ref{tight}, we compare the above level-1 bounds with the actual logical error rates obtained from simulation and the bound acquired through the procedure in \cite{AGP05}. Note that when computing the numerics above, we take into account the factors $\sigma_i$ as introduced in Section~\ref{FTQC1} except for $\sigma_6$ which we keep as a variable. We will obtain level-1 bounds for other exRecs and the corresponding MPM will be left in Appendix \ref{a_others}. Thus for $\varepsilon_{1,\text{joint}}^{(1)}$ we have,
\begin{align*}
    255.3\epsilon^2(1-\epsilon)^{\gamma_1-2}&\leq\varepsilon_{1,\text{joint}}^{(1)} \leq 255.3\epsilon^2+21381.0\epsilon^3\\
    244.6\epsilon^2&\leq \varepsilon_{1,\text{joint}}^{(1)}\leq 321.8\epsilon^2
\end{align*}
The bounds are obtained by applying Prop \ref{prop_bound}. However, for the lower bound, since the CNOT-exRec is the largest exRec in the simulating circuit, $\epsilon\leq1/D_5^{(1)}$, so the coefficient is obtained by $A-B/D_5^{(1)}$ instead. Since there are 3 ancilla verifications here, we obtain,
\begin{align*}
   244.6\epsilon^2&\leq \varepsilon_{1}^{(1)}\leq (1-C\epsilon)^{-3}\cdot321.8\epsilon^2\\
A_1^{(1)}\epsilon^2&\leq\varepsilon_1^{(1)}\leq D_1^{(1)}\epsilon^2
\end{align*}
where $A_1^{(1)}=244.6,D_1^{(1)}=327.0$. Again we denote the lower and upper bounds by $\mu_1^{(1)}$ and $\nu_1^{(1)}$. For $\varepsilon_2^{(1)}$, note that $|\overline{+}\rangle$-prep-exRec can be simply obtained from $|\overline{0}\rangle$-prep-exRec by swapping $|+\rangle$ and $|0\rangle$, meas-$X$ and meas-$Z$, therefore it will have the same logical error rate, except the distribution of $X$- and $Z$- logical error rates are now swapped. Similarly for the destructive measurements $Loc_4$ and $Loc_5$, we have
\begin{align*}
16.6\epsilon^2&\leq\varepsilon_{4}^{(1)}\leq(1-C\epsilon)^{-2}\cdot 78.8\epsilon^2\\
16.6\epsilon^2&\leq\varepsilon_{4}^{(1)}\leq79.6\epsilon^2
\end{align*}
where the lower and upper bounds are denoted by $\mu_4^{(1)}$ and $\nu_4^{(1)}$ respectively.\\

Now we hope to generalize to higher level $k$s. Since we pursue tight upper and lower bounds here, we incorporate the error probabilities of different exRecs in the computation. This also adds extra complexity when $k>1$ because, at a higher level, there is no guarantee that $\varepsilon_i^{(k)}=\sigma_i\varepsilon_5^{(k)}$ still holds. We provide a workaround for this issue. Let $A_t^{(k)}$ and $D_t^{(k)}$ denote the lower and upper bounds of second-order(in $\varepsilon_5^{(k)}$) coefficients for $Loc_t$ on $k^{\text{th}}$-level encoding, $\sigma_{L,t}^{(k)}$ and $\sigma_{U,t}^{(k)}$ denote the lower and upper bound for $\sigma_i^{(k)}$ where $\varepsilon_i^{(k)}=\sigma_i^{(k)}\varepsilon_5^{(k)}$. From the bounds and expressions above, we can generalize to the following system of equations.
\begin{align*}
    B_t^{(k+1)}=&G_2\left( \sum_{1\leq i,j\leq5}\alpha_t(i,j)\sigma_{U,i}^{(k)}\sigma_{U,j}^{(k)},\overline{F}(\mathbf{n}_t,\underline{
\mathbf{\sigma}}_L^{(k)},\underline{
\mathbf{\sigma}}_U^{(k)},\alpha_t)\right)\\
    D_t^{(k+1)}=& B_t^{(k+1)}\bigg(1-C\epsilon_0^{(k)}\bigg)^{-s}\\
    \epsilon_0^{(k+1)}&=\min_{i\leq (k+1) } 1/D_5^{(i)}\\
    A_t^{(k+1)}=&\left(\sum_{1\leq i,j\leq5}\alpha_t(i,j)\sigma_{L,i}^{(k)}\sigma_{L,j}^{(k)} \right)\left(1-n_te_0^{(k)} \right)\\
    \sigma_{L,t}^{(k+1)}&=\frac{D_t^{(k)}}{A_5^{(k)}}\\
    \sigma_{U,t}^{(k+1)}&=\frac{A_t^{(k)}}{D_5^{(k)}}\\
\end{align*}
$1\leq t\leq 7,t\neq6$ represents the fault location index, $t\neq6$ since we hope to find the local threshold here; $\overline{F}$ is almost identical to $F$ except that in the sum, every positive term is replaced by $\sigma_{U,t}^{(k)}$ and negative term is replaced by $\sigma_{L,t}^{(k)}$ to ensure $D_t^{(k+1)}$ is indeed an upper bound. $G_2$ is the upper bound function in Prop~\ref{prop_bound}; $s=3$ when $t=1,2$, $s=2$
when $t=3,4$, $s=8$ when $t=5$ and $s=4$ when $t=7$; $\alpha_t$ is the malignant pair matrix for $Loc_t$; $\epsilon_0^{(k)}$ represents the threshold value and is a non-increasing function. The equation is established by observing that the $Loc_5$-exRec is the largest exRecs in the simulation circuit. $n_t$ is the number of $(k-1)$-gadgets in a $k$-$t$-exRec. \\

The above set of equations effectively forms a discrete-variable dynamical system and fortunately, we have a good initial point $\sigma_{U,t}^{(0)}=\sigma_{L,t}^{(0)}=4/15,\space\forall 1\leq i\leq4$, $\sigma_{U,5}^{(0)}=\sigma_{L,5}^{(0)}=1, \epsilon_0^{(0)}=1/D_5^{(1)}$ and $\sigma_{U,7}^{(0)}=\sigma_{L,7}^{(0)}=4/5$. We then apply the fixed-point iteration method to find the non-trivial fixed-point ($\sigma_5=1$ will not change through iteration). 
\begin{align*}
    \sigma_{U,1}^*&=\sigma_{U,2}^*=0.302,\\
    \sigma_{U,3}^*&=\sigma_{U,4}^*=0.0643,\\
    &\sigma_{U,7}^*=0.540,\\
    \sigma_{L,1}^*&=\sigma_{L,2}^*=9.77\times10^{-2},\\
    \sigma_{L,3}^*&=\sigma_{L,4}^*=0,\\
    &\sigma_{L,7}^*=0.166.
\end{align*}
And the corresponding coefficients are
\begin{align*}
A_1^*=178.6, & \ A_4^*=0, \space A_5^*=979.7,\space A_7^*=302.8,\\
D_1^*=295.8, & \ D_4^*=63.0,\space D_5^*=1827.1,\space D_7^*=529.4,\\
\epsilon_0^*&=1/2129.4.
\end{align*}
Thus $\epsilon_0=1/2129.4\approx 4.70\times10^{-4}$ will be taken as the threshold value. The fixed point can be checked by plugging into the original equations. Through analyzing the Jacobian at the initial point, we found that the largest eigenvalue is less than 1, which indicates this is a stable fixed point. The numerical simulation of the behaviour of convergence is left in Appendix \ref{sys_eqn}. Having obtained $D_t^{(k)}$'s, we can establish the following recursive relation for $\varepsilon_5^{(k)}$
\begin{equation*}
    A_5^{(k)}\left( \varepsilon_5^{(k-1)}\right)^2\leq\varepsilon_5^{(k)}\leq D_5^{(k)}\left( \varepsilon_5^{(k-1)}\right)^2.
\end{equation*}
\section{\texorpdfstring{Bounds on $\mathbb{P}(\text{exRec is bad})$}{Bounds on P(exRec is bad)}}
\section{Bound on third-order terms}\label{third-order}
In this section, we provide proof for Lemma \ref{third_order_thm}. We would convert this to a graph theory problem. To prove the theorem, we will first consider the case where there are two types of locations and then generalize to $n$ locations. Let $G=(V,E)$ denote a graph, $G_{s}=(V_s,E_s)$ and $G_{t}=(V_t,E_t)$ denote two random subgraphs of $G$ for two types of locations $Loc_s$ and $Loc_t$ respectively, where $V_s,V_t$ are the locations and we have $V=V_s\cup V_t$. Let $n_s=|V(G_s)|$ and $n_t=|V(G_t)|$, $E_s,E_t$ are determined by adjacency matrices $A\in\{0,1\}^{n_s\times n_s}$ and $B\in\{0,1\}^{n_t\times n_t}$. We further define an induced subgraph with adjacency matrix $C$ where vertices of this subgraph are $V_s\cup V_t$ and the edges are $E\setminus(E_s\cup E_t)$. Now each entry in $A,B,C$ is a Bernoulli variable such that 
\begin{equation*}
    A_{ij}=\begin{cases}0,&\space\text{with probability $a_{ij}$ if the $i$-th and $j$-th}\\
    & \text{ $Loc_s$ form a malignant pair }\\1,&\space\text{ otherwise} \end{cases}
\end{equation*}
similar for $B$ and $C$ with probability matrices $b$ and $c$ except for $C$ we have $C_{ij}=0$ when the $i$-th $Loc_s$ and the $j$-th $Loc_t$ form a malignant pair. Now we wish to obtain the number of $K_3$(triangle) in $G$. This is equivalent to the number we hope to obtain because if $G$ is a complete graph, then the number of $K_3$ is $\binom{n_s+n_t}{3}$, all possible combinations. Now if $\{v_i,v_j\}$ is a malignant pair and we exclude all triples containing $\{v_i,v_j\}$, this is identical to finding number of triangles in $G'=(V,E\setminus\{v_i,v_j\})$. And from the definition of $\alpha$ we also have $\alpha_{ss}=\sum_{ij\in\binom{V_s}{2}}a_{ij}$, $\alpha_{tt}=\sum_{ij\in\binom{V_t}{2}}b_{ij}$ and $\alpha_{st}=\sum_{i\in V_s,j\in V_t}c_{ij}$. Since $G$ is a random graph, we compute the expectation to obtain a fair estimate. To this end, we recall Hölder's inequality for expectations, for non-negative random variables $X,Y$, for $p$ and $q$ satisfying $\frac{1}{p}+\frac{1}{q}=1$,
\begin{equation}
    \mathbb{E}(XY)\leq (\mathbb{E}(X^p))^{\frac{1}{p}}(\mathbb{E}(Y^q))^{\frac{1}{q}}
    \label{holder}
\end{equation}
Based on this we can immediately generalize to the following proposition,
\begin{proposition}
    Suppose $X,Y,Z$ are three nonnegative random variables, then
    \begin{equation}
        \mathbb{E}(XYZ)\leq \sqrt[3]{\mathbb{E}(X^3)\mathbb{E}(Y^3)\mathbb{E}(Z^3)}
        \label{3_var}
    \end{equation}
\end{proposition}
\begin{proof}
    We first treat $S=XY$ and invoke Eqn \ref{holder} for $p=\frac{3}{2}$ and $q=3$ to obtain $\mathbb{E}(XYZ)\leq(\mathbb{E}(XY)^{3/2})^{2/3}\mathbb{E}(Z^3)^{\frac{1}{3}}=(\mathbb{E}(X^{3/2}Y^{3/2}))^{2/3}\mathbb{E}(Z^3)^{\frac{1}{3}}$. Now we apply Eqn \ref{holder} again on the first term on RHS with $p=q=2$ to show the result.
\end{proof}
Having this we can show the lemma below,
\begin{lemma}
    Let $G=(V,E)$ be constructed from $G_1$ and $G_2$ as described above, let $T$ be the random variable of the number of triangles in $G$ and $\alpha$ be the malignant pair matrix, then
    \begin{equation}
        \mathbb{E}(T)\leq f_{ss}+f_{tt}+f_{st}
    \end{equation}
    where $f_{ss}=\binom{n_s}{3}\sigma_s^3-\frac{1}{3}(n_s-2)\sigma_s\alpha_{ss}$, $f_{tt}=\binom{n_t}{3}\sigma_t^3-\frac{1}{3}(n_t-2)\sigma_t\alpha_{tt}$ and $f_{st}=\binom{n_s}{2}\cdot n_t\cdot\sigma_s^2\sigma_t+\binom{n_t}{2}\cdot n_1\cdot\sigma_s\sigma_t^2-\frac{1}{3}(n_s-1)\sigma_t\alpha_{st}-\frac{1}{3}(n_t-1)\sigma_s\alpha_{st}$
\end{lemma}
\begin{proof}
    Let's denote the three vertices of the triangle by $v_i,v_j$ and $v_k$. We consider different cases according to where the vertices lie in:
    \begin{enumerate}
        \item We first consider the case where they are all in $G_s$, so they form a triangle if and only if $A_{ij}A_{jk}A_{ki}=1$. We shall also take into account the various error rates $\epsilon_i$.
        Therefore to compute the expectation,
        \begin{align*}
\mathbb{E}(T_{G_s})&=\sigma_s^3\mathbb{E}\left[\sum_{\{i,j,k\}\in\binom{V_s}{3}}A_{ij}A_{jk}A_{ki}\right]\\
    &=\sigma_s^3\sum_{\{i,j,k\}\in\binom{V_s}{3}}\mathbb{E}[A_{ij}A_{jk}A_{ki}]\\
        &\leq\sigma_s^3\sum_{\{i,j,k\}\in\binom{V_s}{3}} \sqrt[3]{\mathbb{E}(A_{ij}^3)\mathbb{E}(A_{jk}^3)\mathbb{E}(A_{ki}^3)}
        \end{align*}
    where the inequality follows from \ref{3_var}. Since $A_{ij}$'s are Bernoulli random variables, $\mathbb{E}(A_{ij}^m)=\mathbb{E}(A_{ij})$  $\forall m\in\mathbb{N}$, thus
    \begin{align*}
        \mathbb{E}(T_{G_s})&\leq \sigma_s^3 \sum_{\{i,j,k\}\in\binom{V_s}{3}}\sqrt[3]{\mathbb{E}(A_{ij})\mathbb{E}(A_{jk})\mathbb{E}(A_{ki})}\\
        &=\sigma_s^3\sum_{\{i,j,k\}\in\binom{V_s}{3}}\sqrt[3]{(1-a_{ij})(1-a_{jk})(1-a_{ki})}\\
        &\leq\sigma_s^3\sum_{\{i,j,k\}\in\binom{V_s}{3}}\frac{3-a_{ij}-a_{jk}-a_{ki}}{3}\\
        &=\binom{n_s}{3}\sigma_s^3-\frac{(n_s-2)\sigma_s^3}{3}\sum_{ij\in\binom{V_s}{2}}a_{ij}\\
        &=\binom{n_s}{3}\sigma_s^3-\frac{1}{3}(n_s-2)\sigma_s\alpha_{ss}
    \end{align*}
    where the second inequality follows by AM-GM inequality on each term and we only have $\sigma_s$ of order 1 in the final term because in the simulation of $\alpha_{ss}$ we would have already included $\sigma_s^2$.
    \item The same procedure holds when $v_i,v_j,v_k\in V_t$ and we have 
    \begin{equation*}
        \mathbb{E}(T_{G_t})=\binom{n_t}{3}\sigma_t^3-\frac{1}{3}(n_t-2)\sigma_t\alpha_{tt}
    \end{equation*}
    \item Now we consider the case where two vertices of the triangle are in one part while the third one is in the other. We discuss two separate cases. The first one is when $v_i,v_k\in V_1$ and $v_j\in V_2$, so following a similar procedure, the expected number of triangles will be upper bounded by
    \begin{align*}
        &\sigma_s^2\sigma_t\sum_{\{i,k\}\in V_s,\space j\in V_t}\sqrt[3]{\mathbb{E}(C_{ij})\mathbb{E}(C_{jk})\mathbb{E}(A_{ik})}\\
        =&\sigma_s^2\sigma_t\sum_{\{i,k\}\in V_s,\space j\in V_t}\sqrt[3]{(1-c_{ij})(1-c_{jk})(1-a_{ik})}\\
        \leq&\binom{n_s}{2}\cdot n_t\cdot\sigma_s^2\sigma_t-\sigma_s^2\sigma_t\sum_{\{i,k\}\in V_s,\space j\in V_t}\frac{c_{ij}+c_{jk}+a_{ik}}{3}
    \end{align*}
    Similarly if $v_i\in V_s$ and $v_j,v_k\in V_t$ we have an upper bound $\binom{n_t}{2}\cdot n_s\cdot \sigma_s\sigma_t^2-\sigma_s\sigma_t^2\sum_{i\in V_s,\{j,k\}\in V_t}\frac{c_{ij}+c_{ik}+b_{jk}}{3}$. Now let $T'$ denote the random variable of the number of triangles across $G_s$ and $G_t$, 
    we can add up these two upper bounds to obtain
    \begin{align*}
        \mathbb{E}(T')\leq&\binom{n_s}{2}\cdot n_t\cdot\sigma_s^2\sigma_t-\sigma_s^2\sigma_t\sum_{\{i,k\}\in V_s,\space j\in V_t}\frac{c_{ij}+c_{jk}+a_{ik}}{3}\\
        &+\binom{n_t}{2}\cdot n_s\cdot\sigma_s\sigma_t^2-\sigma_s\sigma_t^2\sum_{i\in V_s,\{j,k\}\in V_t}\frac{c_{ij}+c_{ik}+b_{jk}}{3}\\
        \leq&\binom{n_s}{2}\cdot n_t\cdot\sigma_s^2\sigma_t+\binom{n_t}{2}\cdot n_s\cdot\sigma_s\sigma_t^2\\
        &-\sigma_s^2\sigma_t\sum_{\{i,k\}\in V_s,\space j\in V_t}\frac{c_{ij}+c_{jk}}{3}-\sigma_s\sigma_t^2\sum_{i\in V_s,\{j,k\}\in V_t}\frac{c_{ij}+c_{ik}}{3}\\
        =&\binom{n_s}{2}\cdot n_t\cdot\sigma_s^2\sigma_t+\binom{n_t}{2}\cdot n_s\cdot\sigma_s\sigma_t^2\\
        &-\frac{1}{3}(n_s-1)\sigma_s\alpha_{st}-\frac{1}{3}(n_t-1)\sigma_t\alpha_{st}
    \end{align*}
    \end{enumerate}
    Having analyzed the different cases, the desired bound is obtained by adding them together.
\end{proof}
Applying the lemma inductively would give the desired result.\\
\section{Conversion to second-order bounds}
The following proposition is a modification of a derivation in \cite{AGP05}.
\begin{proposition}
    Let $\epsilon_{(k)}=\mathbb{P}(k\text{-exRec not well-behaved})$, if 
    \begin{equation*}
    A_1(\epsilon^{(k-1)})^2-B_1(\epsilon^{(k-1)})^3\leq\epsilon^{(k)}\leq A_2(\epsilon^{(k-1)})^2+B_2(\epsilon^{(k-1)})^3
    \end{equation*}
then we can bound $\epsilon^{(k)}$ with terms second order in $\epsilon^{(k-1)}$
\begin{equation*}
        A_1'(\epsilon^{(k-1)})^2\leq\epsilon^{(k)}\leq A_2'(\epsilon^{(k-1)})^2
\end{equation*}
with $A_2'= \frac{1}{2}A_2(1+\sqrt{1+4B_2/A_2^2})$ and $A_1'=A_1-B_1/A_2'$, which holds for $\epsilon^{(k-1)}\leq 1/A_2'$.
\label{prop_bound}
\end{proposition}

\section{System of equations for the error bounds}\label{sys_eqn}
We will plot the convergence behaviour of $A_5^{(k)}$ and $D_5^{(k)}$ over 10 iterations
\begin{figure}[H]
    \centering
    \includegraphics[width=0.6\linewidth]{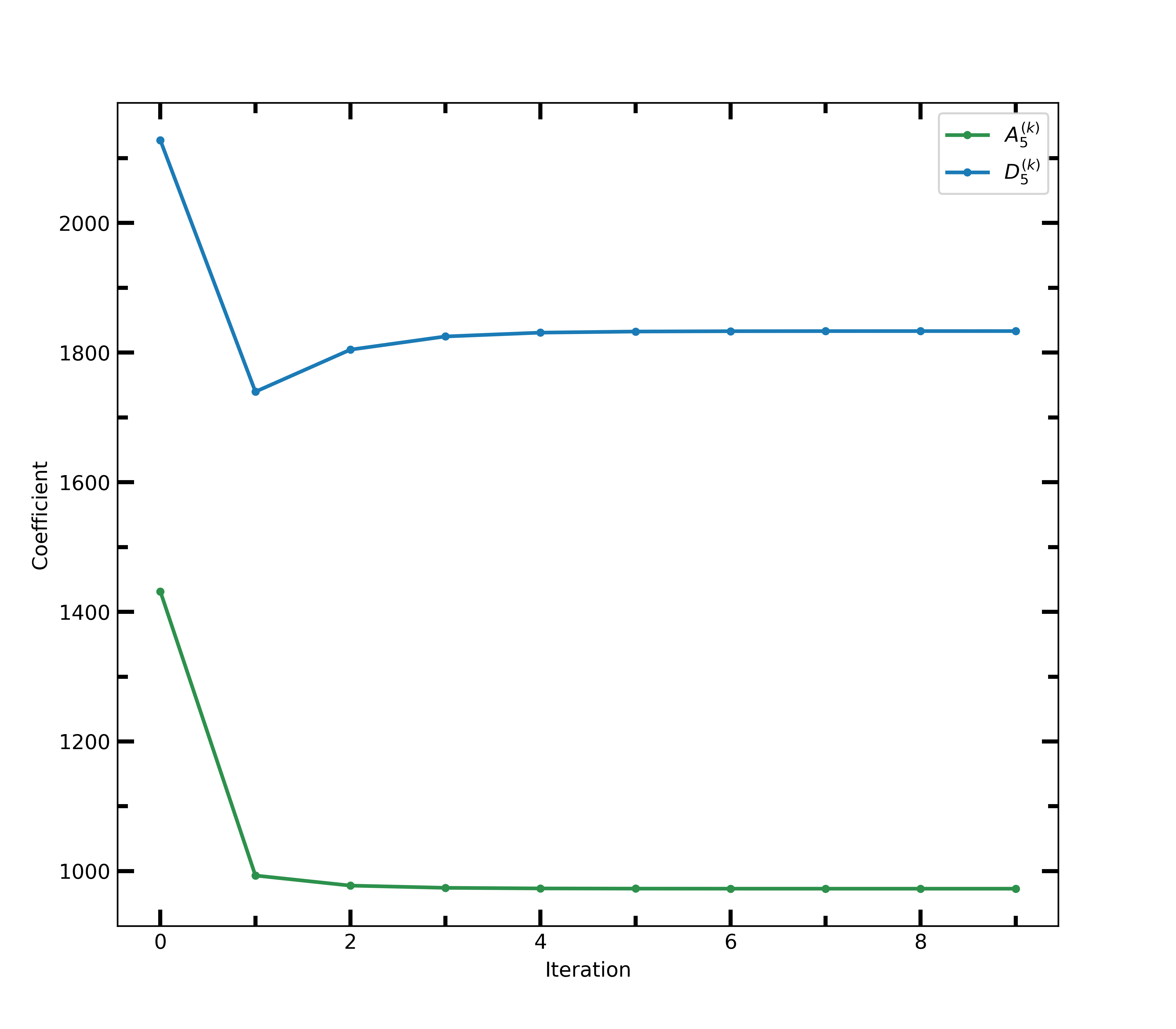}
    \label{Coefficient}
\end{figure}
We can see that the sequences converge just after 5 iterations. Also we plot the behaviour of $\sigma_{U,1}^{(k)}$ and $\sigma_{U,3}^{(k)}$ 
\begin{figure}[H]
    \centering
    \includegraphics[width=0.6\linewidth]{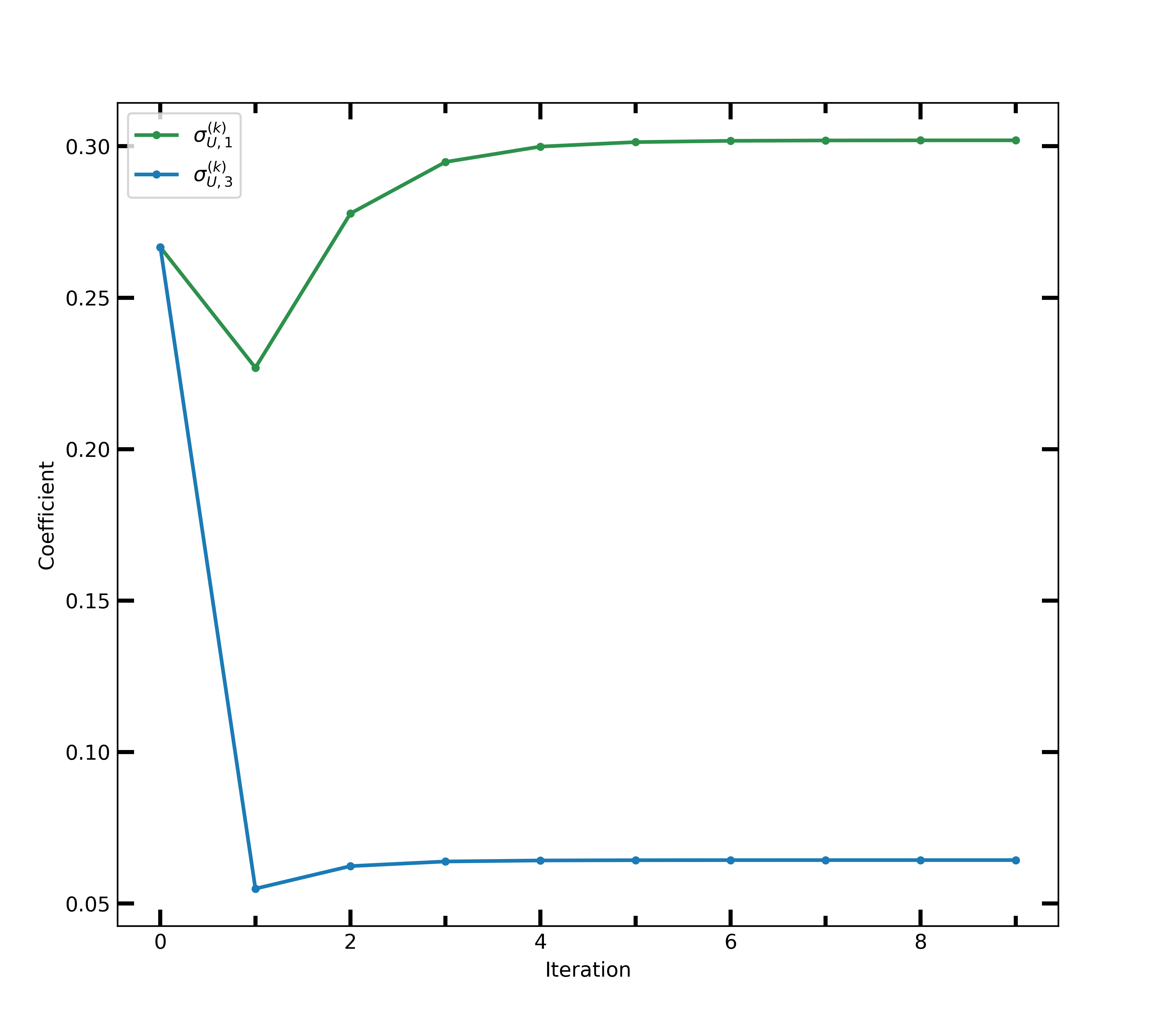}
    \label{sigma}
\end{figure}
$\forall k\geq1$, 
\begin{itemize}
    \item When $\epsilon_0^{(k+1)}=\min_{i\leq (k+1) } \left\{1/D_5^{(i)}\right\}$, which is used when only local exRecs are involved. $\sigma_{U,1}^{\max}=\sigma_{U,2}^{\max}=0.302,\space\sigma_{U,3}^{\max}=\sigma_{U,4}^{\max}=0.0643,\space\sigma_{U,7}^{\max}=0.540$ and $\sigma_{L,1}^{\min}=\sigma_{L,2}^{\min}=9.77\times10^{-2},\sigma_{U,3}^{\max}=\sigma_{L,4}^{\max}=0,\sigma_{L,7}^{\min}=0.166$. 
    \item When $\epsilon_0^{(k+1)}=\min_{i\leq (k+1) } \left\{1/D_5^{(i)}, 1/D_6^{(i)}\right\}$, that is in \textit{Direct Encoding}, when $\epsilon_6$ is involved. The initial point of $\sigma_6=2.09$ and $\epsilon\leq\epsilon_0'=2.25\times10^{-4}$. Then $\sigma_{U,1}^{\max}=\sigma_{U,2}^{\max}=0.269,\space\sigma_{U,3}^{\max}=\sigma_{U,4}^{\max}=0.0587,\space\sigma_{U,6}=1.92,\space\sigma_{U,7}^{\max}=0.469$ and $\sigma_{L,1}^{\min}=\sigma_{L,2}^{\min}=0.105,\sigma_{U,3}^{\max}=\sigma_{L,4}^{\max}=0,\sigma_{L,6}^{\min}=0.538,\space\sigma_{L,7}^{\min}=0.182$.
\end{itemize}

\section{Tightness of bounds}\label{tight}
In this section, we compare the theoretical bounds for the CNOT-exRec derived in the main text to the error rates obtained through numerical simulations for the case of $k=1$. Additionally, we plot the upper bound calculated using the procedure outlined in \cite{AGP05}, but applied to our error correction gadget.
For the simulation the results are obtained for $10^7$ runs, with $\epsilon\leq\epsilon_0$. 
\begin{figure}[H]
    \centering
    \includegraphics[width=0.6\linewidth]{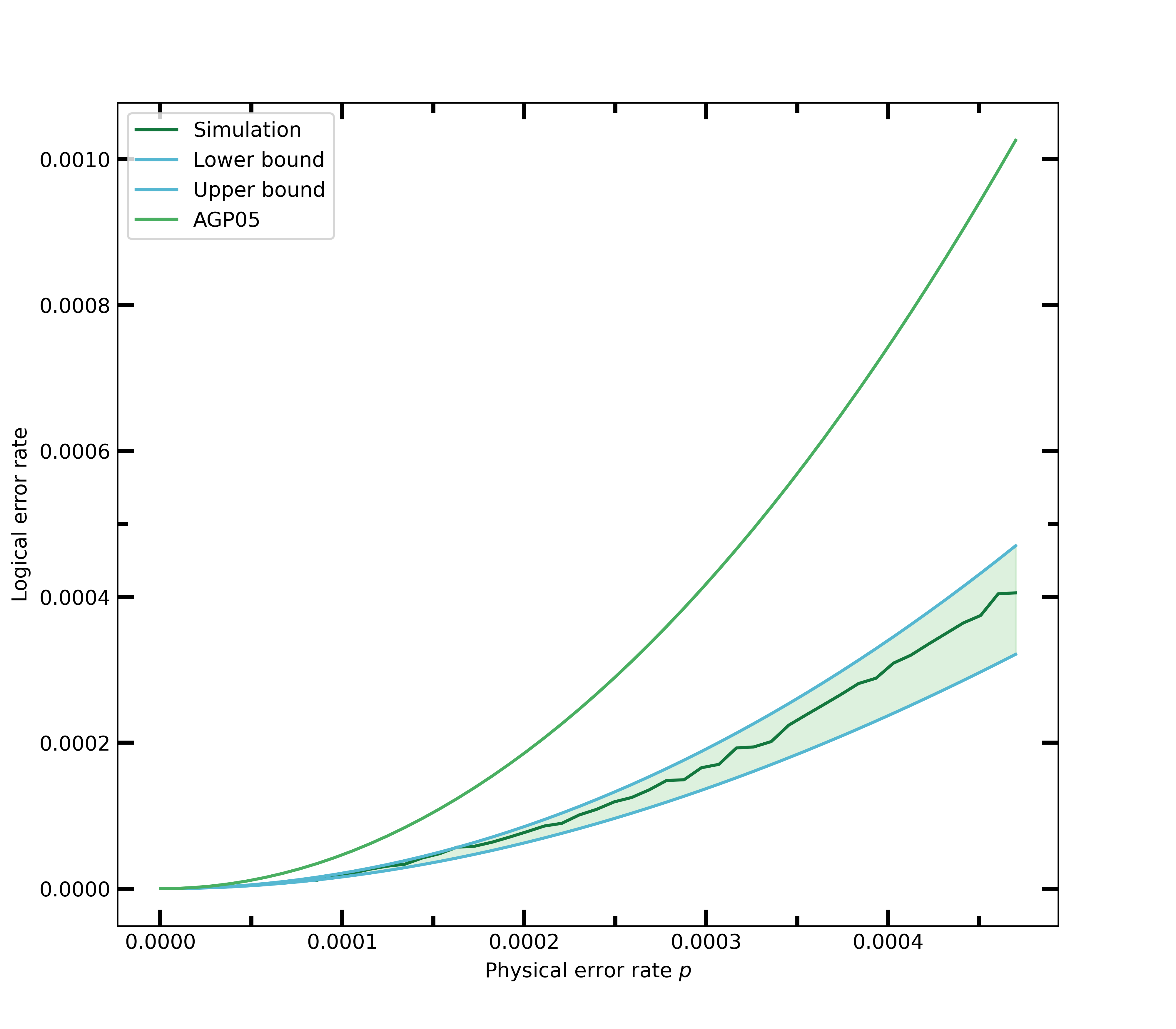}
    \caption{Comparison of bounds on logical error rate of CNOT-exRec}
    \label{fig:enter-label}
    \label{tightness}
\end{figure}
By plotting the curves we are able to examine the tightness of bounds. We observe that our bounds at $k=1$ are decent estimates of the true value. Notably, at $\epsilon=\epsilon_0$, the upper bound shows a $54\%$ improvement over the original bound, thereby affirming the credibility of our bounds even as we extend the generalization to higher concatenation levels. This substantiates the significance of our subsequent comparison of the bounds of the two methods. 
\section{Interface+EPP error analysis}
\subsection{Proof of Interface }
\begin{proof}[Proof of Lem~\ref{EpEnc}]
    To compute $\mathbb{P}(\text{Enc}_l\text{ bad})$, we note that in the interface there are also ancilla state verifications, let $H$ denote the instances that $|\Omega\rangle^{(k)}$ accepted $\forall k\leq l-1$ given that $|\overline{0}\rangle^{(k)}$'s are accepted, what we really want is actually $\mathbb{P}(\text{Enc}_l\text{ bad}|H)$. To compute this conditional probability we will compute $\mathbb{P}(\text{Enc}_l\text{ bad}\wedge H)$ first.\\
    An interface consists of the teleportation(43 locations) and the EC gadget(68 locations). We shall proceed with Enc$_{0\rightarrow1}$ first. By simulation, the teleportation circuit has on average $[0,1,1,1,2,0,0]$ locations(e.g. the second entry means 1 $Loc_2$-fault location) that will cause Enc$_{0\rightarrow1}$ to have a bad encoded state. If we encode to level-$k$, since we have $\text{Enc}_k=\text{Enc}_{(k-1)\rightarrow k}\circ\dots\circ\text{Enc}_{1\rightarrow2}\circ\text{Enc}_{0\rightarrow1}$ and by applying the union bound, the sum of the first-order(in $\varepsilon_5^{(k)}$) terms will be upper bounded by
    \begin{align*}
        \leq&\sum_{k=0}^{l-1}\left(\varepsilon_2^{(k)}+\varepsilon_3^{(k)}+\varepsilon_4^{(k)}+2\varepsilon_5^{(k)} \right)\\
        =&\sum_{k=0}^{l-1}\left(\sigma_2^{(k)}+2\sigma_3^{(k)}+2 \right)\varepsilon_5^{(k)}\\
    \end{align*}
    From Appendix \ref{sys_eqn} we know that $\sigma_i^{(k)}$ is bounded above and we numerically obtain the values. Therefore we arrive at the upper bound,
    \begin{equation*}
        2.8\epsilon+2.43\sum_{k=1}^{\infty}\nu_5^{(k)}
    \end{equation*}
    where we replaced the finite sum with an infinite sum to encompass all $k\in\mathbb{N}$. The convergence can be seen by noticing that $\nu_5^{(k)}\leq\epsilon_0\left(\epsilon/\epsilon_0 \right)^{2^k}$, $\forall k\geq2$, the infinite sum is upper bounded by a convergent geometric series and thus it's also convergent when $\epsilon<\epsilon_0$ \footnote{This infinite series is also known as lacunary series, which has no simply closed-form expression, an expression derived from Fourier transform can be found \cite{stackexchange}.}. For the sake of simplicity and to obtain a tight bound while maintaining the current threshold value, we resort to numerics. Let us denote this converging limit as $\rho_1(\epsilon)$, 
 For example, if $\epsilon=\epsilon_0$, $\rho_1(\epsilon)\approx1.27\times10^{-3}$. In summary, for first-order terms, we have the upper bound $2.8\epsilon+2.43\rho_1(\epsilon)$.\\
 
    If two locations have faults, we again have the MPM $\alpha_{in}$ as in Appendix \ref{a_others}. Similar to before, we have the following upper bound
    \begin{equation*}
        \sum_{k=0}^{l-1}\sum_{i,j}\alpha_{in}(i,j)\nu_i^{(k)}\nu_j^{(k)}+\overline{F}(\mathbf{n}_{in},\underline{\sigma}_L^{(k)},\underline{\sigma}_U^{(k)},\alpha_{in})\left( \nu_5^{(k)}\right)^3
    \end{equation*}
    where $\mathbf{n}_{in}=[14,12,9,11,58,0,7]$. In $\overline{F}$ we may replace $\underline{\sigma}_L^{(k)}$ and $\underline{\sigma}_U^{(k)}$ by $\underline{\sigma}_U^{\sup}$ and $\underline{\sigma}_L^{\inf}$. In this case, $\overline{F}$ evaluates to a constant 36372.3 $\forall k\geq1$ and 44437.8 when $k=0$. For the first term, we can also simplify, 
    \begin{align*}
        &\sum_{k=0}^{l-1}\sum_{i,j}\alpha_{in}(i,j)\nu_i^{(k)}\nu_j^{(k)}\\
        \leq&557.8\epsilon^2+\sum_{i,j}\alpha_{in}(i,j)\sigma_{U,i}^{\max}\sigma_{U,j}^{\max}\sum_{k=1}^{\infty}\left( \nu_5^{(k)}\right)^2\\
        =& 557.8\epsilon^2+451.6\sum_{k=1}^{\infty}\left( \nu_5^{(k)}\right)^2 
    \end{align*}
    Thus by summing up the results above and applying Prop~\ref{prop_bound} we obtain an upper bound for the second-order terms
    \begin{align*}
        \leq628.5\epsilon^2+451.6\sum_{k=1}^{\infty}\left( \nu_5^{(k)}\right)^2+36372.3\sum_{k=1}^{\infty}\left( \nu_5^{(k)}\right)^3
    \end{align*} 
    By the simple observation that $\sum_{k=1}^{\infty}a_i^k\leq\left(\sum_{k=1}^{\infty}a_i \right)^k$ for $a_i\geq0$ $\forall i$, we obtain,
    \begin{align*}
        &\mathbb{P}(\text{Enc}_l\text{ bad}\wedge H)\\        \leq&2.8\epsilon+2.43\rho_1(\epsilon)+628.5\epsilon^2+521.3\rho_1(\epsilon)^2
    \end{align*}
    Now we compute a lower bound on $\mathbb{P}(H)$, as the ancillary states $|\Omega\rangle^{(k)}$ are independent.
    \begin{align*}
        \mathbb{P}(H)&=\sum_{k=0}^{l-1}\mathbb{P}\left(|\Omega^{(k)}\rangle\text{ accepted}||\overline{0}^{(k)}\rangle\text{ accepted}\right)\\
        &=\prod_{k=0}^{l-1}\frac{\mathbb{P}\left(|\Omega^{(k)}\rangle\text{ accepted}\wedge|\overline{0}^{(k)}\rangle\text{ accepted}\right)}{\mathbb{P}\left(|\overline{0}^{(k)}\rangle\text{ accepted}\right)}\\
        &\geq \prod_{k=0}^{l-1}\mathbb{P}\left(|\Omega^{(k)}\rangle\text{ accepted}\wedge|\overline{0}^{(k)}\rangle\text{ accepted}\right)\\
        &\geq \prod_{k=0}^{l-1}\left(1-C_1\nu_5^{(k)} \right)
    \end{align*}
    where $C_1=18.1$. We will now need to upper bound $\mathbb{P}(H)^{-1}$,
    \begin{align*}
        &\prod_{k=0}^{l-1}\left(1-C_1\nu_5^{(k)} \right)^{-1}\\
        \leq&\prod_{k=0}^{\infty}\left(1-C_1\nu_5^{(k)} \right)^{-1}\\
        =&\exp\left(-\sum_{k=0}^{\infty}\log\left(1-C_1\nu_5^{(k)} \right)\right)\\
        \leq& \exp\left(\sum_{k=0}^{\infty}\frac{C_1\nu_5^{(k)}}{\sqrt{1-C_1\nu_5^{(k)}}} \right)\\
        \leq& \exp\left(\sum_{k=0}^{\infty}1.0043\cdot C_1\nu_5^{(k)}\right)\\
        =&\exp\left(1.0043\cdot C_1\left(\epsilon+\rho_1(\epsilon) \right)\right)
    \end{align*}
    where the second inequality follows from the fact that $\forall x\in(0,1)$, $\log(1-x)\geq\frac{-x}{\sqrt{1-x}}$ and the third follows from that $\nu_5^{(k)}\leq\epsilon_0$ $\forall k\geq1$. Combining the above results gives us the Lemma.
\end{proof}
\section{Interface error analysis}\label{error_type}
We will first plot the upper bound of $\mathbb{P}(\text{Enc}_l \text{ bad})$ with respect to $\epsilon$ below
\begin{figure}[H]
    \centering
    \includegraphics[width=0.6\linewidth]{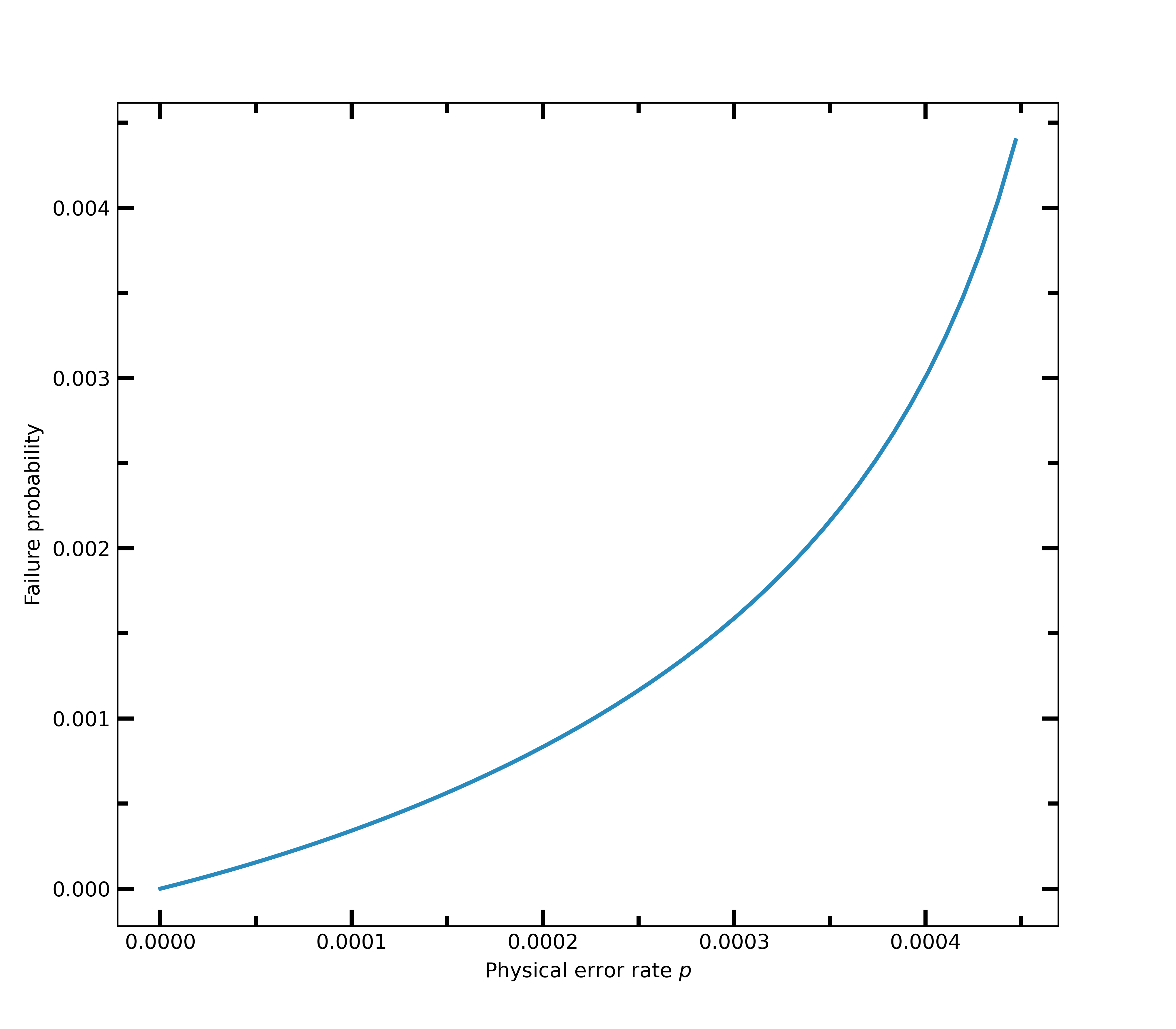}
\end{figure}
Next, we will analyze the faults and resultant logical errors of the interface in more detail. In particular, we will include errors in the incoming state $|\phi\rangle$. We will need to take into account the distinct $\epsilon_6$ error and also the different logical error types. By observing the circuit of the interface, we infer that if there is an incoming $X$ error, this results in a $\overline{Z}$ error; $Y\longrightarrow \overline{Y}$ and $Z\longrightarrow \overline{X}$. Moreover, we take into account the fact that an encoded EPR pair is stabilized by $\{\overline{XX},\overline{ZZ}\}$. We arrive at the following enumeration:
\begin{itemize}
    \item \textbf{Single fault location}\\
    Excluding $Loc_6$:
    \begingroup
        \setlength{\tabcolsep}{10pt}
        \renewcommand{\arraystretch}{1.35}
    \begin{table}[H]
    \centering
    \begin{tabular}{c|c}
    $\overline{X}$ & [0,1,1,0,0.684,0,0]    \\ \hline
    $\overline{Y}$ & [0,0,0,0,0.674,0,0]        \\ \hline
    $\overline{Z}$ & [0,0,0,1,0.656,0,0]    
    \end{tabular}
    \end{table}
    \endgroup
    To include the initial EPR pair, for the upper interface we have an extra $\overline{X}$-error from $Loc_2$ and for the lower interface we have an extra $\overline{Z}$-error from the $Loc_1$.
    For $Loc_6$, as explained before, for an EPR pair the possible errors are $XI/YI/ZI$, with $XI$ occurring with probability $I_X(q,\epsilon_0)$ and similarly for $YI/ZI$.
    \item \textbf{Malignant pairs}\\
    For the malignant pairs, we will only enumerate for $k=1$ and thus we include the respective $\sigma_i$.
    Excluding $Loc_6$:
        \begingroup
        \setlength{\tabcolsep}{10pt}
        \renewcommand{\arraystretch}{1.35}
    
    \begin{table}[H]
    \centering
    \begin{tabular}{c|c|c}
    $\overline{X}$ & $\overline{Y}$ & $\overline{Z}$ \\ \hline
    262.8 & 67.7 & 227.3
    \end{tabular}
    \end{table}
    \endgroup
    Malignant pairs involving $Loc_6$ that induce logical error on $\overline{\text{EPR}}_1$
    
    \begingroup
    \setlength{\tabcolsep}{10pt}
    \renewcommand{\arraystretch}{1.35}
    \begin{table}[H]
    \centering
\begin{tabular}{c|c|c|c|c}
$\overline{YY}$                 & $\overline{XX}$                & $\overline{ZZ}$                & $\overline{XI}$                & $\overline{IX}$                \\ \hline
0 &        0 &      0 &    31.0 &  0.21          \\ \hline
$\overline{ZI}$                 & $\overline{IZ}$                & $\overline{YI}$                & $\overline{XZ}$                & $\overline{ZX}$                \\ \hline
31.4 & 0.242 &  31.2 & 0.321 &  0.394 \\ \hline
$\overline{IY}$                 & $\overline{YX}$                & $\overline{XY}$                & $\overline{ZY}$                & $\overline{YZ}$                \\ \hline
0 &  0 & 0 &   0 &    0
\end{tabular}
\label{dist_of_error2}
\end{table}
    \endgroup

\end{itemize}
\section{Proofs for bounds on EPP failure probabilities}\label{EPP_failure}
\begin{proof}[Proof of \ref{ub_on_InEPP_A}]

We can view the \textit{Interface+EPP-A} circuit as consisting of two major components, one is the interface preparation of logical EPR pairs and the other is the EPP procedure. Below we use $\mathcal{G}_1$ to denote the instance that the preparation of 4 logical EPR pairs is bad and $\mathcal{G}_2$ to denote the cases where the EPP is bad. Thus we have
\begin{equation*}
    \mathbb{P}(\text{EPR rejected})\leq \mathbb{P}(\mathcal{G}_2)+\mathbb{P}(\mathcal{G}_1)
\end{equation*}
We would first compute the upper bound on $\mathbb{P}(\mathcal{G}_2)$. For the upper bound, since we have 6 CNOT-exRecs and 6 mmt-exRecs,
\begin{align*}
    \mathbb{P}(\mathcal{G}_2)\leq\sum_{i}|\text{Loc}_i|\varepsilon_i^{(k)}\leq 2\cdot6\cdot\left(\nu_5^{(k)}+\nu_4^{(k)} \right)
\end{align*}
which follows from the union bound and the factor $2$ follows from P.13 of \cite{AGP05}, resulting in a noisy simulation of the whole quantum circuit. 

Now we will compute an upper bound on $\mathbb{P}(\mathcal{G}_1)$, to do so, observe that the preparation part consists of 4 EPR pairs, and on each side the qubits are expanded via the interface, so 8 interfaces in total. We may now invoke Lemma \ref{EpEnc} and obtain the following
\begin{equation*}
    \mathbb{P}(\mathcal{G}_1)\leq 4\epsilon_6+8f_{in}(\epsilon)
\end{equation*}
we observe that $\mathbb{P}(\mathcal{G}_2)\leq12\epsilon$ $\forall k$, adding up two terms gives the desired upper bound.
\end{proof}
Observe that this analysis is applicable not only to concatenated Steane code, any proposal of code concatenation that has the corresponding interface with a logical error rate upper bound independent of $k$ will have EPP failure probability that is also independent of $k$ (given the EPP part can be performed fault-tolerantly in the logical space).

\begin{proof}[Proof of Corollary] \ref{lb_on_InEPP_A}
We would now derive a crude lower bound on $\mathbb{P}(\text{EPR rejected})$ and we would consider terms up to $\mathcal{O}(\varepsilon^3)$, so this bound will hold well for $\epsilon<\epsilon_0$ considered in this work. We may again use the Bonferroni inequality (or inclusion-exclusion principle) and thus,
\begin{align*}
    \mathbb{P}(\text{EPR rejected}) \geq \mathbb{P}(\mathcal{G}_1)+\mathbb{P}(\mathcal{G}_1)-\mathbb{P}(\mathcal{G}_1\wedge\mathcal{G}_2)
\end{align*}
For the first term above, based on the observation that if certain locations are faulty then they will surely cause a logical error, for $\mathcal{O}(\epsilon)$ terms, this will cause a logical error if it is in the initial $|\Phi\rangle$ or Enc$_{0\rightarrow1}$ while other locations are error-free. Note that the error probability for exRecs in $\text{Enc}_{l\rightarrow l+1}$ is $\varepsilon_i^{(l)}$, so for $\mathcal{O}(\epsilon)$ terms we have the lower bound (in the following $\gamma_7=111$ being the number of locations in $\text{Enc}_{0\rightarrow1}$,)
\begin{align*}
    8&\cdot2.8\cdot\epsilon(1-\epsilon)^{-1}(1-\epsilon_6)^4\prod_{i=0}^{k-1}\prod_{t=1}^{7}\left(1-\nu_t^{(i)}\right)^{8n_{in,t}}\\
    &+4\cdot \epsilon_6(1-\epsilon_6)^3\prod_{i=0}^{k-1}\prod_{t=1}^{7}\left(1-\nu_t^{(i)}\right)^{8n_{in,t}}
    \tag{$\dagger$}
    \label{first_order}
\end{align*} 
For the product term, we use the Weierstrass product inequality and obtain
\begin{equation*}
    \prod_{i=0}^{k-1}\prod_{t=1}^{7}(1-\nu_t^{(i)})^{8n_{in,t}}\geq\left(1-8\sum_{i=0}^{1}\sum_{t=1}^7n_{in,t}\nu_t^{(i)} \right)
\end{equation*}
for the early terms, we apply $(1+x)^r\geq1+rx$. Combining these with the specific values of $\nu_t^{(i)}$, we can simplify Eqn \eqref{first_order} as, to $\mathcal{O}(\epsilon^3)$,
\begin{align*}
    &22.4\epsilon(1-605.9\epsilon-4\epsilon_6+2423.6\epsilon\epsilon_6-1.10\times10^6\epsilon^2)\\
    +&4\epsilon_6(1-3\epsilon_6-606.9\epsilon-1.10\times10^6\epsilon^2+1820.7\epsilon\epsilon_6)
\end{align*}
Similarly, for two-fault terms, there are a number of cases to be considered, 
\begin{enumerate}
    \item Logical errors resulting from malignant pairs in $\text{Enc}_{0\rightarrow1}$, we have the lower bound,
    \begin{align*}
    &8\sum_{ij}\alpha_{in}(i,j)\sum_{s=0}^{k-1}\mu_i^{(s)}\mu_j^{(s)}(1-\nu_5^{(s)})^{-2}\dots\\
    \dots&\prod_{i=0}^{k-1}\prod_{t=1}^{7}(1-\nu_t^{(i)})^{8n_{in,t}}(1-\epsilon_6)^4\\
    \end{align*}
    \item For $k\geq2$, logical errors resulting from one faulty-exRec in $\text{Enc}_{1\rightarrow2}$. Note that for exRecs in $\text{Enc}_{l\rightarrow(l+1)}$ $l\geq3$ to be faulty, $\mathcal{O}(\epsilon^4)$ would be needed, so are not considered here. Hence for this case, we have the lower bound
    \begin{align*}
    &8\cdot2.8\cdot\mu_5^{(1)}(1-\nu_5^{(1)})^{-1}\dots\\
    \dots&\prod_{i=0}^{k-1}\prod_{t=1}^{7}(1-\nu_t^{(i)})^{8n_{in,t}}(1-\epsilon_6)^4
    \end{align*}
    \item If we have one fault in the initial EPR and one in one of the $\text{Enc}_{0\rightarrow1}$. From Table~\ref{dist_of_error1} we have
    \begin{align*}
    &4\cdot94.7\cdot\epsilon_6\epsilon(1-\epsilon)^{-1}(1-\epsilon_6)^3\prod_{i=0}^{k-1}\prod_{t=1}^{7}(1-\nu_t^{(i)})^{8n_{in,t}}
    \end{align*}
    \item We further have cases where there are two faults in two interfaces respectively. Apart from naively computing all combinations, there is a subtlety here, being that if two logical EPR pairs induce the same type of logical error, after EPP they end up canceling each other, for example, if on Alice's side, the first and second logical EPRs have $X$-error, the first logical EPR will still be accepted, so we should not include such cases. To this end, we will refer to Table \ref{dist_of_error1} and enumerate a lower bound of 192.8 such cases. So such cases occur with a probability
    \begin{equation*}
    192.8\cdot\epsilon^2(1-\epsilon)^{-2}(1-\epsilon_6)^4\prod_{i=0}^{k-1}\prod_{t=1}^{7}(1-\nu_t^{(i)})^{8n_{in,t}}
    \end{equation*}
\end{enumerate}
Summing up the above second-order cases and lower bound all the $(1-x)^r$'s, we have, to $\mathcal{O}(\epsilon^3)$
\begin{equation*}
    26768.7\epsilon^2+378.8\epsilon\epsilon_6-1.62\times10^7\epsilon^3-3.37\times10^5\epsilon^2\epsilon_6-1136.4\epsilon\epsilon_6^2
\end{equation*}
for $k\geq2$. For the next term, we note that when $k\geq2$, for $\mathbb{P}(\mathcal{G}_2))$ to cause a logical error we need at least $\mathcal{O}(\epsilon^4)$. Now for the last term, we note that for both components to have logical errors, we need at least $\mathcal{O}(\epsilon^5)$ in the circuit for $k\geq2$ because for $\mathcal{G}_2$ to occur we need $\mathcal{O}(\epsilon^4)$ and one fault in the EPR preparation. Combining all the terms above and the fact that $\epsilon\leq\epsilon_0$ we obtain the desired result.
\end{proof}
\section{Miscellaneous}\label{a_others}
\section{Distribution of logical error after CNOT-exRec}\label{dist_CNOT}
The probability distribution of malignant pairs that cause different types of logical errors for a CNOT-1 exRec (when $\epsilon_5=\epsilon_6$).
    \begingroup
    \setlength{\tabcolsep}{2pt}
    \renewcommand{\arraystretch}{1.35}
\begin{table}[H]
\centering
\begin{tabular}{c|c|c|c|c}
$\overline{YY}$                 & $\overline{XX}$                & $\overline{ZZ}$                & $\overline{XI}$                & $\overline{IX}$                \\ \hline
$7.21\times 10^{-5}$ & $9.01\times10^{-3}$ & $1.53\times10^{-2}$ & 0.161               & 0.298               \\ \hline
$\overline{ZI}$                 & $\overline{IZ}$                & $\overline{YI}$                & $\overline{XZ}$                & $\overline{ZX}$                \\ \hline
0.302                & 0.160               & $2.03\times10^{-2}$ & $5.65\times10^{-4}$ & $8.52\times10^{-3}$ \\ \hline
$\overline{IY}$                 & $\overline{YX}$                & $\overline{XY}$                & $\overline{ZY}$      l          & $\overline{YZ}$                \\ \hline
$1.97\times10^{-2}$  & $1.83\times10^{-3}$ & $1.20\times10^{-4}$ & $1.65\times10^{-3}$ & $1.44\times10^{-4}$ 
\end{tabular}
\label{dist_of_error1}
\end{table}
    \endgroup

\section{Bounds for Shor's parity measurement}\label{Shor-mmt}
Here we investigate the Shor parity measurement. We will exemplify this by examining the case where Alice performs a $\overline{XX}$ measurement on her side. Firstly we will obtain an upper bound on the rejection probability of the 14-qubit ancilla cat state. The number of single fault locations that will lead to rejection is [14, 0,0,1,8,0,0]. There are 38 locations in total with a breakdown of $\mathbf{n}_{\text{Shor}}=[14,1,7,1,15,0,0]$. Thus when 
we encode to level-$k$, we simply arrive at the upper bound
\begin{align*}
    \mathbb{P}(\text{Cat state rejected})\leq&\left(14\sigma_{U,1}^{\max}+\sigma_{U,3}^{\max}+8\right)\nu_5^{(k-1)}\\
    +&\sum_{\{i,j\}\in\binom{\mathbf{n}_{\text{Shor}}}{2}}\sigma_{U,i}^{\max}\sigma_{U,j}^{\max}\left(\nu_5^{(k-1)} \right)^2\\
    =& 11.8\nu_5^{(k-1)}+196.3\left(\nu_5^{(k-1)}\right)^2\\
    \leq&11.9\nu_5^{(k-1)}
\end{align*}
In general, two scenarios of the parity measurement would lead to the failure of the magic square game. First, if output logical errors occur during the initial two parity measurements on the data block, subsequent measurements are adversely affected. Second, faults within the measurement process itself can induce a 'parity-change', contributing to the failure. We say the Shor's measurement is bad if either of these happens. Having this observation, we compute the bounds. As this procedure was justified to be fault-tolerant, we again start by enumerating the malignant pairs. The MPM is listed in Appendix \ref{MPM}. So for level-$k$ concatenation, if we denote $\varepsilon^{(k)}_{\text{Shor,joint}}$ to be the case when all cat state ancillas are accepted but Shor's measurement is bad, following the same procedure as in Section \ref{Dir_Enc}, we arrive at the following bounds
\begin{align*}
     3.27\mu_0\left( \frac{\epsilon}{\mu_0
     }\right)^{2^k}\leq\varepsilon^{(k)}_{\text{Shor,joint}}\leq 1.27\epsilon_0\left(\frac{\epsilon}{\epsilon_0} \right)^{2^k}
\end{align*}
Hence for the conditional probability $\varepsilon^{(k)}_{\text{Shor}}$
\begin{equation*}
         3.27\mu_0\left( \frac{\epsilon}{\mu_0
     }\right)^{2^k}\leq\varepsilon^{(k)}_{\text{Shor}}\leq 1.27\epsilon_0\left(\frac{\epsilon}{\epsilon_0} \right)^{2^k}\left( 1-11.9\nu_5^{(k-1)}\right)^{-3}
\end{equation*}
We shall denote the lower and upper bounds by $\mu_{\text{Shor}}^{(k)}$ and $\nu_{\text{Shor}}^{(k)}$ respectively. 
\section{Malignant pair matrix(MPM)}\label{MPM}
\begin{enumerate}
    \item \textit{$|\overline{0}\rangle$\space/\space$|\overline{+}\rangle$-exRec}\\
    $\mathbf{n}_1=[11, 13, 9, 8, 47, 0, 0]$
    \begin{equation*}
    \alpha_1 = 
    \begin{pmatrix}
        10.0 & & & & & &\\
        0 & 11.0 & & & &&\\
        0 & 25.0 & 0 & & &&\\
        14.0 & 0 & 0 & 0 & &&\\
        70.0 & 69.9 & 81.2 & 74.4 & 172.2\\
        0&0&0&0&0&0 \\
        0&0&0&0&0&0&0
    \end{pmatrix}
    \end{equation*}
    \item \textit{$Z$-mmt-exRec}\\
    $\mathbf{n}_3=[8,8,8,15,36,0,0]$
    \begin{equation*}
        \begin{pmatrix}
            0   &&& &\\
 0     &0    && \\
 0 &    0   & 0 & & \\
 18.0 & 0 &     0  &  63.0 &\\
 0 &     0 &    0 &  42.8 &   0\\
         0&0&0&0&0&0 \\
        0&0&0&0&0&0&0
        \end{pmatrix}
    \end{equation*}
    \item \textit{EC-gadget}\\
    $\mathbf{n}_3=[8,8,8,8,36,0,0]$
    \begin{equation*}
        \begin{pmatrix}
5.0 &  &   &  &  \\
0 &   5.0 &  &   &   \\
0 & 7.0 &  0 &   &   \\
7 & 0 &  0 & 0 &  \\
44.0 & 42.0 & 59.5 & 59.0 & 130.2\\
        0&0&0&0&0&0 \\
        0&0&0&0&0&0&0
        \end{pmatrix}
    \end{equation*} 

    \item\textit{Identity gate}\\
    $\mathbf{n}_{in}=[16,16,16,16,72,0,7]$
    \begin{equation*}
    \begin{pmatrix}
    14.0 &  &  &  &  &  \\
 0 & 14.0 & &  &  &    \\
 0 & 49.0 &  42.0 &  &  & \\
 49.0  &  0 &   0 &   42.0 &  & \\
 100.7 & 100.7 &  181.4 &  181.4 & 278.6 & \\
 0 & 0 & 0 & 0 & 0 & 0 \\
 28.0 &   28.0 &  56.0 &  56.0 & 148.8 &0 &  16.4 \\
0&0&0&0&0&0&0

    \end{pmatrix}
    \end{equation*}

\item \textit{$\overline{\text{ebit}}^{(k)}$ from Direct Encoding}
\begin{equation*}
    \alpha_{\text{EPR}} = 
    \begin{pmatrix}
53 &  &  & & & \\
 0 & 53  &  &  &  &  \\
 0 &  189.96& 167.95  &   &   &      \\
 189.95 &  0  &   0  & 167.94  &  & \\
 360.36 &359.9 & 640.91& 646.04& 916.07 &   \\
 43.73 & 44.93&  89.56&  90.08& 231.49  &12.7 \\
 0&0&0&0&0&0&0
    \end{pmatrix}
\end{equation*}

    \item\textit{Interface},  $\text{Enc}_{0\rightarrow1}\circ \text{EC}$\\
    $\mathbf{n}_{in}=[14, 12, 9, 11, 58, 0, 7]$
    \begin{equation*}
    \alpha_{in} = 
    \begin{pmatrix}
        10.0 && & && \\
       5.0 & 18.0& & & &\\
       5.0&  39.0&   7.0  & &  &\\
       20.0& 15.0 &15.0 &   7.0 &  &\\
       92.6& 127.0 & 140.3& 137.9 & 281.5&\\
       0 & 0 & 0 & 0 & 0 & 0 \\  
       19.8 & 30.9 &35.0 &35.1& 121.6 & 0 &  16.4
    \end{pmatrix}
    \end{equation*}

    \item\textit{Logical EPP}  
    \begin{equation*}
    \alpha_{EPP} = 
    \begin{pmatrix}
        102 && & && \\
       0 & 34.0& & & &\\
       0&  110&   84.0  & &  &\\
       426 & 0 &0 &   420 &  &\\
       832 & 303 & 500 & 1626 & 1634&\\
       0 & 0 & 0 & 0 & 0 & 0 \\  
       0 &0 & 0 &0& 0 & 0 &  0
    \end{pmatrix}
    \end{equation*}
    
    \item\textit{Shor-$\overline{X}$-mmt (full)}\\
        $\mathbf{n}_{S_3}=[106, 67, 85,67, 375, 0,0]$\\
        \begin{equation*}
    \alpha_{\text{Shor}} = 
    \begin{pmatrix}
        249 & &   & & \\
       0 & 114 &  &  &  \\
       0 & 546 & 756 &  &   \\
       315 &  0 &  0 &  252 &   \\
       992 & 1094 & 2720 & 1492 & 3532\\
        0&0&0&0&0&0 \\
        0&0&0&0&0&0&0
    \end{pmatrix}
    \end{equation*}
    
\end{enumerate}

\end{document}